\documentclass[9pt,preprint,nocopyrightspace]{sigplanconf}
\newif\ifcrazy
\crazytrue
\crazyfalse
\newif\ifLONGVERSION
\LONGVERSIONfalse
\LONGVERSIONtrue
\newif\ifCONSTANTPATTERN
\CONSTANTPATTERNfalse
\ifLONGVERSION
\def\shortenbib{9pt}
\else
\def\shortenbib{8.37pt}
\fi
\ifcrazy
\newcommand{\appcrazy}{Section}
\else
\newcommand{\appcrazy}{Appendix}
\fi
\usepackage{macros}
\begin{document}
\title{Static and Dynamic Semantics of NoSQL Languages}
\authorinfo{V\'eronique Benzaken$^1$ \and Giuseppe Castagna$^2$ \and Kim Nguy\~{\^e}n$^1$ \and J\'er\^ome Sim\'eon$^3$\\[2mm]\normalsize
$^1${LRI, Universit\'e Paris-Sud, Orsay, France}, \quad
$^2${CNRS, PPS, Univ Paris Diderot, Sorbonne Paris Cit\'e, Paris, France}\\
$^3${IBM Watson Research, Hawthorne, NY, USA}}{}%\vspace{-.9in}}

\maketitle

\iffalse
\begin{quote}\em
A family of unimplemented computing languages is described that is
intended to span differences of application area by a unified
framework.  [...] The design of a specific language splits into two
independent parts. One is the choice of written appearances of
programs. The other is the choice of the abstract entities.
\end{quote}
\hfill P.J.\ Landin. ``The next 700 programming languages'' (1966)
\fi

\begin{abstract}
%\paragraph{Abstract.}  
  We present a calculus for processing semistructured data that spans
differences of application area among several novel query languages,
broadly categorized as ``NoSQL''. This calculus lets users define
their own operators, capturing a wider range of data processing
capabilities, whilst providing a typing precision so far typical only
of primitive hard-coded operators. The type inference algorithm is
based on semantic type checking, resulting in type information that is
both precise, and flexible enough to handle structured and
semistructured data.
  We illustrate the use of this calculus by encoding a large fragment
of Jaql, including operations and iterators over JSON, embedded SQL
expressions, and co-grouping, and show how the encoding directly
yields a typing discipline for Jaql as it is, namely without the
addition of any type definition or type
annotation in the code.
\end{abstract}

\let\thefootnote\relax\footnotetext{An extended abstract of this work is included in the proceeding of \emph{POPL\,13,  40th ACM Symposium on Principles of Programming Languages}, ACM Press, 2013.} 

\section{Introduction}
\label{label:introduction}
The emergence of Cloud computing, and the ever growing importance of
data in applications, has given birth to a whirlwind of new data
models~\cite{json,odata} and languages. Whether they are developed
under the banner of ``NoSQL''~\cite{UnQL,JSONiq}, for BigData
Analytics~\cite{OlstonRSKT08,BeyerEGBEKOS11,jaql}, for Cloud
computing~\cite{Asterix}, or as domain specific languages (DSL)
embedded in a host
language~\cite{Squeryl,Meijer11,DBLP:conf/icfp/OhoriU11}, most of them
share a common subset of SQL and the ability to handle semistructured
data. \changes{While there is no consensus yet on the precise boundaries of this class of languages, they all share two common traits: \emph{(i)} an
emphasis on sequence operations (\emph{eg}, through the popular MapReduce
paradigm) and \emph{(ii)} a lack of types for both data and programs
(contrary to, say, XML programming or relational databases where data schemas are pervasive).}
 In~\cite{MeijerB11,Meijer11}, Meijer argues that such languages
can greatly benefit from formal foundations, and suggests
comprehensions~\cite{WadTri89,DBLP:conf/icdt/TannenBW92,DBLP:journals/sigmod/BunemanLSTW94}
as a unifying model. Although we agree with Meijer for the need to
provide unified, formal foundations to those new languages, we argue
that such foundations should account for novel features critical to
various application domains that are not captured by
comprehensions. Also, most of those languages provide limited type
checking, or ignore it altogether. We believe type checking is
essential for many applications, with usage ranging from error
detection to optimization.  \changes{But we understand the designers and programmers
of those languages who are averse to any kind of type definition or annotation.}
In this paper, we propose a calculus which
is expressive enough to capture languages that go beyond SQL or
comprehensions. We show how the calculus adapts to various data models
while retaining a precise type checking  \changes{that can exploit in a flexible way limited type
information, information that is deduced directly from the structure of the
program even in the absence of any explicit type declaration or
annotation.}\looseness -1

\paragraph*{Example.}

We use Jaql~\cite{BeyerEGBEKOS11,jaql}, a language over
JSON~\cite{json} developed for BigData analytics, to illustrate how
our proposed calculus works. Our reason for using Jaql is that it
encompasses all the features found in the previously cited query
languages and includes a number of original ones, as well.  Like
Pig~\cite{OlstonRSKT08} it supports sequence iteration, filtering, and
grouping operations on non-nested queries.  Like AQL~\cite{Asterix}
and XQuery~\cite{xquery}, it features nested queries. Furthermore,
Jaql uses a rich data model that allows arbitrary nesting of data (it
works on generic sequences of JSON records whose fields can contain
other sequences or records) while other languages are limited to flat
data models, such as AQL whose data-model is similar to the standard
relational model used by SQL databases (tuples of scalars and of lists
of scalars). Lastly, Jaql includes SQL as an embedded sub-language for
relational data. For these reasons, although in the present work we focus
almost exclusively on Jaql, we believe that our work can be adapted without
effort to a wide array of sequence processing languages.

% Among the various languages cited
%previously, Jaql is the most expressive, and includes a number of
%original features. It is strictly more expressive than
%Pig~\cite{OlstonRSKT08} (which supports non-nested , and allows for nested queries like
%AQL~\cite{Asterix} or XQuery~\cite{xquery}. It also includes SQL as an
%embedded sub-language for relational data.

%
The following Jaql program illustrates some of those features. It
performs co-grouping~\cite{OlstonRSKT08} between one JSON input,
containing information about departments, and one relational input
containing information about employees. The query returns for each
department its name and id, from the first input, and the number of
high-income employees from the second input. A SQL expression is used
to select the employees with income above a given value, while a Jaql
filter is used to access the set of departments and the elements of
these two collections are processed by the group expression (in Jaql
``\verb|$|'' denotes the current element).\looseness -1
\begin{alltt}\ibm
  group
    (depts -> filter each x (x.size > 50))
          by g = $.depid  as ds,
    (SELECT * FROM employees WHERE income > 100)
          by g = $.dept as es
  into \{ dept: g,
         deptName: ds[0].name,
         numEmps: count(es) \};
\end{alltt}
%$
The query blends Jaql expressions (\emph{eg}, \texttt{\ibm filter}
which selects, in the collection \texttt{\ibm depts}, departments with
a \texttt{\ibm size} of more than \texttt{\ibm 50} employees, and the
grouping itself) with a SQL statement (selecting employees in a
relational table for which the salary is more than \texttt{\ibm 100}).
Relations are naturally rendered in JSON as collections of records. In
our example, one of the key difference is that field access in SQL
requires the field to be present in the record, while the same
operation in Jaql does not. Actually, field selection in Jaql is very
expressive since it can be applied also to collections with the effect
that the selection is recursively applied to the components of the
collection and the collection of the results returned, and similarly
for \texttt{\ibm filter} and other iterators. In other words, the
expression \texttt{\ibm filter each x (x.size > 50)} above will work
as much when \texttt{\ibm x} is bound to a record (with or without a
\texttt{\ibm size} field: in the latter case the selection returns {\ibm\Null}),
as when \texttt{\ibm x} is bound to a collection of
records or of arbitrary nested collections thereof.  This accounts for
the semistructured nature of JSON compared to the relational
model. Our calculus can express both, in a way that illustrates the
difference in both the dynamic semantics and static typing.

In our calculus, the selection of all records
whose \emph{mandatory} field \texttt{income} is greater than 100 is
defined as:
%% \begin{alltt}
%% let Sel =
%%     `nil => `nil
%%   | (\{income: x, dept: y, ..\},tail) =>
%%        if x > 100 then (y,Sel(tail)) else Sel(tail)
%% \end{alltt}
\vspace{-1.5mm}
\begin{alltt}
let Sel =
    `nil => `nil
  | (\{income: x, .. \} as y , tail) =>
       if x > 100 then (y,Sel(tail)) else Sel(tail)
\end{alltt}\vspace{-1.5mm}
(collections are encoded as lists \emph{à la} Lisp) while the filtering
among records or arbitrary nested collections of records of those where the (optional) \texttt{size} field is present and larger than
50 is:\vspace{-1.5mm}
\begin{alltt}
let Fil =
    `nil => `nil
  | (\{size: x, .. \} as y,tail) =>
        if x > 50 then (y,Fil(tail)) else Fil(tail)
  | ((x,xs),tail) =>  (Fil(x,xs),Fil(tail))
  | (\_,tail) => Fil(tail)
\end{alltt}\vspace{-1.5mm}
The terms above show nearly all the basic building blocks of our
calculus (only composition is missing), building blocks that we dub
\emph{filters}. Filters can be defined recursively (\emph{eg},
\texttt{Sel(tail)} is a recursive call); they can perform pattern
matching as found in functional languages (the filter $\fpat{p}{f}$ executes $f$ in the environment
resulting from the matching of pattern $p$); they can be composed in
alternation ($\falt{f_1}{f_2}$ tries to apply $f_1$ and if it fails it
applies $f_2$), they can spread over the structure of their argument
(\emph{eg}, $\fprod{f_1}{f_2}$ ---of which \texttt{(x,Sel(tail))} is
an instance--- requires an argument of a product type and applies the
corresponding $f_i$ component-wise).

For instance, the filter \texttt{Fil} scans collections encoded as
lists \emph{\`a la} Lisp (\emph{ie}, by right associative pairs with
\NIL denoting the empty list). If its argument is the empty list, then
it returns the empty list; if it is a list whose head is a record with a
\texttt{size} field (and possibly other fields matched by
``\textbf{.\,.}''), then it captures the whole record in \texttt{y}, the
content of the field in \texttt{x}, the tail of the list in
\texttt{tail}, and keeps or discards \texttt{y} (\emph{ie}, the record)
according to whether \texttt{x} (\emph{ie}, the field)
 is larger than 50; if the head is also a list, then it recursively
applies both on the head and on the tail; if the head of the list is
neither a list, nor a record with a size field, then the head is
discarded.%
\ifLONGVERSION
\footnote{For the sake of precision, to comply with Jaql's semantics the last pattern should rather
  be \texttt{(\{..\},tail) => Fil(tail)}: \changes{since field selection $e\texttt{.size}$ fails
  whenever $e$ is not a record or a list, this definition would detect the possibility of this failure by a static type error.}  } 
\else
\ 
\fi
The encoding of the
whole grouping query is given in Section~\ref{sec:example}.

%\beppe{Old version commented in the source}
% \begin{tabbing}
%   $\frec{X}{(}$\=\kill
%   $\frec{X}{(\falt{\fpat{\NIL}{\NIL}~}{$\\
%       \> $\fpat{\patpair{\patorec{income \col x}}{~l}}{\fseq{x > 5000}{$\\
%           \qquad\qquad\=(\=$\fpat{\true}{\fprod{x}{~X~l}}$\\
%           \>\falt{}{}\>$\fpat{\false}{~X ~l}}}}}))$\\
% \\
%   $\frec{X}{(}$\=\kill
%   $\frec{X}{(\falt{\fpat{\NIL}{\NIL}~}{$\\
%       \> $($\=$\fpat{\patpair{\patorec{income \col x}}{~l}}{x > 5000}$\\
%       \> \> |$\fpat{\patpair{\_}{~l}}{\fseq{\NIL)}{$\\
%           \qquad\qquad\=(\=$\fpat{\true}{\fprod{x}{~X~l}}$\\
%           \>\falt{}{}\>$\fpat{\false}{~X ~l}$\\
%           \>\falt{}{}\>$\fpat{\NIL}{~X ~l}}}}}))$\\
% \end{tabbing}}
% {\color{darkgreen}
% \beppe{
% I would write the second as follows:}
% \begin{alltt}
% let filter X = (
%     `nil => `nil
%   | ((\{income : x\&[5000--*],..\}=>x),X)
%   | (\_,tail) => X(tail) )
% \end{alltt}
% \beppe{or without syntactic sugar, the first and the second:}
% \begin{alltt}
% let filter X = (
%     `nil => `nil
%   | ( \{income : x\&[5000--*],..\} , tail )=> (x,X(tail))
%   | ( \{income : _,..\} ,tail ) => X(tail) )

% let filter X = (
%     `nil => `nil
%   | ((\{income : x\&[5000--*],..\},tail)=> (x,X(tail))
%   | (\_,tail) => X(tail) )
% \end{alltt}
% }

%\medskip\noindent
\changes{Our aim is not to propose yet another ``NoSQL/cloud computing/bigdata
analytics'' query language, but rather to show how to
\emph{express} and \emph{type} such languages
via an encoding into our core calculus.} Each such language can in this
way preserve its execution model but obtain for free a formal
semantics, a type inference system and, as it happens, a prototype
implementation. \changes{The type information is deduced via the encoding (without the
need of \emph{any} type annotation) and can be used for early error
detection and debugging purposes. The encoding also yields an
executable system that can be used for rapid prototyping.} Both possibilities are
critical in most typical usage scenarios of these languages, where
deployment is very expensive both in time and in resources. As
observed by Meijer~\cite{Meijer11} the advent of big data makes it
more important than ever for programmers (and, we add, for language
and system designers) to have a single abstraction that allows them to
process, transform, query, analyze, and compute across data presenting
utter variability both in volume and in structure, yielding a
``mind-blowing number of new data models, query languages, and
execution fabrics''~\cite{Meijer11} . The framework we present here,
we claim, encompasses them all.  A long-term goal is that the compilers
of these languages could use the type information inferred from the
encoding and the encoding itself to devise further optimizations.

\iffalse%%%%%%%%%%%%%%%%%%%%%%=========>
\beppe{WHAT SHALL WE DO OF THESE SENTENCES?:

Proposed extensions to comprehensions for group by~\cite{Jones2007} do
not support. Instead, we propose a small but very expressive calculus,
inspired by tree transducers, and the work done
in~\cite{Castagna2008}. The calculus captures a large class of
recursive transformations which can express queries over collections,
pattern matching, as well as tree transformations.
}
\jer{Should be expanded a little to talk about what can be expressed in the calculus, and relate to the example.}
\fi%%%%%%%%%%%%%%%%%%%%%%%<==========

\paragraph*{Types.}

Pig~\cite{OlstonRSKT08}, Jaql~\cite{jaql,biginsights},
AQL~\cite{Asterix} have all been conceived by considering just the
map-reduce execution model. The type (or, schema) of the manipulated
data did not play any role in their design. As a consequence these
languages are untyped and, when present, types are optional and
clearly added as an afterthought. Differences in data model or type
discipline are particularly important when embedded in a host language
(since they yield the so-called impedance mismatch).  The reason why
types were/are disregarded in such languages may originate in an
alleged tension between type inference and
heterogeneous/semistructured data: on the one hand these languages are
conceived to work with collections of data that are weakly or
partially structured, on the other hand current languages with
type inference (such as Haskell or ML) can work only on
homogeneous collections (typically, lists of elements of the same
type).\looseness -1

In this work we show that the two visions can coexist: we type data by
semantic subtyping~\cite{jacm08}, a type system conceived for
semi-structured data, and describe computations by our \emph{filters}
which are untyped combinators that, thanks to a technique of weak
typing introduced in~\cite{Castagna2008}, can polymorphically type the
results of data query and processing with a high degree of
precision. The conception of \emph{filters} is driven by the schema of
the data rather than the execution model and we use them
$(\romannumeral 1)$ to capture and give a uniform semantics to a wide
range of semi structured data processing capabilities, $(\romannumeral
2)$ to give a type system that encompasses the types defined for such
languages, if any, notably Pig, Jaql and AQL (but also XML query and
processing languages: see Section~\ref{sec:example}), $(\romannumeral
3)$ to infer the precise result types of queries written in these
languages \emph{as they are} (so without the addition of any explicit
type annotation/definition or new construct), and $(\romannumeral 4)$ to show how
minimal extensions/modifications of the current syntax of these languages can bring
dramatic improvements in the precision of the inferred types.\looseness-1

The types we propose here are extensible record types and
heterogeneous lists whose content is described by regular expressions
on types as defined by the following grammar:\\[2mm]
$\begin{array}{lrclr}
\textbf{Types}
  & t &::=  & v &\text{(singleton)} \\
  &   &\mid &\patrec{\ell\col t, \ldots, \ell\col t} &\text{(closed record)}\\
  &   &\mid &\patorec{\ell\col t, \ldots, \ell\col t} &\text{(open record)}\\
  &   &\mid &\LIST{r} &\text{(sequences)}\\
  &   &\mid &\texttt{int}\mid\texttt{char} &\text{(base)} \\
  &   &\mid &\ANY\mid\EMPTY\mid\texttt{null} &\text{(special)}\\
  &   &\mid &\pator{t}{t} &\text{(union)}\\
  &   &\mid &\patdiff{t}{t} &\text{(difference)}\\
\\
\textbf{Regexp} & r&::= &\varepsilon~\mid~ t ~\mid~ r\ks~\mid~ r\texttt{+}~\mid~ r\why~\mid~ r\,r~\mid~\pator{r}{r}\hspace*{-15mm}
\end{array}$\\[2mm]
where $\varepsilon$ denotes the empty word.  The semantics of types can
be expressed in terms of sets of \emph{values} (values are either
constants ---such as \texttt{1}, \texttt{2}, \texttt{true},
\texttt{false}, \Null, \texttt{'1'}, the latter denoting the character
\texttt{1}---, records of values, or lists of values). So the
singleton type $v$ is the type that contains just the value $v$ (in
particular \Null\ is the singleton type containing the value
\Null). The closed record type $\patrec{\fa\col\Int, \fb\col\Int}$
contains all record values with exactly two fields \fa~ and \fb~ with
integer values, while the open record type $\patorec{\fa\col\Int,
  \fb\col\Int}$ contains all record values with \emph{at least} two
fields \fa~ and \fb~ with integer values. The sequence type $\LIST{r}$
is the set of all sequences whose content is described by the
\emph{regular expression} $r$; so, for example \LIST{\texttt{char*}}
contains all sequences of characters (we will use \texttt{string} to
denote this type and the standard double quote notation to denote its values) while \LIST{(\patrec{\fa\col\Int}
  \patrec{\fa\col\Int})\texttt{+}} denotes nonempty lists of even
length containing record values of type $\patrec{\fa\col\Int}$.  The
union type $\pator{s}{t}$ contains all the values of $s$ and of $t$,
while the difference type $\patdiff{s}{t}$ contains all the values of
$s$ that are not in $t$.  We shall use \texttt{bool} as an
abbreviation of the union of the two singleton types containing true
and false: $\pator\true\false$. \ANY~ and \EMPTY~ respectively contain
all and no values. Recursive type definitions are also used (see Section~\ref{sec:types} for formal details).

These types can express all the types of Pig, Jaql and
AQL%
\ifLONGVERSION\footnote{The only exception are the ``bags'' types we did not
  include in order to focus on essential features.}%
\fi%%%%%%%%%%%
, all XML types, and
much more. So for instance, AQL includes only homogeneous lists of type
$t$, that can be expressed by our types as \LIST{ t\ks\ }. In Jaql's
documentation one can find the type \texttt{\ibm [~long(value=1),
    string(value="a"), boolean* ]} which is the type of arrays whose
first element is 1, the second is the string "a" and all the other are
booleans. This can be easily expressed in our types as \LIST{\texttt{1
    "a" bool*}}. But while Jaql only allows a limited use of regular
expressions (Kleene star can only appear in tail position) our types
do not have such restrictions.  So for example \LIST{\texttt{char* '@'
    char* '.' (('f' 'r')|('i' 't'))}} is the type of all strings
(\emph{ie}, sequences of chars) that denote email addresses ending by
either \texttt{.fr} or \texttt{.it}. We use some syntactic sugar
to make terms as the previous one more readable (\emph{eg},
\LIST{\texttt{ .* '@' .* ('.fr'|'.it')}}). Likewise, henceforth we
use \texttt{\patrec{a?\col $t$}} to denote that the field
\texttt{a} of type $t$ is optional; this is just syntactic sugar for
stating that either the field is undefined or it contains a value of type $t$ (for the formal definition see 
\appcrazy~\ref{tsum}%
\ifcrazy%%%%%%%%%%%%%%%%%%
\ available online%
\fi%%%%%%%%%%%%%%%%%%
).

Coming back to our initial example, the filter \texttt{Fil} defined
before expects as argument a collection of  the following type:\\[1.5mm]
%\begin{alltt}
\centerline{\tt  type Depts = [ ( \{size?: int, ..\} | Depts )* ]}\\[1.5mm]
%\end{alltt}
that is a, possibly empty, arbitrary nested list of records with an
optional \texttt{size} field of type \Int: notice that it is important
to specify the optional field and its type since a \texttt{size} field
of a different type would make the expression \texttt{x > 50} raise a
run-time error. \changes{This information is deduced just from the \emph{structure} 
of the filter (since \texttt{Fil} does not contain any type definition or annotation).} 

We define a type inference system that rejects any argument of
\texttt{Fil} that has not type \texttt{Depts}, and deduces for
arguments of type
%\\\centerline
{\texttt{[(\{size: int, addr: string\}| \{sec: int\} | Depts)+]}}
(which is a subtype of \texttt{Depts}) the result type
\texttt{[(\{size: int, addr: string\}|Depts)*]} (so it does not forget
the field \texttt{addr} but discards the field \texttt{sec}, and by
replacing \ks\ for \kss\ recognizes that the test may fail).

%% technique that will allow us to deduce that will reject
%% when the Jaql expression

%% \texttt{transform each x with \patrec{ x.*, sum\col x.a + x.b }}

%% \noindent
%% that scans a list of records adds them a sum field containing the sum
%% of the fields a and b, is applied to an array of
%% type \LIST{\texttt{\patrec{a\col int, b\col int, c\col bool}*}}, the
%% result will have type \LIST{\texttt{\patrec{a\col int, b\col int,
%% c\col bool, sum\col int}*}} and that if applied to list of records
%% that miss the a or b fields will yield a type error.
By encoding primitive Jaql operations into a formal core calculus we
shall provide them a formal and clean semantics as well as precise
typing. So for instance it will be clear that applying the following
dot selection \verb|[ [{a:3}] {a:5, b:true} ].a| the result will be
\verb|[ [3] 5 ]| and we shall be able to deduce that \verb|_.a|
applied to arbitrary nested lists of records with an optional integer
\texttt{a} field (\emph{ie}, of type $\texttt{\textit{t} =
  \texttt{\patrec{a?\col int} | [ \textit{t}* ]}}$ ) yields arbitrary
nested lists of \Int\ or \Null\ values (\emph{ie}, of type
$\texttt{\textit{u} = int | null | [ \textit{u}* ]}$).

Finally we shall show that if we accept to extend the current syntax
of Jaql (or of some other language) by some minimal filter syntax
(\emph{eg}, the pattern filter) we can obtain a huge
improvement in the precision of type inference.

\paragraph{Contributions.}
The main contribution of this work is the definition of a calculus
that encompasses structural operators scattered over
NoSQL languages
%, capable both to formally define their semantics and to
%define/integrate new/different ones within a type system of a
%precision typical only of built-in operators.
%
%The resulting formalism 
and that possesses some characteristics that make it
unique in the swarm of current semi-structured data processing
languages. In particular it is parametric (though fully embeddable)
in a host language; it uniformly handles both width and deep nested
data recursion (while most languages offer just the former and limited
forms of the latter); finally, it includes first-class arbitrary deep
composition (while most languages offer this operator only at top
level), whose power is nevertheless restrained by the type
system.\looseness -1

% An important contribution of this work is that it directly compares a
% programming language approach with the tree transducer one: our
% calculus implements transformations typical of top-down tree
% transducers but $(1)$ expressed in a formalism immediately intelligible
% to any functional programmer, \changes{$(2)$ that, in its untyped version,
% is Turing complete and $(3)$ can be statically typed (at the expenses of Turing completeness) without any
% annotation yielding precise result types} $(4)$ but which even
% restricted to well-typed terms is strictly more expressive than well-known
% and widely studied deterministic top-down tree transducer formalisms.

An important contribution of this work is that it directly compares a
programming language approach with the tree transducer one. Our
calculus implements transformations typical of top-down tree
transducers but has several advantages over the transducer approach:
$(1)$ the transformations are expressed in a formalism immediately intelligible
to any functional programmer; \changes{$(2)$ our calculus, in its untyped version,
is Turing complete; $(3)$ its transformations can be statically typed (at the expenses of Turing completeness) without any
annotation yielding precise result types} $(4)$ even if we restrict
the calculus only to well-typed terms (thus losing Turing completeness), it still is strictly more expressive than well-known
and widely studied deterministic top-down tree transducer formalisms.

The technical contributions are $(\romannumeral 1)$ the proof of
Turing completeness for our formalism, $(\romannumeral 2)$ the
definition of a type system that copes with records with computable labels
$(\romannumeral 3)$ the definition of a static type system for filters
and its correctness, $(\romannumeral 4)$ the definition of a static
analysis that ensures the termination (and the proof thereof) of the type inference algorithm with complexity bounds expressed
in the size of types and filters and $(\romannumeral 4)$ the proof
that the terms that pass the static analysis form a language strictly more
expressive than top-down tree transducers.

\paragraph{Outline.}
In Section~\ref{label:definitions} we present the syntax of
the three components of our system. Namely, a minimal set of \emph{expressions},
the calculus of \emph{filters} used to program user-defined operators
or to encode the operators of other languages, and the core
\emph{types} in which the types we just presented are to be
encoded.
Section~\ref{label:dynamic} defines the operational semantics
of filters and a declarative semantics for operators. The type system
as well as the type inference algorithm are described in
Section~\ref{label:typing}. In Section~\ref{label:jaql} we present how
to handle a large subset of Jaql. 
\ifLONGVERSION%%%%%%%%%%%%%
Section~\ref{label:comment} reports
on some subtler design choices of our system. We compare
related works in Section~\ref{label:related} and
conclude in Section~\ref{label:conclusion}.
\else%%%%%%%%%%%%%%%%%%%%%
Section~\ref{label:comment} reports
on some subtler design choices of our system and compare
with related work. 
\fi%%%%%%%%%%%%%%%%%%%%%%%%
\ifcrazy%%%%%%%%%%%%%%%%
Due to strict space limitations,  proofs, secondary results, encodings, and further extensions are only available in the version accessible online at any of the first three authors' web pages or from PC-chair: references in the text starting with capital letters (\emph{eg}, Section H.1) refer to it.
\else
In order to avoid blurring the presentation, proofs, secondary results, further encodings, and extensions are 
moved into a separate appendix.
\fi%%%%%%%%%%%%%%%%%%%%%%%

\section{Syntax}
\label{label:definitions}
In this section we present the syntax of the three components of our
system: a minimal set of \emph{expressions}, the calculus of
\emph{filters} used to program user-defined operators or to encode the operators of other languages, and
the core \emph{types} in which the types presented in the
introduction are to be encoded.

The core of our work is the definition of filters and types. The key property
of our development is that filters can be grafted to any host language that satisfies
minimal requirements, by simply adding filter application to the
expressions of the host language. The minimal requirements of the host language for this to be possible are quite simple: it must have
constants (typically for types \Int, \Char, \String, and \Bool),
variables, and either pairs or record values (not necessarily both).
On the static side the host language must have at least basic and products
types and be able to assign a type to expressions in a given type environment (\emph{ie}, under some typing assumptions for variables).
By the addition of filter applications, the host language can acquire
or increase the capability to define polymorphic user-defined
iterators, query and processing expressions, and be enriched with a powerful
and precise type system.

\subsection{Expressions}
%\ifLONGVERSION%%%%%%%%%%%%%%%%%%%%%%%%%%%%%%%
In this work we consider the following set of expressions 
%\fi%%%%%%%%%%%%%%%%%%%%%%%%%%%%%%%%%%%%%%%%%%
%which, intuitively, represent the host language
\begin{definition}[expressions\vspace{-1mm}%
]\label{def:exp}
\begin{displaymath}\vspace{-3.5mm}
\begin{array}{lclr}
\textbf{Exprs}\quad
e & :: = & c & \text{(constants)}\\
&\mid& x & \text{(variables)}\\
&\mid& (e,e) & \text{(pairs)}\\
&\mid& \{e\col e,...,e\col e\} & \text{(records)}\\
&\mid& e + e & \text{(record concatenation)}\\
&\mid& e \setminus\ell  & \text{(field deletion)}\\
&\mid& \textup{\sf op}(e,...,e) &\text{(built-in operators)}\\
&\mid& fe & \text{(filter application)}\\[2mm]
\end{array}
\end{displaymath}
where $f$ ranges over \emph{filters} (defined later on),
$c$ over generic constants, and $\ell$ over \emph{string} constants.
\end{definition}
Intuitively, these expressions represent the syntax supplied by the
host language ---though only the first two and one of the next two are really needed--- that
we extend with (the missing expressions and) the expression of filter
application. Expressions are formed by constants,
variables, pairs, records, and operation on records: record
concatenation gives priority to the expression on the right. So if in
$r_1+r_2$ both records contains a field with the same label, it is the
one in $r_2$ that will be taken, while field deletion does not require
the record to contain a field with the given label (though this point
is not important).  The metavariable \textsf{op} ranges over operators
as well as functions and other constructions belonging to or defined by
the host language. Among expressions we single out a set of
\emph{values}, intuitively the results of computations, that are
formally defined as follows:
$$ \begin{array}{lcl}
   v & ::= & c ~|~ (v, v) ~|~ \{ \ell\col v;~ \ldots ;~ \ell\col v \}\\
 \end{array}
$$
 We use \textup{\texttt{"foo"}} for character
string constants, \textup{\texttt{'c'}} for characters,
\textup{\texttt{1 2 3 4 5}} and so on for
integers, and backquoted words, such as \textup{\texttt{`foo}}, for
atoms (\emph{ie}, user-defined constants). We use three distinguished atoms \NIL, \true, and \false. Double quotes can be omitted for strings that are labels of record fields: thus we write $\{\texttt{name}\col\texttt{"John"}\}$ rather than  $\{\texttt{"name"}\col\texttt{"John"}\}$.
Sequences (\emph{aka}, heterogeneous lists, ordered collections, arrays) are encoded \emph{à la LISP}, as nested pairs where the
atom \NIL denotes the empty list. We use 
  \LIST{$e_1$ ~\ldots~ $e_n$} 
as syntactic sugar for
  $(e_1, \ldots, (e_n, \NIL)...)$.

\subsection{Types}\label{sec:types}
\ifLONGVERSION%%%%%%%%%%%%%%%
Expressions, in particular filter applications, are typed by
the following set of types (typically only basic, product,
recursive and ---some form of--- record types will be provided by the
host language):
\fi%%%%%%%%%%%%%%%%%%%%%%%%
\begin{definition}[types]\label{def:types}
\begin{displaymath}
\begin{array}{rclr}
\textbf{Types}\quad
t &::= & b & \text{(basic types)}\\
 &~|~& v & \text{(singleton types)}\\
 &~|~& \patpair{t}{t} & \text{(products)}\\
 &~|~& \patrec{\ell\col t, \ldots, \ell\col t} & \text{(closed records)}\\
 &~|~& \patorec{\ell\col t, \ldots, \ell\col t}\hspace*{-4mm} & \text{(open records)}\\
 &~|~& \pator{t}{t} & \text{(union types)}\\
 &~|~& \patand{t}{t}& \text{(intersection types)} \\
 &~|~& \synneg t & \text{(negation type)}\\
 &~|~& \tempty & \text{(empty type)}\\
 &~|~& \tany &\text{(any type)} \\
 &~|~& \mu T.t & \text{(recursive types)}\\
 &~|~& T& \text{(recursion variable)}\\
 &~|~& \textup{\sf Op}(t,...,t)& \text{(foreign type calls)}
\end{array}
\end{displaymath}
where every recursion is guarded, that is, every type
variable is separated from its binder by at least one application of a
type constructor (\emph{ie}, products, records, or
\textup{\textsf{Op}}).
\end{definition}
Most of these types were already explained in the introduction. We have basic types (\Int,
\Bool, ....) ranged over by $b$ and singleton types $v$
denoting the type that contains only the value $v$. Record types come
in two flavors: closed record types whose values are records with
exactly the fields specified by the type, and open record types whose
values are records with \emph{at least} the fields specified by the
type. Product types are standard and we have a complete set of type connectives, that is, finite
unions, intersections and negations. We use $\tempty$, to denote the
type that has no values and $\tany$ for the type of all values
(sometimes denoted by ``\texttt{\_}'' when used in patterns).
We added a term for recursive types, which allows us to encode both the regular expression types defined in the introduction and, more generally, the recursive type definitions we used there.
Finally,
we use \textup{\sf Op} (capitalized to distinguish it from expression
operators) to denote the host language's \emph{type} operators (if any). Thus,
when filter applications return values whose type belongs just to the foreign
language (\emph{eg}, a list of functions) we suppose the typing
of these functions be given by some type operators. For instance, if
\textsf{succ} is a user defined successor function, we will suppose to
be given its type in the form \textsf{Arrow}(\Int,\Int) and, similarly,
for its application, say \textsf{apply}(\textsf{succ},3) we will be
given the type of this expression (presumably \Int). Here \textsf{Arrow} is a type operator and \textsf{apply} an expression operator.

\rechanges{The denotational semantics of types as sets of values, that
  we informally described in the introduction, is at the basis of the
  definition of the subtyping relation for these types. We say that a
  type $t_1$ is a subtype of a type $t_2$, noted $t_1\sleq t_2$, if
  and only if the set of values denoted by $t_1$ is contained (in the
  set-theoretic sense) in the set of values denoted by $t_2$. For the
  formal definition and the decision procedure of this subtyping
  relation the reader can refer to the work on semantic
  subtyping~\cite{jacm08}.}\looseness -1
\subsection{Patterns}
%\beppe{Since patterns are not recursive I left the constant pattern only for the CONSTANTPATTERN} 
%
Filters are our core untyped operators. All they can do are three different
things: $(1)$ they can structurally decompose and transform the values they
are applied to, or $(2)$ they can be sequentially composed, or $(3)$ they can do pattern
matching. In order to define filters, thus, we first need to define
patterns.
\begin{definition}[patterns]\label{patdef}
\ifCONSTANTPATTERN%%%%%%%%%%%%%%%
\begin{displaymath}
  \begin{array}{lclr}
\textbf{Patterns}\quad
p & ::= & t                                              & \text{(type)}\\
  &  |  & x                                              & \text{(variable)}\\
  &  |  &  \patpair p p                                  & \text{(pair)}\\
  &  |  &  \patrec {\ell\col p, \ldots, \ell\col p}   & \text{(closed rec)}\\
  &  |  &  \patorec {\ell\col p, \ldots, \ell\col p}  & \text{(open rec)}\\
  &  |  &  \pator p p                                    & \text{(or/union)}\\
  &  |  &  \patand p p                                   & \text{(and/intersection)}\\
  &  |  &  \patcon{x}{c}                                 & \text{(constant)}
  \end{array}
 \end{displaymath}
\else%%%%%%%%%%%%%%%%%%%%%%%%%%%%%%
\begin{displaymath}
  \begin{array}{lclr}
\textbf{Patterns}\quad
p & ::= & t                                              & \text{(type)}\\
  &  |  & x                                              & \text{(variable)}\\
  &  |  &  \patpair p p                                  & \text{(pair)}\\
  &  |  &  \patrec {\ell\col p, \ldots, \ell\col p}   & \text{(closed rec)}\\
  &  |  &  \patorec {\ell\col p, \ldots, \ell\col p}  & \text{(open rec)}\\
  &  |  &  \pator p p                                    & \text{(or/union)}\\
  &  |  &  \patand p p                                   &\hspace*{-5mm} \text{(and/intersection)}
  \end{array}
 \end{displaymath}
\fi%%%%%%%%%%%%%%%%%%%%%%%%%%
where the subpatterns forming pairs, records, and intersections have distinct capture variables, and those forming unions have the same capture variables.
\end{definition}
Patterns are essentially types in which capture variables (ranged over by $x$, $y$, \ldots) may occur in every position that is not under a negation or a recursion%
\ifCONSTANTPATTERN%%%%%%%%%%%%%%%%%%%%%%%
,plus the default value
pattern $\patcon{x}{c}$, which is a pattern matched by every value and
returning the substitution $\{x/c\}$. 
\else.\fi%%%%%%%%%%%%%%%%%%%%%%%%%%%%%
~A pattern is used to match a
value. The matching of a value $v$ against a pattern $p$, noted $v/p$, either fails (noted $\Omega$) or it returns a substitution from the
variables occurring in the pattern, into values.
 The
substitution is then used as an environment in which some expression
is evaluated. If the pattern is a type,
then the matching fails if and only if the pattern is matched against
a value that does not have that
type, otherwise it returns the empty substitution. If it is a variable, then the matching always succeeds and returns the
substitution that assigns the matched value to the variable. The pair pattern
$\patpair{p_1}{p_2}$ succeeds if  and only if it is matched against a pair of values
and each sub-pattern succeeds on the corresponding projection of the value (the
union of the two substitutions is then returned). Both record patterns are similar to the product pattern with the specificity that in the open record pattern ``\textbf{..}'' matches all the fields that are not specified in the pattern.
An intersection pattern
$\patand{p_1}{p_2}$ succeeds if and only if both patterns succeed (the union of
the two substitutions is then returned). The union pattern $\pator{p_1}{p_2}$
first tries to match the pattern $p_1$ and if it fails it tries the pattern $p_2$.

For instance, the pattern $\patpair{\patand{\Int}{x}}{y}$ succeeds only if
the matched value is a pair of values $(v_1,v_2)$ in which $v_1$ is an integer
---in which case it returns the substitution $\{x/v_1, y/v_2\}$--- and
fails otherwise. Finally notice that the notation ``\texttt{$p$ as $x$}'' we used in the examples of the introduction, is syntactic sugar for $\patand{p}{x}$.

This informal semantics of matching (see~\cite{jacm08} for the formal definition)
explains the reasons of the restrictions on capture variables in Definition~\ref{patdef}:
in intersections, pairs, and records all patterns
must be matched and, thus, they have to assign distinct variables, while in union
patterns just one pattern will be matched, hence the same set of variables must be
assigned, whichever alternative is selected.

The strength of patterns is their connections with types and the fact
that the pattern matching operator can be typed \emph{exactly}. This is
entailed by the following theorems (both proved in~\cite{jacm08}):
\begin{theorem}[Accepted type~\cite{jacm08}]\label{th:at}
For every pattern $p$, the set of all values $v$ such that $v/p
\neq \Omega$ is a type. We call this set the \emph{accepted type} of $p$ and note it by $\ftype{p}$.
\end{theorem}
\noindent The fact that the exact set of values for which a  matching succeeds is a type is not obvious. \rechanges{It states that for every pattern $p$ there exists a syntactic type produced by the grammar in Definition~\ref{def:types} whose semantics is exactly the set of all and only values that are matched by $p$. The existence of this syntactic type, which we note $\ftype{p}$,} is of  utmost
importance for a precise typing of pattern matching. In particular, given a pattern $p$ and a type $t$ contained in (\emph{ie}, subtype of) $\ftype p$, it allows us to compute the \emph{exact} type of the capture variables of $p$ when it is matched against a value in $t$:
\begin{theorem}[Type environment~\cite{jacm08}] \label{typenv}There exists an algorithm that
for every pattern $p$, and $t\sleq\ftype{p}$ returns a type
environment $t/p\in \Vars(p)\to\Types$ such that
$(t/p)(x)=\{(v/p)(x)~|~v : t\}$.
\end{theorem}

\subsection{Filters}
%We are now equipped to present the filter calculus:
\begin{definition}[filters]
  A \emph{filter} is a term generated by:
  \begin{displaymath}
  \begin{array}{lclr}
  \textbf{Filters}\hfill     f & ::= & e           & \textrm{(expression)}\\
      & |   & \fpat{p}{f} & \textrm{(pattern)}\\
      & |   & \fprod{f}{f} & \textrm{(product)}\\
      &  |  &  \forec {\ell\col f, \ldots, \ell\col f}\hspace*{-20mm}  & \text{(record)}\\
      & |   & \falt{f}{f}  & \textrm{(union)}\\
      & |   & \frec{X}{f}   & \textrm{(recursion)}\\
      & |   & X a          & \textrm{(recursive call)}\\
      & |   & \fseq{f}{f}    & \textrm{(composition)}\\
      & |   & o & \textrm{(declarative operators)}\\
\\
 \textbf{Operators} \hfill     o & ::=   & \GROUPBY{f} & \textrm{(filter grouping)}\\
      & |   & \ORDERBY{f} & \textrm{(filter ordering)}\\
\\
   \textbf{Arguments}\quad  a & ::= & x & \textrm{(variables)}\\
       & |   &  c   & \textrm{(constants)}\\
       & |   &  \fprod{a}{a}   & \textrm{(pairs)}\\
       & |   &  \patrec{\ell\col a,...,\ell\col a}   & \textrm{(record)}
  \end{array}
  \end{displaymath}
such that for every subterm of the form $\fseq{f}{g}$, no recursion variable is free in$f$.
\end{definition}
Filters are like transducers, that when applied to a value
return another value. However, unlike transducers they possess more 
``programming-oriented'' constructs, like the ability to test an input
and capture subterms, recompose an intermediary result from captured
values and a composition operator. We first describe informally the
semantics of each construct.

The expression filter $e$ always returns the
value corresponding to the evaluation of $e$ (and discards its
argument). The filter $\fpat{p}{f}$ applies the
filter $f$ to its argument in the environment obtained by matching the
argument against $p$ (provided that the matching does not fail).
This rather powerful feature allows a filter to perform two critical
actions: \textit{(i)} inspect an input with regular pattern-matching
before exploring it and \textit{(ii)} capture part of the input that
can be reused during the evaluation of the subfilter $f$.
 If the argument application of $f\!_i$ to $v_i$ returns $v'_i$ then the
application of the product filter $\fprod{f_1}{f_2}$ to an argument
$(v_1,v_2)$ returns $(v'_1,v'_2)$; otherwise, if any application fails or if the argument is
not a pair, it fails. The record filter is similar: it applies to each specified
field the corresponding filter and, as stressed by the ``\textbf{.\,.}'', leaves the other fields unchanged; it fails if any of the applications does, or if any of
the specified fields is absent, or if the argument is not a record. The filter $\falt{f_1}{f_2}$ returns the
application of $f_1$ to its argument or, if this fails, the application
of $f_2$. The semantics of a recursive filter is given by standard
unfolding of its definition in recursive calls. 
The only real restriction that we introduce for filters is that recursive calls can
be done only on arguments of a given form (\emph{ie}, on arguments that have
the form of values where variables may occur). This restriction
in practice amounts to forbid recursive calls on the result of another
recursively defined filter (all other cases can be easily
encoded). The reason of this restriction is technical, since it greatly simplifies
the analysis of Section~\ref{sec:sound} (which ensures the termination of type inference) without hampering expressiveness: filters are
Turing complete even with this restriction  (see Theorem~\ref{th:tc}).
Filters can be composed:
the filter $\fseq{f_1}{f_2}$ applies $f_2$ to the result of applying
$f_1$ to the argument and fails if any of the two does.
The condition that in every subterm of the form $\fseq{f}{g}$, $f$
does not contain free recursion variables is not strictly
necessary. Indeed, we could allow such terms. The point is that the
analysis for the termination of the typing would then reject all such
terms (apart from trivial ones in which the result of the recursive
call is not used in the composition). But since this restriction \emph{does
not} restrict the expressiveness of the calculus (Theorem~\ref{th:tc}
proves Turing completeness with this restriction), then the addition of this restriction is just a design (rather than a technical)
choice: we prefer to
forbid the programmer to write recursive calls on the left-hand side  of a
composition, than systematically reject all the programs that use them
in a non-trivial way.

\begin{figure*}[t!]
\small\vspace{-2mm}
%\frame{
\begin{minipage}{250pt}
  \begin{displaymath}
\begin{array}{lc{r}}
        \textbf{(expr)} & \fract{}{ \japp{\delta\textbf{;}\gamma}{e}{v}{r} } &\hspace*{-8mm}r =
        \texttt{eval}(\gamma , e ) \\
        & & \\
        \textbf{(prod)} &  \fract{
          \japp{\delta\textbf{;}\gamma}{f_1}{v_1}{r_1}~~~\japp{\delta\textbf{;}\gamma}{f_2}{v_2}{r_2}
        }{\japp{\delta\textbf{;}\gamma}{\fprod{f_1}{f_2}}{v_1,v_2}
          {(r_1,r_2)}} & \begin{array}{r}\text{if}~~ r_1 \neq \Omega\\\text{and}~ r_2
        \neq  \Omega\end{array}\\
        & & \\
     \textbf{(patt)} &  \fract{\japp{\delta\textbf{;}\gamma ~,~%\tcup
          v/p}{f}{v}{r}}{ \japp{\delta\textbf{;}\gamma}{(\fpat{p}{f})}{v}{r}} &
    \textrm{if}~ v/p \neq \Omega \\
      & & \\
      \textbf{(comp)} & \fract{
          \japp{\delta\textbf{;}\gamma}{f_1}{v}{r_1}~~~\japp{\delta\textbf{;}\gamma}{f_2}{r_1}{r_2}
      }{
        \japp{\delta\textbf{;}\gamma}{(\fseq{f_1}{f_2})}{v}{
         r_2}} & \textrm{if}~ r_1 \neq \Omega  \\
    \end{array}
  \end{displaymath}
\end{minipage}\qquad
\begin{minipage}{250pt}
  \begin{displaymath}
    \begin{array}{lcr}
      \textbf{(union1)} &    \fract{\japp{\delta\textbf{;}\gamma}{f_1}{v}{ r_1}}{
        \japp{\delta\textbf{;}\gamma}{(\falt{f_1}{f_2})}{v}{r_1}} & \textrm{if}~ r_1 \neq \Omega  \\
      & & \\
      \textbf{(union2)} & \fract{
        \japp{\delta\textbf{;}\gamma}{f_1}{v}{\Omega}~~~\japp{\delta\textbf{;}\gamma}{f_2}{v}{
          r_2}}{\japp{\delta\textbf{;}\gamma}{(\falt{f_1}{f_2})}{v}{r_2} }& \\
      & & \\
      \textbf{(rec)} & \fract{
       \japp{\delta,(X\mapsto f)\textbf{;}\gamma}{f}{v}{r}}{\japp{\delta\textbf{;}\gamma}{(\frec{X}{f})}{v}{r}}  & \\
      & & \\
      \textbf{(rec-call)} & \fract{
       \japp{\delta\textbf{;}\gamma}{(\delta(X))}{a}{r}}{\japp{\delta\textbf{;}\gamma}{(Xa)}{v}{r}}  & \\
      & &\\
      \textbf{(error)} & \fract{}{
        \japp{\delta\textbf{;}\gamma}{f}{a}{\Omega} }
        & \makebox[3mm][r]{if no other rule applies}\\
\end{array}
\end{displaymath}
\end{minipage}
\begin{displaymath}
\begin{array}{lcr}
        \textbf{(recd)}\hspace*{6mm} & \fract{
          \japp{\delta\textbf{;}\gamma}{f_1}{v_1}{r_1}\qquad\cdots\qquad\japp{\delta\textbf{;}\gamma}{f_n}{v_n}{r_n}
        }
        {\japp{\delta\textbf{;}\gamma}{\forec{\ell_1\col f_1,...,\ell_n\col f_n}}%
          {\patrec{\ell_1\col v_1,...,\ell_n\col v_n,...,\ell_{n+k}\col v_{n+k}}}
          {\patrec{\ell_1\col r_1,...,\ell_n\col r_n,...,\ell_{n+k}\col v_{n+k}}}} &
          \text{if}~ \forall i,r_i \not= \Omega

\end{array}
\end{displaymath}
%}
\caption{Dynamic semantics of filters}
\label{fig:sem}\vspace{-3mm}
\end{figure*}

Finally, we singled out some specific filters (specifically, we chose
\GROUPBY{} and \ORDERBY) whose semantics is generally specified in a declarative
rather than operational way. These do not bring any expressive power to
the calculus (the proof of Turing completeness, Theorem~\ref{th:tc},
does not use these declarative operators) and actually they can be encoded
by the remaining filters,
%\beppe{We should add the encoding in the Appendix}
but it is interesting to single
them out because they yield either simpler encodings or more precise typing.

\section{Semantics}
\label{label:dynamic}
The operational semantics of our calculus is given by the
reduction semantics for filter application and for the record operations. 
Since the former is the only novelty of our work, we save space and omit the latter, which are standard anyhow.

\ifLONGVERSION%%%%%%%%%%%
\subsection{Big step semantics}
\else
\medskip

\noindent
\fi%%%%%%%%%%%%%%%%%%%%%%
We define a big step operational semantics for filters.
The definition is given by the inference rules in Figure~\ref{fig:sem} for judgements
of the form $\japp{\delta\textbf{;}\gamma}{f}{a}{r}$ and describes
how the evaluation of the application of filter
$f$ to an argument $a$ in an environment $\gamma$ yields an object $r$ where $r$ is
either a value or $\Omega$. The latter is a special value which
represents a runtime error: it is raised by the rule \textbf{(error)} 
either because a filter did not match the form of its argument
(\emph{eg}, the argument of a filter product was not a pair) or because some pattern matching failed
(\emph{ie}, the side condition of \textbf{(patt)} did not
hold). Notice that the argument $a$ of a filter is always a value $v$ unless the filter is the unfolding of a recursive call, in which case variables may occurr in it (\emph{cf.} rule \textbf{rec-call}). Environment $\delta$ is used to store the
body of recursive definitions.

The semantics of filters is quite straightforward and inspired by the
semantics of patterns.
The \emph{expression} filter discards its input and
evaluates (rather, asks the host language to evaluate) the expression
$e$ in the current environment (\textbf{expr}). It can be thought of
as the right-hand side of a branch in a \texttt{match\_with} construct.

The \emph{product} filter  expects a pair as input, applies its
sub-filters component-wise and returns the pair of the results
(\textbf{prod}). This
filter is used in particular to express sequence mapping, as the first
component $f_1$ transforms the element of the list and $f_2$ is applied
to the tail. In practice it is often the case that $f_2$ is a recursive call
that iterates on arbitrary lists and stops
when the input is \NIL. If the input is not a pair, then the filter fails
(rule (\textbf{error}) applies).

\rechanges{The \emph{record} filter expects as input a record value with \emph{at least} the same fields as those specified by the filter. It applies each sub-filter to the value in the corresponding field leaving the contents of other fields  unchanged (\textbf{recd}). If the argument is not a record value or it does not contain all the fields specified by the record filter, or if the application of any subfilter fails, then the whole application of the record filter fails.} 

The \emph{pattern} filter
matches its input value $v$ against the pattern $p$. If the matching fails 
so the filter does, otherwise it evaluates its
sub-filter in the environment augmented by the substitution $v/p$
(\textbf{patt}).

The \emph{alternative} filter follows a standard first-match policy:
If the filter $f_1$ succeeds, then its result is returned
(\textbf{union-1}). If $f_1$ fails, then $f_2$ is evaluated
against the input value (\textbf{union-2}). This filter is particularly useful to write
the alternative of two (or more) \emph{pattern} filters, making it possible to
conditionally continue a computation based on the shape of the input.

The \emph{composition} allows us to pass the result of $f_1$
as input to $f_2$. The composition filter is of paramount
importance. Indeed, without it, our only way to iterate (deconstruct)
an input value is to use a \emph{product} filter, which always rebuilds
a pair as result.

Finally, a \emph{recursive} filter is evaluated by recording its body in
$\delta$ and evaluating it (\textbf{rec}),  while for a \emph{recursive call}
%(the only rule that does not require a value) 
we replace the recursion variable by its definition (\textbf{rec-call}).

This concludes the presentation of the semantics of non-declarative filters (\emph{ie}, without groupby and orderby). These form a Turing complete formalism (full proof in \appcrazy~\ref{label:tc}):

\begin{theorem}[Turing completeness]\label{th:tc}
The language formed by constants, variables, pairs, equality, and applications of non-declarative filters is Turing complete.
\end{theorem}
\paragraph{\it Proof \rm(sketch).} We can encode untyped call-by-value $\lambda$-calculus by first applying continuation passing style (CPS) transformations and encoding CPS term reduction rules and substitutions via filters. Thanks to CPS we eschew the restrictions on composition.\hfill\qed

\ifLONGVERSION%%%%%%%%%%%%%%%%%%%%%%
\subsection{Semantics of declarative filters}
\else
\medskip

\noindent
\fi%%%%%%%%%%%%%%%%%%%%%%
To conclude the presentation of the semantics we have to define the semantics of groupby and orderby. We prefer to give the semantics in a declarative form rather than operationally in order not to tie it to a particular order (of keys or of the execution):

\paragraph{Groupby:}
$\GROUPBY{f}$ applied to a sequence \LIST{$v_1$ \ldots{} $v_m$}
reduces to a sequence \LIST{ $(k_1,l_1)$ \ldots
  $(k_n, l_n)$ } such that:
\begin{enum}
\item $\forall i, ~1\leq i \leq m, ~\exists j,~1\leq j \leq n,
  ~\text{s.t. } k_j = f(v_i)$
\item $\forall j, ~1\leq j \leq n, ~\exists i,~1\leq i \leq m,
  ~\text{s.t. } k_j = f(v_i)$
\item $\forall j, ~1\leq j \leq n$, $l_j$ is a sequence: \LIST{
    $v^j_1$ \ldots{} $v^j_{n_j}$}
\item $\forall j, ~1\leq j \leq n$, $\forall k, ~1\leq k \leq n_j$,
  $f(v^j_k)=k_j$
\item $k_i=k_j\Rightarrow i=j$
\item $l_1$, \ldots, $l_n$ is a partition of \LIST{$v_1$ \ldots{} $v_m$ }
\end{enum}
% Let \texttt{Transform} be the filter defined in Section~\ref{builtin}. If \texttt{Transform[$\fpat{x}{((fx),x)}$]} applied to $v$ reduces to $v_1$, then $\GROUPBY{f}$ applied to $v$ reduces to $v_2$ such that:
% \begin{enumerate}
% \item $\forall(k,v)\in v_1. \exists(k'l')\in v_2. k=k'\wedge v\in l'$
% \item $\forall (k',l')\in v_2.\forall v'\in l'. (k',v')\in v_1$
% \item if $(k_1,l_1)$ and $(k_2,l_2)$ are distinct elements of $v_2$, then $k_1\not=k_2$.
% \end{enumerate}

\paragraph{Orderby:} 
$\ORDERBY{f}$ applied to $\LIST{v_1 \ldots{} v_n}$ reduces to $\LIST{v'_1 \ldots{} v'_n}$ such
that:
\begin{enum}
\item  $\LIST{v'_1 \ldots{} v'_n}$ is a permutation of  $\LIST{v_1 \ldots{} v_n}$,
\item $\forall i,j ~\text{s.t.} ~1\leq i \leq j \leq n, f(v_i) \leq f(v_j)$
\end{enum}
Since the semantics of both operators is deeply connected to a notion
of equality and order on values of the host language, we give them as
``built-in'' operations. However we will illustrate how our type
algebra allows us to provide very precise typing rules, specialized
for their particular semantics. It is also possible to encode
co-grouping (or $\GROUPBY{}$ on several input sequences) with a
combination of $\GROUPBY{}$ and filters
(\emph{cf}.\ \appcrazy~\ref{sec:cogroupingencoding}). 

\ifLONGVERSION%%%%%%%%%%%%%%%%%%%%%%
\subsection{Syntactic sugar}
Now that we have formally defined the semantics of filters we can use
them to introduce some syntactic sugar.

\paragraph{Expressions.} 
\else
\paragraph{Syntactic sugar.}
\fi%%%%%%%%%%%%%%%%%%%%%%
\label{sec:sugar}
The reader may have noticed that the productions for expressions (Definition~\ref{def:exp}) do not define any destructor (\emph{eg}, projections, label selection, \ldots), just
constructors. The reason is that destructors, as well as other common expressions, can be encoded by
filter applications:
\begin{tabbing}
$\ITE{e}{e_1}{e_2}$ ~ \= ~$\eqdef$~ \= ~\kill
$e.\ell$  \> $\eqdef$ \> $(\fpat{\patorec{\ell\col x}}{x})e$\\[.1mm]
$\fst(e)$ \> $\eqdef$ \> $(\fpat{\patpair{x}{\tany}}{x})e$\\[.1mm]
$\snd(e)$ \> $\eqdef$ \> $(\fpat{\patpair{\tany}{x}}{x})e$\\[.1mm]
$\LET{p}{e_1}{e_2}$ \> $\eqdef$ \> $(\fpat{p}{e_2})e_1$\\ [.1mm]
$\ITE{e}{e_1}{e_2}$ \> $\eqdef$ \> $(\pator{\fpat{\true}{e_1}}{\fpat{\false}{e_2}})e$\\[.6mm]
$\MATCH{e}{\fpat{p_1}{e_1}\texttt{|\ldots|}\fpat{p_n}{e_n}}$ \\[-.5mm]
 \>$\eqdef$\> $(\pator{\fpat{p_1}{e_1}}{\pator{\ldots}{\fpat{p_n}{e_n}}})e$
\ignore{
$\MATCH{e}{}$\\
\hspace*{0.7cm}$p_1 \texttt{ -> } e_1$\\
\hspace*{0.4cm}$\texttt{\ \ }$\vdots\\
\hspace*{0.4cm}$\texttt{| }p_n\texttt{ -> } e_n$ \> $\eqdef$ \> $(\pator{\fpat{p_1}{e_1}}{\pator{\ldots}{\fpat{p_n}{e_n}}})e$
}
\end{tabbing}
These are just a possible choice, but others are possible. For instance in Jaql dot selection is overloaded: when \texttt{\_.$\ell$} is applied to a record, Jaql returns the content of its $\ell$ field; if the field is absent or the argument is \Null{}, then Jaql returns \Null{} and fails if the argument is not a record; when applied to a list (`array' in Jaql terminology) it recursively applies to all the elements of the list. So Jaql's ``$\_.\ell$'' is precisely defined as \\[1.5mm]
\centerline{$\frec{X}{(\fpat{\patorec{\ell{\col} x}}{x}\texttt{\;\pmb|\;}\fpat{(\patrec{\textbf{..}}\pmb{|}\Null)}{\Null}\texttt{\;\pmb|\;}{\fpat{\patpair{h}{t}}{\fprod{Xh}{Xt\,}}})}$}
\\[1.5mm]
\ifCONSTANTPATTERN%%%%%%%%%
or, equivalently, $\mu X.(\;\fpat{\patorec{\ell\col x}}{x}\texttt{\pmb|}\fpat{\patrec{\textbf{..}}}{\Null}\texttt{\pmb|}{\fprod{X}{X}}\;)$ where $X$ is syntactic sugar for  $\fpat{x}{Xx}$.

\paragraph{Types and patterns.}
Syntactic sugar can also be defined for types, in particular to denote optional record fields whose name, following Jaql notation, is suffixed by a question mark:
\begin{displaymath}
\begin{array}{r@{\quad\eqdef\quad}l}
\patrec{\ell\whycol t} & \patrec{\ell\col \pator{t}{\bot}}\\
\patrec{\ell\col  x\texttt{ else } c , m\col p} & \patand{(\pator{\patrec{\ell\col x}}{(x:=c)})}{\patrec{m\col p}}\\
\end{array}
\end{displaymath}
The first encoding is used to specify the type of optional fields: if
an expression has type \texttt{\{name\col string , age\why\col\Int\}} then it will
return a record with a name field and an optional age field that is either of type \Int, or is undefined as specified by the constant $\bot$ (see Appendix~\ref{tsum} for the formal definition of $\bot$).
present, will have type \Int. The second syntactic sugar is to work with such optional fields. For instance
$$\LET{\patrec{\texttt{name} \col  n , \texttt{age} \col  a \texttt{ else }99}}{\textit{person}}{e}$$
executes $e$ in an environment where $n$ is
bound to the name of \textit{person} and $a$ is bound to its age, or to
$99$ if the age field is absent.
\fi%%%%%%%%%%%%%%%%%%%%%%%%%%%%%%%%%%%%%%%%%%%%%%%%%%%%%%%%%%%
Besides the syntactic sugar above, in the next section we will use
$t_1+t_2$ to denote the record type formed by all field types in $t_2$
and all the field types in $t_1$ whose label is not already present in
$t_2$. Similarly $t\setminus\ell$ will denote the record types formed
by all field types in $t$ apart from the one labelled by $\ell$, if
present. Finally, we will also use for expressions, types, and patterns the syntactic sugar for lists
used in the introduction. So, for instance, $\LIST{p_1\ p_2\ ...\ p_n}$ is matched by lists of $n$ elements provided that their $i$-th element matches $p_i$.

\section{Type inference}
\label{label:typing}
\ifLONGVERSION
In this section we describe a type inference algorithm for our expressions.

\subsection{Typing of simple and foreign expressions}
Variables, constants, and pairs are straightforwardly typed by
% the following rules:
\\[2mm]
%\begin{center}
\hspace*{-3mm}
\begin{tabular}{l}
  {\small\sc [Vars]}\\
  \AxiomC{\BLANK}
  \UnaryInfC{$\Gamma\vdash x:\Gamma(x)$}
\DisplayProof
\end{tabular}
%\hspace*{-3mm}
\quad
\begin{tabular}{l}
  {\small\sc [Constant]}\\
  \AxiomC{\BLANK}
  \UnaryInfC{$\Gamma\vdash c:c$}
\DisplayProof
\end{tabular}
%\hspace*{-3mm}
\quad
\begin{tabular}{l}
  {\small\sc [Prod]}\\
  \AxiomC{$\Gamma\vdash e_1:t_1~~
           \Gamma\vdash e_2:t_2$}
  \UnaryInfC{$\Gamma\vdash (e_1,e_2):\patpair{t_1}{t_2}$}
  \DisplayProof
\end{tabular}\\[2mm]
where $\Gamma$ denotes a typing environment that is a function from
expression variables to types and $c$ denotes both a constant and the singleton type containing the constant.
Expressions of the host language are typed by the \texttt{type} function
which given a type environment and a foreign expression returns the
type of the expression, and that we suppose to be given for each host
language.\footnote{The function \texttt{type} must be able to handle
  type environments with types of our system. It can do it either by
  subsuming variable with specific types to the types of the host
  language (\emph{eg}, if the host language does not support singleton
  types then the singleton type \texttt{3} will be subsumed to \Int)
  or by typing foreign expressions by using our types.}
\\[2mm]
\hspace*{-2mm}\begin{tabular}{l}
  {\small\sc [Foreign]}\\
  \AxiomC{$\Gamma\vdash e_1:t_1\quad\cdots{}\quad\Gamma\vdash
    e_n:t_n$} \UnaryInfC{$\Gamma\vdash \textsf{op}(e_1{,}...,e_n):
    \texttt{type}((\Gamma,x_1{:}t_1,...,x_n{:}t_n),
    \textsf{op}(x_1{,}...,x_n))$} \DisplayProof
\end{tabular}\\[2mm]
Since the various $e_i$ can contain filter applications, thus unknown to the host language's type system, the rule  {\sc [Foreign]} swaps them with variables having the same type.

Notice that our expressions, whereas they include \emph{filter} applications,
do not include applications of expressions to expressions. Therefore if
the host language provides function definitions, then the applications of
the host language must be dealt as foreign expressions, as well (\emph{cf.} the expression operator \textsf{apply} in Section~\ref{sec:types}).
\else%%%%%%%%%%%%%%%%%%%%%%%%%%%%%%%%%%%%%%%%%%%%%%%%
\rechanges{The type inference system assign types to \emph{expressions}}. Variables, constants, and pairs are typed by standard rules, while we suppose that the typing of foreign expressions is provided by the host language.\footnote{\rechanges{Notice that our expressions, whereas they include \emph{filter} applications,
do not include applications of expressions to expressions. Therefore if
the host language provides function definitions, then the applications of
the host language must be dealt as foreign expressions, as well (\emph{cf.} \textsf{apply} in \S\ref{sec:types}).}} So we omit the corresponding rules (they are in \appcrazy~\ref{sec:foreign}). The core of our type system starts with records.
\fi%%%%%%%%%%%%%%%%%%%%%%%%%%%%%%%%%%%%%%%%%%%%%%%%%%

\subsection{Typing of records}\label{se:tyrcd}

The typing of records is novel and challenging because record
expressions may contain string \emph{expressions} in label position, such as
in $\patrec{e_1\col e_2}$, while in all type systems for record we are
aware of, labels are never computed. It is difficult to give a type to
$\patrec{e_1\col e_2}$ since, in general, we do not statically know
the value that $e_1$ will return, and which is required to form a
record type. All we can (and must) ask is that this value will be a
string. To type a record expression $\patrec{e_1\col e_2}$, thus, we distinguish two cases according to whether the type $t_1$
of $e_1$ is finite (\emph{ie}, it contains only finitely many values,
such as, say, \texttt{Bool}) or not. If a type is finite, (finiteness of
regular types seen as tree automata can be decided
in polynomial time \cite{tata07}), then it is possible to
write it as a finite union of values (actually, of singleton types). So
consider again $\patrec{e_1\col e_2}$ and let $t_1$ be the type of $e_1$
and $t_2$ the type of $e_2$. First, $t_1$ must be a subtype of
\String{} (since record labels are strings). So if $t_1$ is
finite it can be expressed as $\pator{\ell_1}{\pator{\cdots}{\ell_n}}$
which means that $e_1$ will return the string $\ell_i$ for some
$i\in[1..n]$. Therefore $\patrec{e_1\col e_2}$ will have type
$\patrec{\ell_i:t_2}$ for some $i\in[1..n]$ and, thus, the union of
all these types, as expressed by the rule \textsc{[Rcd-Fin]} below. If $t_1$
is infinite instead, then all we can say is that it will be a record with
some (unknown) labels, as expressed by rule
\textsc{[Rcd-Inf]}.\looseness -1
\begin{center}
\begin{tabular}{l}
  {\small\sc [Rcd-Fin]}\\
  \AxiomC{$\Gamma\vdash e:\pator{\ell_1}{\pator{\cdots}{\ell_n}}$}
  \AxiomC{$\Gamma\vdash e':t$}
  \BinaryInfC{$\Gamma\vdash\{e\col e'\}:\pator{\patrec{\ell_1\col t}}{\pator{\cdots}{\patrec{\ell_n\col t}}}$}
  \DisplayProof
\end{tabular}
\end{center}

\begin{center}
\hspace*{18mm}
\begin{tabular}{l}
  {\small\sc [Rcd-Inf]}\\
  \AxiomC{$\Gamma\vdash e:t$}
  \RightLabel{\small$\begin{array}{c}
                  t\leq\String\\
                  t \text{ is infinite}
                \end{array}
  $}
  \AxiomC{$\Gamma\vdash e':t'$}
  \BinaryInfC{$\Gamma\vdash\{e\col e'\}:\patrec{\textbf{..}}$}
  \DisplayProof
\end{tabular}
\end{center}

\begin{center}
\begin{tabular}{l}
  {\small\sc [Rcd-Mul]}\\
  \AxiomC{$\Gamma\vdash\{e_1\col e'_1\}:t_1\quad\cdots\quad\Gamma\vdash\{e_n\col e'_n\}:t_n$}
\UnaryInfC{$\Gamma\vdash\{e_1\col e'_1,\ldots,e_n\col e'_n\}:t_1+\cdots+t_n$}
\DisplayProof
\end{tabular}
\end{center}

%\begin{center}
\hspace*{-7mm}
\begin{tabular}{l}
  {\small\sc [Rcd-Conc]}\\
  \AxiomC{$\Gamma\vdash e_1:t_1$\quad $\Gamma\vdash e_2:t_2$}
  \RightLabel{\small\hspace*{-2mm}$\begin{array}{l}t_1\leq\patrec{\textbf{..}}\\t_2\leq\patrec{\textbf{..}}\end{array}$}
  \UnaryInfC{$\Gamma\vdash e_1+e_2:t_1+t_2$}
  \DisplayProof
\end{tabular}
%% \end{center}
%% \begin{center}
\hspace*{-2mm}
\begin{tabular}{l}
  {\small\sc [Rcd-Del]}  \\
  \AxiomC{$\Gamma\vdash e:t$}
  \RightLabel{\small$t\leq\patrec{\textbf{..}}$}
  \UnaryInfC{$\Gamma\vdash e\setminus\ell:t\setminus\ell$}
  \DisplayProof
\end{tabular}
%\end{center}
\smallskip\\{\color{orange}
Records with multiple fields are handled by the rule
\textsc{[Rcd-Mul]} which ``merges'' the result of typing single fields
by using the type operator $+$ as defined in \cduce~\cite{BCF03,alainthesis}, which is a right-priority record
concatenation defined to take into account undefined and unknown fields: for instance,
$\patrec{a\col\Int,b\col\Int}+\patrec{a\texttt{?}\col\Bool}= \patrec{a\col\pator{\Int}{\Bool},b\col\Int}$; unknown fields in the right-hand side may override known fields of the left-hand side, which is why, for instance, we have
$\patrec{a\col\Int,b\col\Bool}+\patorec{b\col\Int}=\patorec{b\col\Int}$; likewise, for every record type $t$  (\emph{ie}, for every $t$ subtype
of $\patrec{\textbf{..}}$) we have $t+\patrec{\textbf{..}}=\patrec{\textbf{..}}$.}
 Finally, \textsc{[Rcd-Conc]} and
\textsc{[Rcd-Del]} deal with record concatenation and field deletion,
respectively, in a straightforward way: the only constraint is that all
expressions must have a record type (\emph{ie}, the constraints of the form $...\leq\patrec{\textbf{..}}$).
See \appcrazy~\ref{tsum}
for formal definitions of all these type operators. 

Notice that these rules do not ensure that a record will not have two
fields with the same label, which is a run-time error. Detecting such an error
needs sophisticated type systems (\emph{eg}, dependent types) beyond the scope of this work. This is why in the rule \textsc{[Rcd-Mul]} we used
type operator ``$+$'' which, in case of multiple occurring labels, since records are
unordered, corresponds to randomly choosing one of the types bound
to these labels: if such a field is selected, it would yield a run-time
error, so its typing can be ambiguous. We can fine tune the rule
\textsc{[Rcd-Mul]} so that when all the $t_i$ are finite unions of
record types, then we require to have pairwise disjoint sets of
labels; but since the problem would still persist for infinite types
we prefer to retain the current, simpler formulation.

\subsection{Typing of filter application}\label{sec:tfp}
\rechanges{Filters are not first-class: they can be applied but not passed around or computed. Therefore we do not assign types to filters but, as for any other expression, we assign types to \emph{filter applications}.} The typing rule for filter application
\begin{center}
\hspace*{12mm}
\begin{tabular}{l}
  {\small\sc [Filter-App]}\\
  \AxiomC{$\Gamma\vdash e:t$}
  \AxiomC{\ftypej{$\Gamma\sep\emptyset\sep \emptyset}{f}{t}{s}$}
  \BinaryInfC{$\Gamma\vdash fe:s$}
  \DisplayProof
\end{tabular}
\end{center}
relies on an auxiliary deduction system for judgments of the form
$\ftypej{\Gamma\sep\Delta\sep M}{f}{t}{s}$ that states that if in the
environments $\Gamma,\Delta, M$ (explained later on) we apply the
filter $f$ to a value of type $t$, then it will return a result of
type $s$.

To define this auxiliary deduction system, which is the core of our type analysis, we first need to define $\ftype{f}$, the type accepted by a filter $f$. Intuitively, this type
gives a necessary condition on the input for the filter not to fail:
\begin{definition}[Accepted type]
Given a filter $f$, the \emph{accepted type} of $f$, written $\ftype{f}$ is
the set of values defined by:
\begin{displaymath}
\hspace*{-1mm}
\begin{array}{l}
\begin{array}{l@{~~=~~}l}
  \ftype{e} & \tany\\
  \ftype{\fpat{p}{f}} & \patand{\ftype{p}}{~\ftype{f}} \\
  \ftype{\falt{f_1}{f_2}} & \pator{\ftype{f_1}}{\ftype{f_2}}\\
  \ftype{\fprod{f_1}{f_2}} & \patpair{\ftype{f_1}}{\ftype{f_2}}\\
  \ftype{\fseq{f_1}{f_2}} & \ftype{f_1}
\end{array}\qquad
\begin{array}{l@{~~=~~}l}
  \ftype{X a} & \tany\\
  \ftype{\frec{X}{f}} & \ftype{f}\\
  \ftype{\GROUPBY{f}} & \LIST{\tany\ks} \\
  \ftype{\ORDERBY{f}} & \LIST{\tany\ks}
\end{array}\\
\begin{array}{l@{~~=~~}l}
  \ftype{\forec{\ell_1{\col} f_1{,}..,\ell_n{\col} f_n}} & \patorec{\ell_1{\col}\ftype{f_1}{,}..,\ell_n{\col}\ftype{f_2}}\\
\end{array}
\end{array}
\end{displaymath}
\end{definition}\noindent
It is easy to show that an argument included in the accepted type is a necessary (but not sufficient, because of the cases for composition and recursion) condition
for the evaluation of a filter not to fail:
\begin{lemma}
\label{lem:omega}
Let $f$ be a filter and $v$ be a value such that $v\notin\ftype{f}$.
For every $\gamma$, $\delta$,
if $\japp{\delta\textbf{;}\gamma}{f}{v}{r}$, then $r\equiv\Omega$.
\end{lemma}
\ifLONGVERSION%%%%%%%%%%%%
The proof is a straightforward induction on the structure of the
derivation, and is detailed in Appendix~\ref{prf:subred}.
\fi%%%%%%%%%%%%%%%%%%%%%%%%%%%%%%%%%%%%%%%%%%%%
The last two auxiliary definitions we need are related to product
and record types. In the presence of unions, the most
general form for a product type is a finite union of products (since intersections distribute on products). For
instance consider the type\\
\centerline{$\pator{\patpair{\Int}{\Int}}{\patpair{\String}{\String}}$}
This type denotes the set of pairs for which either both projections
are $\Int$ or both projections are $\String$. A type such as\\
\centerline{$\patpair{\pator{\Int}{\String}}{\pator{\Int}{\String}}$}
is less precise, since it also allows pairs whose first projection is an
$\Int$ and second projection is a $\String$ and \emph{vice
  versa}. We see that it is necessary to manipulate finite
unions of products (and similarly for records), and therefore, we
introduce the following notations:
\begin{lemma}[Product decomposition]
Let $t\in\Types$ such that $t\leq \patpair{\ANY}{\ANY}$. A
\emph{product decomposition} of $t$, denoted by $\pdec(t)$ is a set of
types:\\[-1mm]
%\begin{displaymath}
\centerline{$
\pdec(t) = \{ \patpair{t^1_1}{t^1_2}, \ldots, \patpair{t^n_1}{t^n_2}\}
$}\\[2mm]
%\end{displaymath}
such that $t = \bigvee_{t_i\in\pdec(t)}t_i$. For a given
product decomposition, we say that $n$ is the \emph{rank} of $t$,
noted $\prank(t)$, and use the notation $\pdec_i^j(t)$ for the type
$t_i^j$.
\end{lemma}
There exist several suitable decompositions whose details are out of
the scope of this paper. We refer the interested reader to
\cite{alainthesis} and \cite{kimthesis} for practical algorithms
that compute such decompositions for any subtype of
$\patpair{\ANY}{\ANY}$ or of $\patrec{\pmb{..}}$.
These notions of decomposition, rank and projection can be generalized
to records:
\begin{lemma}[Record decomposition]\label{le:rcddec}
Let $t\in\Types$ such that $t\leq \patrec{..}$. A
\emph{record decomposition} of $t$, denoted by $\rdec(t)$ is a finite set of
types $\rdec(t){=}\{r_1,\ldots,r_n\}$ where each $r_i$ is either of the form $\patrec{\ell^i_1\col t^i_1,\ldots, \ell^i_{n_i}\col
  t^i_{n_i}}$ or of the form  $\patorec{\ell^i_1\col t^i_1,\ldots, \ell^i_{n_i}\col
  t^i_{n_i}}$ and such that $t = \bigvee_{r_i\in\rdec(t)}r_i$.
For a given
record decomposition, we say that $n$ is the \emph{rank} of $t$, noted $\prank(t)$, and use the notation $\rdec_{\ell}^j(t)$ for the
type of label $\ell$ in the $j^{\text{th}}$ component of $\rdec(t)$.
\end{lemma}
In our calculus we have three different sets of variables. The set
\Vars of term variables, ranged over by $x, y, ...$, introduced in
patterns and used in expressions and in arguments of calls of recursive
filters. The set \RVars of term recursion variables, ranged over by
$X, Y, ...$ and that are used to define recursive filters. The set
\TVars of type recursion variables, ranged over by $T, U, ...$ used to
define recursive types. In order to use them we need to define three
different environments: $\Gamma:\Vars\to\Types$ denoting \emph{type
  environments} that associate term variables with their types;
$\Delta:\RVars\to\Filters$ denoting \emph{definition environments}
that associate each filter recursion variable with the body of its
definition; $M:\RVars\times\Types\to\TVars$ denoting \emph{memoization
  environments} which record that the call of a given recursive
filter on a given type yielded the introduction of a fresh recursion
type variable. Our typing rules, thus work on judgments of the form
$\Gamma\sep\Delta\sep M\vdash f(t):t'$ stating that applying $f$ to an
expression of type $t$ in the environments $\Gamma,\ \Delta,\ M$
yields a result of type $t'$. This judgment can be derived with the
set of rules given in Figure~\ref{reffig:typfilter}.

\begin{figure*}[ht!]
\small\vspace{-2mm}
\begin{tabular}{l}
{\small\sc [Fil-Expr]}\\
\AxiomC{\BLANK}
\UnaryInfC{$\ftypej{\Gamma\sep\Delta\sep M}{e}{t}{\tyof{\Gamma,e}}$}
\DisplayProof
\end{tabular}%\\[2mm]
\quad
\begin{tabular}{l}
{\small\sc[Fil-Pat]}\\
\AxiomC{$\ftypej{\Gamma\cup t/p\sep\Delta\sep M}{f}{t}{s}$}
\RightLabel{\scriptsize$t\leq\patand{\ftype{p}}{\ftype{f}}$}
\UnaryInfC{$\ftypej{\Gamma\sep\Delta\sep M}{\fpat{p}{f}}{t}{s}$}
\DisplayProof
\end{tabular}
\begin{tabular}{l}
{\small\sc[Fil-Prod]}\\
\AxiomC{$i{=} 1..\prank(t),\ 
    j {=} 1,2\qquad\ftypej{\Gamma\sep\Delta\sep M}{f_j}{\pdec_j^i(t)}{s_j^i}$}
\UnaryInfC{$\ftypej{\Gamma\sep\Delta\sep M}{\fprod{f_1}{f_2}}{t}{\displaystyle\bigvee_{i=1..\prank(t)}\patpair{s_1^i}{s_2^i}}$}
\DisplayProof
\end{tabular}
\\%%%%%%%%%%%%%%%%%%%%%%%%
\begin{tabular}{l}
{\small\sc [Fil-Rec]}\\
%\RightLabel{\scriptsize$r_i=\patorec{\ell_1\!\col\!s_1^i,\ldots,\ell_m\!\col\! s_m^i}$}
\AxiomC{${\footnotesize i{=} 1..\prank(t),\ j{=} 1..m}\qquad\ftypej{\Gamma\sep\Delta\sep M}{f_j}{\rdec_{\ell_j}^i(t)}{s_j^i}$}
\UnaryInfC{$\ftypej{\Gamma\sep\Delta\sep
    M}{\forec{\ell_1\!\col\!f_1,\ldots,\ell_m\!\col\!f_m}}{t}{\displaystyle\bigvee_{i=1..\prank(t)}\!\!\!\!\patorec{\ell_1\!\col\!s_1^i,\ldots,\ell_m\!\col\! s_m^i}}
%    r_i}
$}
%\noLine
\DisplayProof
\end{tabular}%\\[2mm]
\qquad
\begin{tabular}{l}
{\small\sc [Fil-Union]}\\
\AxiomC{{\small $i = 1,2$
}\quad$\ftypej{\Gamma\sep\Delta\sep M}{f_i}{t_i}{s_i}$}
\RightLabel{\scriptsize $\begin{array}{l}
t\leq \pator{\ftype{f_1}}{\ftype{f_2}}\\
t_1 = \patand{t}{\ftype{f_1}}\\
t_2 = \patand{t}{\lnot{\ftype{f_1}}}
\end{array}$}
\UnaryInfC{$\ftypej{\Gamma\sep\Delta\sep M}{\falt{f_1}{f_2}}{t}{\displaystyle\bigvee_{\{i|s_i\neq\tempty\}}\!\!\!\!s_i}$}
\DisplayProof
\end{tabular}%\\
\\%%%%%%%%%%%%%%%%%%%%%%%%%%%%%%%%%%%
\begin{tabular}{l}
{\small\sc[Fil-Comp]}\\
\AxiomC{$\ftypej{\Gamma\sep\Delta\sep M}{f_1}{t}{s}$}
\AxiomC{$\ftypej{\Gamma\sep\Delta\sep M}{f_2}{s}{s'}$}
\BinaryInfC{$\ftypej{\Gamma\sep\Delta\sep M}{\fseq{f_1}{f_2}}{t}{s'}$}
\DisplayProof
\end{tabular}%\\[2mm]
\qquad\qquad
\begin{tabular}{l}
{\small\sc [Fil-Fix]}\\
  \AxiomC{$\ftypej{\Gamma\sep\Delta,(X\mapsto f)\sep M,((X,t)\mapsto T)}{f}{t}{s}$}
  \RightLabel{\scriptsize$T$~fresh}
  \UnaryInfC{$\ftypej{\Gamma\sep\Delta\sep M}{(\frec{X}{f})}{t}{\frec{T}{s}}$}
  \DisplayProof
\end{tabular}\\[2mm]
\\%%%%%%%%%%%%%%%%%%%%%%%%%%%%%%%%%%
\begin{tabular}{l}
{\small\sc [Fil-Call-New]}\\
  \AxiomC{$\ftypej{\Gamma\sep\Delta\sep M,((X,t)\mapsto T)}{\Delta(X)}{t}{t'}$}
  \RightLabel{\scriptsize$%\left(
                 \begin{array}{c}
                  t=\tyof{\Gamma}{a}\\
                  (X,t)\not\in\dom(M)\\
                  T \text{~fresh}
                \end{array}
                % \right)
                $}
  \UnaryInfC{$\ftypej{\Gamma\sep\Delta\sep M}{(Xa)}{s}{\frec{T}{t'}}$}
  \DisplayProof
\end{tabular}%\\[2mm]
\qquad\qquad\quad
\begin{tabular}{l}
{\small\sc [Fil-Call-Mem]}\\
  \AxiomC{\BLANK}
  \RightLabel{\scriptsize$%\left(
                  \begin{array}{c}
                  t=\tyof{\Gamma}{a}\\
                  (X,t)\in\dom(M)
                \end{array}
                %\right)
                $}
  \UnaryInfC{$\ftypej{\Gamma\sep\Delta\sep M}{(Xa)}{s}{M(X,t)}$}
  \DisplayProof
\end{tabular}\\[2mm]
\\%%%%%%%%%%%%%%%%%%%%%%%%%%%%%%%%%%
\begin{tabular}{l}{\small\sc [Fil-OrdBy]}\\
\AxiomC{$\ftypej{\forall t_i\in\ITEM{t}\quad \Gamma\sep\Delta\sep M}{f}{t_i}{s_i}$}
\RightLabel{\scriptsize$\begin{array}{l}t\leq\LIST{\tany\ks}\\\bigvee_i\!\!s_i\text{ is ordered}\end{array}$}
\UnaryInfC{$\ftypej{\Gamma\sep\Delta\sep M}{(\ORDERBY{f})}{t}{\TORDERBY{(t)}}$}
\DisplayProof
\end{tabular}
\qquad\qquad
\begin{tabular}{l}{\small\sc [Fil-GrpBy]}\\
\AxiomC{$\ftypej{\forall t_i\in\ITEM{t}\quad \Gamma\sep\Delta\sep M}{f}{t_i}{s_i}$}
  \RightLabel{\scriptsize$t\leq\LIST{\tany\ks}$}
\UnaryInfC{$\ftypej{\Gamma\sep\Delta\sep M}{(\GROUPBY{f})}{t}{\LIST{\fprod{(\bigvee_i s_i)}{\TORDERBY{(t)}}\ks}}$}
\DisplayProof
\end{tabular}%\\2mm]
\caption{Type inference algorithm for filter application}
\label{reffig:typfilter}\vspace{-3mm}
\end{figure*}

These rules are straightforward, when put side by side with the
dynamic semantics of filters, given in
Section~\ref{label:dynamic}. It is clear that this type system
simulates \emph{at the level of types} the computations that are carried
out by filters on values at runtime. For instance, rule
\RSTYLE{Fil-Expr} calls the typing function of the host language to
determine the type of an expression $e$. Rule \RSTYLE{Fil-Prod}
applies a product filter recursively on the first and second
projection for each member of the product decomposition of the input
type and returns the union of all result types. Rule \RSTYLE{Fil-Rec}
for records is similar, recursively applying sub-filters label-wise
for each member of the record decomposition and returning the union of
the resulting record types. As for the pattern filter (rule
\RSTYLE{Fil-Pat}), its subfilter $f$ is typed in the environment
augmented by the mapping $t/p$ of the input type against the pattern (\emph{cf.} Theorem~\ref{typenv}).
The typing rule for the union filter, \RSTYLE{Fil-Union} reflects the
first match policy: when typing the second branch, we know that the
first was not taken, hence that at runtime the filtered value will
have a type that is in $t$ but not in $\ftype{f_1}$. 
{\color{orange} Notice that this is \emph{not} ensured by the
definition of accepted type ---which is a rough approximation that
discards grosser errors but, as we stressed right after its
definition, is not sufficient to ensure that evaluation of $f_1$ will
not fail--- but by the type system itself: the premises \emph{check}
that $f_1(t_1)$ is well-typed which, by induction, implies that $f_1$
will never fail on values of type $t_1$ and, ergo, that these values
will never reach $f_2$}. 
Also, we discard
from the output type the contribution of the branches that cannot be
taken, that is, branches whose accepted type have an empty intersection
with the input type $t$. Composition (rule \RSTYLE{Fil-Comp}) is
straightforward. In this rule, the restriction that $f_1$ is a filter
with no open recursion variable ensures that its output type $s$ is
also a type without free recursion variables and, therefore, that we can
use it as input type for $f_2$. The next three rules work
together. The first, \RSTYLE{Fil-Fix} introduces for a recursive filter
a fresh recursion variable for its output type. It also memoize in
$\Delta$ that the recursive filter $X$ is associated with a body $f$
and in $M$ that for an input filter $X$ and an input type $t$, the
output type is the newly introduced recursive type variable.
When dealing with a recursive call $X$ two situations may arise. One
possibility is that it is the first time the filter $X$ is applied to
the input type $t$. We therefore introduce a fresh type variable $T$
and recurse, replacing $X$ by its definition $f$. Otherwise, if the
input type has already been encountered while typing the filter
variable $X$, we can return its memoized type, a type variable $T$.
Finally, Rule \RSTYLE{Fil-OrdBy} and Rule \RSTYLE{Fil-GrpBy} handle
the special cases of $\GROUPBY$ and $\ORDERBY$ filters. Their typing
is explained in the following section.\looseness -1

\subsection{Typing of $\ORDERBY{}$ and $\GROUPBY{}$}
While the ``structural'' filters enjoy simple, compositional typing
rules, the ad-hoc operations $\ORDERBY{}$ and $\GROUPBY{}$ need
specially crafted rules. Indeed it is well known that when
transformation languages have the ability to compare data values
type-checking (and also type inference) becomes undecidable (\emph{eg}, see~\cite{Alon01,Alon03}). We therefore provide two typing approximations
that yield a good compromise between precision and
decidability.
%\vero{En quoi c'est un bon compromis, si on pouvait donner un exemple cela paraitrait moins incantatoire}
First we define an auxiliary function over sequence
types:
\begin{definition}[Item set]
Let $t\in\Types$ such that $t\leq\LIST{\tany\ks}$. The \emph{item set} of
$t$ denoted by $\ITEM{t}$ is defined by:
\begin{displaymath}
\begin{array}{l@{~=~}l}
\ITEM{\EMPTY} & \emptyset\\
\ITEM{t}    & \ITEM{\patand{t}{\fprod{\tany}{\tany}}} \hspace*{0.5cm}\text{if }t\not\leq\fprod{\tany}{\tany}\\
\ITEM{\displaystyle\bigvee_{\mathclap{1\leq i \leq
      \prank(t)}}\patpair{t^1_i}{t^2_i}} &
\displaystyle\bigcup_{\mathclap{1\leq i \leq
      \prank(t)}} (\{t^1_i\}\cup\ITEM{t^2_i})
\end{array}
\end{displaymath}
\end{definition}\noindent
The first and second line in the definition ensure that $\ITEM{}$
returns the empty set for sequence types that are not
products, namely for the empty sequence. The third line handles the
case of non-empty sequence type. In this case $t$ is a finite union
of products, whose first components are the types of the ``head'' of
the sequence and second components are recursively the types of the
tails. Note also that this definition is well-founded. Since types are
regular trees the number of distinct types accumulated by $\ITEM{}$ is
finite. We can now defined typing rules for the $\ORDERBY$ and
$\GROUPBY$ operators.

\paragraph{\ORDERBY{$f$}:} The \ORDERBY{} filter uses its argument
filter $f$ to compute a key from each element of the input sequence
and then returns the same sequence of elements, sorted with respect to
their key. Therefore, while the types of the elements in the result
are still known, their order is lost. We use $\ITEM$ to compute the
output type of an \ORDERBY{} application:\\
\centerline{$\TORDERBY(t) =
  \LIST{ (\displaystyle\bigvee_{\mathclap{t_i\in\ITEM{t}}}t_i)* }$}

\paragraph{\GROUPBY{$f$}:}
The typing of \ORDERBY can be used to give a rough approximation of the typing of \GROUPBY as stated by rule {\sc [Fil-GrpBy]}.
%  \centerline{$\TGROUPBY{(f,t)}= \LIST{\fprod{(\displaystyle\bigvee_{\mathclap{t'\in F}}t')}{\TORDERBY{(t)}}\ks}$}
% where
% \centerline{$    F = \displaystyle\bigcup_{\mathclap{\ftypej{\emptyset\sep \emptyset\sep
%           \emptyset}{f}{t_i}{t'_i}, \text{for
%         }t_i\in\ITEM{t}}} \{t'_i\}$}
%
In words, we obtain a list of pairs where the key component is the
result type of $f$ applied to the items of the sequence, and use
\TORDERBY{} to shuffle the order of the list. A far more precise
typing of \GROUPBY that keeps track of the relation between
list elements and their images via $f$ is given in
\appcrazy~\ref{groupby}.

\subsection{Soundness, termination, and complexity}\label{sec:sound}
The soundness of the type inference system is given by the property of subject reduction for filter application
\begin{theorem}[subject reduction]
\label{thm:subred}
If $\ftypej{\emptyset\sep\emptyset\sep\emptyset}{f}{t}{s}$, then for
all $v:t$,  $\japp{\emptyset\textbf{;}\emptyset}{f}{v}{r}$ implies
$r:s$.
\end{theorem}
\noindent whose full proof is given in \appcrazy~\ref{prf:subred}.
It is easy to write a filter for which the type inference algorithm, that  is the deduction of $\vdash_{\text{\tiny fil}}$, does not terminate: $\frec{X}{\fpat{x}{X\fprod{x}{x}}}$.
The deduction of $\ftypej{\Gamma\sep\Delta\sep M}{f}{t}{s}$ simulates
an (abstract) execution of the filter $f$ on the type $t$. Since
filters are Turing complete, then in general it is not possible to
decide whether the deduction of $\vdash_{\text{\tiny fil}}$ for a
given filter $f$ will terminate for every input type $t$. For this
reason we define a static analysis \chk{f} for filters that ensures
that if $f$ passes the analysis, then for every input type $t$ the
deduction of $\ftypej{\Gamma\sep\Delta\sep M}{f}{t}{s}$
terminates. For space reasons the formal definition of \chk{f} is
relegated to \appcrazy~\ref{label:termination}, but its behavior can be easily explained.
Imagine that a recursive filter $f$ is applied to some input type
$t$. The algorithm tracks all the recursive calls occurring in $f$;
next it performs one step of reduction of each recursive call by
unfolding the body; finally it checks in this unfolding that if a
variable occurs in the argument of a recursive call, then it is
bound to a type that is a subtree of the original type $t$. In other
words, the analysis verifies that in the execution of the derivation for
$f(t)$ every call to $s/p$ for some type $s$ and pattern $p$ always
yields a type environment where variables used in recursive calls are
bound to subtrees of $t$. This implies
that the rule \RSTYLE{Fil-Call-new} will always memoize for a
given $X$, types that are obtained from the arguments of the recursive
calls of $X$ by replacing their variables with a subtree of the original
type $t$ memoized by the rule \RSTYLE{Fil-Fix}. Since $t$ is regular,
then it has finitely many distinct subtrees, thus
\RSTYLE{Fil-Call-New} can memoize only finitely many distinct types,
and therefore the algorithm terminates.\looseness -1

More precisely, the analysis proceeds in two passes. In the first pass
the algorithm tracks all recursive filters and for each of them it
$(i)$ marks the variables that occur in the arguments of its recursive
calls, $(ii)$ assigns to each variable an abstract identifier
representing the subtree of the input type to which the variable will
be bound at the initial call of the filter, and $(iii)$ it returns the
set of all types obtained by replacing variables by the associated
abstract identifier in each argument of a recursive call. The last set
intuitively represents all the possible ways in which recursive calls
can shuffle and recompose the subtrees forming the initial input type.
The second phase of the analysis first
abstractly reduces by one step each recursive filter by applying it on
the set of types collected in the first phase of the analysis and then
checks whether, after this reduction, all the variables marked in the
first phase (\emph{ie}, those that occur in arguments of recursive
calls) are still bound to subtrees of the initial input type: if this
checks fails, then the filter is rejected.\looseness -1

It is not difficult to see that the type inference algorithm
converges if and only if for every input type there exists a integer
$n$ such that after $n$ recursive calls the marked variables are bound
only to subtrees of the initial input type (or to something that does not
depend on it, of course). Since deciding whether
such an $n$ exists is not possible, our analysis checks whether for all
possible input types a filter satisfies it for $n{=}1$, that is to say, that
at every recursive call its marked variables satisfy the property;
otherwise it rejects the filter.\looseness -1

\begin{theorem}[Termination]
If \chk{f}, then for every type $t$ the deduction of
$\ftypej{\Gamma\sep\emptyset\sep \emptyset}{f}{t}{s}$ is in
2-EXPTIME. Furthermore, if $t$ is given as a non-deterministic tree
automaton (NTA) then $\ftypej{\Gamma\sep\emptyset\sep
  \emptyset}{f}{t}{s}$ is in EXPTIME, where the size of the problem is
$|f|\times|t|$.
\end{theorem}
\noindent
(for proofs see \appcrazy~\ref{label:termination} for termination and \appcrazy~\ref{sec:complex} for complexity).
This complexity result is in line with those of similar
formalisms. For instance in \cite{Martens03}, it is shown that
type-checking non deterministic top-down tree transducers is in EXPTIME
when the input and output types are given by a NTA.

All filters defined in this paper (excepted those in \appcrazy~\ref{label:tc}) pass the analysis. As an example consider the filter \texttt{rotate} that applied to
a list returns the same list with the first element moved to the last
position (and the empty list if applied to the empty list):\\
\centerline{
$\frec{X}{\ (\ \falt{\fpat{\fprod{x}{\fprod{y}{z}}}{\fprod{y}{X\fprod{x}{z}}}\quad}{\quad\fpat{w}{w}}\ )}$}\\
The analysis succeeds on this filter. If we denote by $\iota_x$ the
abstract subtree bound to the variable $x$, then the recursive call
will be executed on the abstract argument
$\fprod{\iota_x}{\iota_z}$. So in the unfolding of the recursive call
$x$ is bound to $\iota_x$, whereas $y$ and $z$ are bound to two distinct
subtrees of $\iota_z$. The variables in the recursive call, $x$ and $z$,
are thus bound to subtrees of the original tree (even though the argument of the recursive call is \emph{not} a
subtree of the original tree), therefore the filter is
accepted . In order to appreciate the precision of the inference
algorithm consider the type
$\LIST{\Int\kss~~\Bool\kss}$, that is, the type of lists formed by some integers
(at least one) followed by some booleans (at least one). For the
application of \texttt{rotate} to an argument of this type our
algorithm \emph{statically} infers the most precise type, that is,
$\LIST{\Int\ks~~\Bool\kss~~\Int}$. If we apply it once more the inferred type is $\LIST{\Int\ks~~~\Bool\kss~~~\Int~~~\Int}\texttt{|}\LIST{\Bool\ks~~~\Int~~~\Bool}$.

Generic filters are Turing complete. However, requiring
that \chk{} holds ---meaning that the filter is typeable by our system---
restricts the expressive power of our filters by preventing them from
\emph{recomposing} a new value before doing a recursive call. For
instance, it is not possible to typecheck a filter which reverses the
elements of a sequence. Determining the exact class of transformations
that typeable filters can express is challenging. However it is
possible to show (\emph{cf.} \appcrazy~\ref{sec:comptttrl}) that typeable
filters are strictly more expressive than top-down tree transducers
with regular look-ahead, a formalism for tree transformations
introduced in \cite{Engelfriet76}. The intuition about this result can be conveyed by and example. Consider the tree:\\
\centerline{$a(u_1(\ldots(u_n())) v_1(\ldots(v_m())))$}
that is, a tree whose root is labeled $a$ with
two children, each being a monadic tree of height $n$ and $m$,
respectively. Then it is not possible to write a top-down tree
transducer with regular look-ahead that creates the tree\looseness -1\\
\centerline{$a(u_1(\ldots(u_n(v_1(\ldots v_m())))))$}
which is just the concatenation of the two children of the root, seen
as sequences, a transformation that can be easily programmed by typeable filters. The key difference in expressive power comes from the
fact that filters are evaluated with an \emph{environment} that binds
capture variables to sub-trees of the input. This feature is essential
to encode sequence concatenation and sequence flattening ---two pervasive operations when dealing with sequences--- that cannot be expressed  by top-down tree
transducers with regular look-ahead.\looseness -1

\section{Jaql}
\label{label:jaql}
In this Section, we show how filters can be used to capture some
popular languages for processing data on the Cloud. 
We consider Jaql~\cite{jaql}, a query language for JSON developed by
IBM. We give translation rules from a subset of Jaql into filters.
\begin{defn}[Jaql expressions]
We use the following simplified grammar for Jaql (where we distinguish simple expressions, ranged over by $e$, from ``core expressions'' ranged over by $k$).\vspace{-1.5mm}
\begin{jaql}[escapeinside='']
'\ibm$e$' ::= '\ibm c' '$\hfill (\text{\rm constants})$'
  |  '\ibm$x$' '$\hfill (\text{\rm variables})$'
  |  '\ibm\$' '$\hfill (\text{\rm current value})$'
  |  '\ibm[$e$,'...'\ibm, $e$]' '$\hfill (\text{\rm arrays})$'
  |  '\ibm\{ $e$:$e$,'...'\ibm, $e$:$e$ \}' '$\hfill (\text{\rm records})$'
  |  '\ibm$e$.l' '$\hfill (\text{\rm field~access})$'
  |  '\ibm{\sf op}($e$,'...'\ibm,$e$)' '$\hfill (\text{\rm function~call})$'
  |  '\ibm$e$ -> $k$' '$\hfill (\text{\rm pipe})$'
'\ibm$k$' ::= '\ibm filter' '$($\ibm each $x$' '$)?$' '\ibm$e$''$\hfill (\text{\rm filter})$'
  |  '\ibm transform ''$($''\ibm each $x$''$)?$' '\ibm$e$''$\hfill (\text{\rm transform})$'
  |  '\ibm expand' '$(($''\ibm each $x$''$)?$' '\ibm$e$''$)?$''$\hfill (\text{\rm expand})$'
  |  '\ibm group\;''$( ($''\ibm each$\;x$''$)?$' '\ibm by $x\,$=$\,e$' '$($\ibm as$\;x$''$)?)?\;$\ibm into$\;e$''$\hfill (\text{\rm grouping})$'
\end{jaql}\vspace{-1.5mm}
%  |   sort '$($'each '$x)?$' by '$e$' '$\hfill (sorting)$'
%  |   top '$e$' '$\hfill (top)$'
\end{defn}
\ifLONGVERSION%%%%%%%%%%%%%%%%%%%%%%%%%%%%%
\subsection{Built-in filters}
\fi%%%%%%%%%%%%%%%%%%%%%%%%%%%%%%%%%%%%
\label{builtin}
\noindent In order to ease the presentation we extend our syntax by adding ``filter definitions''  (already informally used in the introduction) to  filters and ``filter calls'' to expressions:
\begin{displaymath}
\quad\begin{array}{lr}
e\ \ ::=\ \ \LET{\FILTER F\LIST{F_1,\ldots,F_n}}{f}{e} & \textbf{(filter defn.)}\\[1mm]
f\ \ ::=\ \ F\LIST{f,\ldots,f}    & \textbf{(call)}
\end{array}
\end{displaymath}
where $F$ ranges over \emph{filter names}.
The mapping for most of the language we consider rely on the following
built-in filters.
\begin{tabbing}
LETT\=FI\=LTER \kill
$\LETT{\FILTER \Filter\LIST{F}}\textup{\texttt{\,=\,}}\frec{X}{}$\\
\>\>$\fpat{\NIL}{\NIL}$\\
\>\pmb{|}\>$\fpat{\patpair{\patpair{x}{~xs}}{tl}}{\fprod{X(x,xs)}{X(tl)}}$\\
\>\pmb{|}\>$\fpat{\patpair{x}{tl}}{\fseq{Fx\;}{(\falt{\fpat{\true}{\fprod{x}{~X(tl)}}}{\fpat{\false}{X(tl)}})}}$
\ifLONGVERSION\\[5mm]\else\\[1.8mm]\fi
$\LETT{\FILTER \Transform\LIST{F}}\textup{\texttt{\,=\,}}\frec{X}{}$\\
\>\>$\fpat{\NIL}{\NIL}$\\
\>\pmb{|}\>$\fpat{\patpair{\patpair{x}{~xs}}{tl}}{\fprod{X(x,xs)}{X(tl)}}$\\
\>\pmb{|}\>$\fpat{\patpair{x}{tl}}{\fprod{Fx}{~X(tl)}}$
\ifLONGVERSION\\[5mm]\else\\[1.8mm]\fi
$\LETT{\FILTER \Expand}\textup{\texttt{\,=\,}}\frec{X}{}$\\
\>\>$\fpat{\NIL}{\NIL}$\\
\>\pmb{|}\>$\fpat{\patpair{\NIL}{tl}}{X(tl)}$\\
\>\pmb{|}\>$\fpat{\patpair{\patpair{x}{~xs}}{tl}}{\fprod{x}{X(xs,tl)}}$
\end{tabbing}
\ifLONGVERSION%%%%%%%%%%%%%%%%%%%%%%%%%%%
\subsection{Mapping}
\else\noindent
\fi%%%%%%%%%%%%%%%%%%%%%%%%%%%%%%%%%%%%%%%
\label{sec:jaqltr}
Jaql expressions are mapped to our expressions as
follows  (where \texttt{\$} is a distinguished expression variable interpreting Jaql's \texttt{\ibm\$}):\\
%\begin{displaymath}
$\begin{array}{lcl}
    \SemParen{\texttt{c}} & = & c\\
    \SemParen{x} & = & x\\
    \SemParen{\texttt{\$}} & = & \textup{\texttt{\tt\$}}\\
    \!\left\llbracket{\ibm\jq!{!e_1 \jq!:! e'_1\jq!,! ...\jq!,! e_n\jq!:! e'_n\jq!}!}\right\rrbracket & = & \patrec{\SemParen{e_1}\col\SemParen{e'_1}, ..., \SemParen{e_n}\col\SemParen{e'_n}}\\
    \SemParen{e\jq!.l!}  & = & \SemParen{e}.l\\
    \SemParen{{\sf op}(e_1,...,e_n)} & = & {\sf op}(\SemParen{e_1},...,\SemParen{e_n})\\
    \!\left\llbracket{\ibm\jq![!e_1\jq!,...,! e_n\jq!]!}\right\rrbracket & = & (\SemParen{e_1},...(\SemParen{e_n},\NIL)...)\\
    \!\left\llbracket{\ibm e\;\jq!->!\;k}\right\rrbracket & = & \fseq{\SemParen{e}}{\SemParen{k}_{\textsf{F}}}\\[1mm]
  \end{array}$

%\end{displaymath}
\noindent Jaql core expressions are mapped to filters as follows:\\[1mm]
\hspace*{-1.5mm}
$\begin{array}{l@{}c@{\;}l}
  \SemParen{\jq!filter!~e}_{\textsf{F}}  &=&  \SemParen{\jq!filter each $!~e}_{\textsf{F}}\\
  \SemParen{\jq!filter each!~x~e}_{\textsf{F}}  &=& \Filter\LIST{\fpat{x}{\SemParen{e}}}\\
  \SemParen{\jq!transform!~e}_{\textsf{F}}  &=&  \SemParen{\jq!transform each $!~e}_{\textsf{F}}\\
  \SemParen{\jq!transform each!~x~e}_{\textsf{F}}  &=& \Transform\LIST{\fpat{x}{\SemParen{e}}}\\
  \SemParen{\jq!expand each!~x~e}_{\textsf{F}}  &=&  \SemParen{\jq!expand!}_{\textsf{F}}; \SemParen{\jq!transform!\;\jq!each!\;x~e}_{\textsf{F}}\\
  \SemParen{\jq!expand!}_{\textsf{F}}  &=& \Expand\\
  \SemParen{\jq!group into!~e}_{\textsf{F}}  &=&  \SemParen{\jq!group by!~y\jq!=true into!~e}_{\textsf{F}}\\
  \SemParen{\jq!group!\,\jq!by!\,y\jq!=!e_1\,\jq!into!\,e_2}_{\textsf{F}}  &=&  \SemParen{\jq!group!\;\jq!each $ by!\,y\jq!=!e_1\;\jq!into!\;e_2}_{\textsf{F}}\\
\multicolumn{3}{l}{\SemParen{\jq!group!\,\jq!each!\,x\,\jq!by!\,y\jq!=!e_1\,\jq!into!\,e_2}_{\textsf{F}} =}\\
   \multicolumn{3}{r}{\SemParen{\jq!group!\;\jq!each!~x~\jq!by!~y~\jq!=!~e_1~\jq!as $ into!~e_2}_{\textsf{F}}}\\
\multicolumn{3}{l}{\SemParen{\jq!group each!~x~\jq!by!~y~\jq!=!~e_1~\jq!as!~g~\jq!into!~e_2}_{\textsf{F}} =}\\
\multicolumn{3}{r}{\fseq{\GROUPBY{\fpat{x}{\SemParen{e_1}}}}{\Transform\LIST{\fpat{\fprod{y}{g}}{\SemParen{e_2}}}}\mbox{ }}
\end{array}$\\[1mm]
%% \begin{tabbing}
%%   $\SemParen{\jq!group by!~k~\jq!=!~e_1~\jq!into!~e_2}_{\textsf{F}}$ \= $=$ \= $\frec{X}{(}$\=\kill
%%   $\SemParen{\jq!filter!~e}_{\textsf{F}}$ \> $=$ \> $\SemParen{\jq!filter each $!~e}_{\textsf{F}}$\\
%%   $\SemParen{\jq!filter each!~x~e}_{\textsf{F}}$ \> $=$ \>$\Filter\LIST{\SemParen{e}}$\\
%%   $\SemParen{\jq!transform!~e}_{\textsf{F}}$ \> $=$ \> $\SemParen{\jq!transform each $!~e}_{\textsf{F}}$\\
%%   $\SemParen{\jq!transform each!~x~e}_{\textsf{F}}$ \> $=$ \>$\Transform\LIST{\SemParen{e}}$\\
%%   $\SemParen{\jq!expand each!~x~e}_{\textsf{F}}$ \> $=$ \> $\SemParen{\jq!expand!}_{\textsf{F}}; \SemParen{\jq!transform each!~x~e}_{\textsf{F}}$\\
%%   $\SemParen{\jq!expand!}_{\textsf{F}}$ \> $=$ \>$\Expand$\\
%%   $\SemParen{\jq!group into!~e}_{\textsf{F}}$ \> $=$ \> $\SemParen{\jq!group by k = true into!~e}_{\textsf{F}}$\\
%%   $\SemParen{\jq!group!\,\jq!by!\,k\jq!=!e_1\,\jq!into!\,e_2}_{\textsf{F}}$ \> $=$ \> $\SemParen{\jq!group each $ by!~k~\jq!=!~e_1~\jq!into!~e_2}_{\textsf{F}}$\\
%%   %% Changes layout at that point.
%%   $\SemParen{\jq!filter!~e}_{\textsf{F}}$ \= \kill
%%   $\SemParen{\jq!group!\,\jq!each!\,x\,\jq!by!\,k\jq!=!e_1\,\jq!into!\,e_2}_{\textsf{F}}$ $=$\\
%%   \> $\SemParen{\jq!group!\;\jq!each!~x~\jq!by!~k~\jq!=!~e_1~\jq!as $ into!~e_2}_{\textsf{F}}$\\
%%   $\SemParen{\jq!group each!~x~\jq!by!~k~\jq!=!~e_1~\jq!as!~g~\jq!into!~e_2}_{\textsf{F}}$ $=$\\
%%   \> $\fseq{\GROUPBY{\fpat{x}{\SemParen{e_1}}}}{\Transform\LIST{\fpat{\fprod{k}{g}}{\SemParen{e_2}}}}$
%% \end{tabbing}
This translation defines the (first, in our knowledge) formal semantics of
Jaql.  Such a translation is \emph{all} that is needed to
define the semantics of a NoSQL language and, as a bonus, endow it 
with the type inference system we described \changes{\emph{without requiring any modification of the original language}}. No further action is demanded since the machinery to exploit it is
all developed in this work. 

As for typing, every Jaql expression is encoded into a filter for which type-checking is ensured to terminate: \textit{Check$()$} holds for \texttt{Filter[]}, \texttt{Transform[]}, and
\texttt{Expand} (provided it
holds also for their arguments) since they only perform recursive
calls on recombinations of subtrees of their input; by its definition, the encoding does not introduce any new recursion and, hence, it always yields a composition and application of filters for
which \textit{Check$()$} holds.

% we can first remark that \textit{Check()} holds for $\Filter$, $\Transform$ and
% $\Expand$ (provided that it
% also holds for their parameter). Indeed, they only perform recursive
% calls on (recombination) of subtrees of their input thus ensuring the termination of
% typechecking.
% Second, since the result of the translation is a composition of filters for
% which \textit{Check} holds (with no recursion introduced), it can be
% typechecked as well.

\def\negspace{-2mm}
\subsection{Examples}\label{sec:example}
To show how we use the encoding, let us encode the
example of the introduction. For the sake of the concision we
will use filter definitions (rather than expanding them in details). We
use \texttt{Fil} and \texttt{Sel} defined in the
introduction, \texttt{Expand} and \texttt{Transform[]} defined
at the beginning of the section, the encoding of Jaql's field selection as defined in Section~\ref{sec:sugar}, and finally \texttt{Head} that returns the first element of a sequence and a family of recursive filters
\texttt{Rgrp}$i$ with $i\in\mathbb{N^+}$ both defined below:
\begin{alltt}
let filter Head = \NIL => \Null | (x,xs) => x\\[\negspace]
let filter Rgrp\(i\) = \NIL => \NIL
                 | ((\(i\),x),tail) => (x , Rgrp\(i\) tail) 
                 | _ => Rgrp\(i\) tail
\end{alltt}
Then, the query in the introduction is encoded as follows
\begin{alltt}\internallinenumbers
  [employees depts];
  [Sel Fil];
  [Transform[x =>(1,x)]  Transform[x =>(2,x)]];
  Expand;
  groupby ( (1,$)=>$.dept | (2,$)=>$.depid );
  Transform[(g,l)=>(
           [(l; Rgrp1) (l; Rgrp2)];
           [es ds] =>
                  \{ dept: g,
                    deptName: (ds ; Head).name),
                    numEmps: count(es) \} )]
\end{alltt}
In words, we perform the selection on employees and filter the
departments (lines\,{\small1-2}); we tag each element by \texttt{1} if it comes from
employees, and by \texttt{2} if it comes from departments (line\,{\small3}); we merge
the two collections (line\,{\small4}); we group the heterogeneous list according to
the corresponding key (line\,{\small5}); then for each element of the result of
grouping we capture in \texttt{g} the key (line\,{\small6}), split the group into
employees and depts (line\,{\small7}), capture each subgroup into the corresponding
variable (\emph{ie}, \texttt{es} and \texttt{ds}) (line\,{\small8}) and return the
expression specified in the query after the ``\texttt{\ibm into}'' (lines\,{\small8-10}).
The general definition of the encoding for the co-grouping is given in \appcrazy~\ref{sec:cogroupingencoding}.

Let us now illustrate how the above composition of filters is typed. Consider an instance where:
\begin{itemize}
\item \texttt{employees} has type \texttt{ [ \textrm{\it Remp}\ks\ ]}, where\\
      \textit{Remp} $\equiv$ \texttt{\{ dept: int,
    income:int, ..\} }
\item \texttt{depts} has type \texttt{[ (\textrm{\it Rdep}
    | \textrm{\it Rbranch})* ]}, where\\
     \textit{Rdep} $\equiv$ \texttt{\{depid:int, name:
    string, size: int\}}\\
    \textit{Rbranch} $\equiv$ \texttt{\{brid:int, name:
    string\}}\\
    (this type is a subtype of \texttt{Dept} as defined in the introduction)
\end{itemize}
The global input type is therefore (line\,{\small1})\\
\centerline{\texttt{[ [ \textrm{\it Remp}* ] [ (\textrm{\it Rdep}
    | \textrm{\it Rbranch})* ] ] }}
which becomes, after selection and filtering (line\,{\small2})\\
\centerline{\texttt{[ [ \textrm{\it Remp}* ] [ \textrm{\it Rdep}* ] ] }}
(note how all occurrences of \textit{Rbranch} are ignored by
\texttt{Fil}). Tagging with an integer (line\,{\small3}) and flattening (line\,{\small4}) yields\\
\centerline{\texttt{[ (1,\textrm{\it Remp})* (2,\textrm{\it Rdep})*  ] }}
which illustrates the precise typing of products coupled
with singleton types (\emph{ie}, \texttt{1} instead of
\texttt{int}). While the \texttt{groupby} (line\,{\small5}) introduces an approximation
the dependency between the tag and the corresponding type is kept\\
\centerline{\texttt{[ (int, [ ((1,\textrm{\it Remp}) | (2,\textrm{\it Rdep})
  )+ ]) * ]}}
Lastly the transform is typed exactly, yielding the final type\\[.5mm]
\centerline{\texttt{ [ \{dept:int, deptName:string|null, numEmps:int \}* ]}}\\[.5mm]
Note how \texttt{null} is retained in the output type (since there may
be employees without a department, then \texttt{Head} may be applied
to an empty list returning \Null{}, and the selection of \texttt{name}
of \Null{} returns \Null). 
For instance suppose to pipe the Jaql grouping defined in the
introduction into the following Jaql expression, in order  to
produce a printable representation of the records of the result
\begin{alltt}\ibm
transform each x (
 (x.deptName)@":"@(to_string x.dep)@":"@(x.numEmps))
\end{alltt}
where  \texttt{@} denotes string concatenation and \texttt{to\_string}
is a conversion operator (from any type to string). The composition is ill-typed for three reasons:  the field \texttt{dept} is misspelled as \texttt{dep}, \texttt{x.numEmps} is of type \texttt{int} (so it must be applied to \texttt{to\_string} before concatenation), and the programmer did not account for the fact
that the value stored in the field \texttt{deptName} may be
\texttt{null}. The encoding produces the following lines to be appended to the previous code:
\resetlinenumber[12]
\begin{alltt}\internallinenumbers
Transform[ x =>  
 (x.deptName)@":"@(to_string x.dep)@":"@(x.numEmps)]
\end{alltt}
in which all the three errors are detected by our type system. A subtler example of error is given by the following alternative code
\resetlinenumber[12]
\begin{alltt}\internallinenumbers
Transform[
   \{ dept : d, deptName: n\&String, numEmps: e \} =>
      n @ ":" @ (to_string d) @ ":" @ (to_string e)
 | \{ deptName: null, .. \} => ""
 | _ => "Invalid department" ]
\end{alltt}
which corrects all the previous errors but adds a new one since, as
detected by our type system, the last branch can be never selected.
\rechanges{As we can see, our type-system ensures soundness, forcing the programmer
to handle exceptional situations (as in the
\texttt{null} example above) but is also precise enough to detect
that some code paths can never be reached.}

In order to focus on our contributions we kept the language of types and filters simple. However there already exists several contributions on the types and expressions used here. Two in particular are worth mentioning in this context: recursive patterns and XML.

Definition~\ref{patdef} defines patterns inductively but,
alternatively, we can consider the (possibly infinite) regular trees
\emph{coinductively} generated by these productions and, on the lines of
what is done in \cduce, use the recursive patterns so obtained
to encode regular expressions patterns~(see~\cite{BCF03}). Although
this does not enhance expressiveness, it greatly improves the writing of
programs since it makes it possible to capture distinct subsequences of a
sequence by a single match. For instance, when a sequence  is
matched against a pattern such as \texttt{[ (\Int\ as x | \Bool\ as y |
  \_)\ks\ ]}, then \texttt{x} captures (the list of) all integer elements
 (capture variables in regular expression patterns are bound to lists), \texttt{y} captures all  Boolean elements, while the
remaining elements are ignored. By such patterns,
co-grouping can be encoded without the \texttt{Rgrp}. For instance,
the transform in lines 6-11  can be more compactly
rendered as:\vspace{-1mm}
\resetlinenumber[6]
\begin{alltt}\internallinenumbers
  Transform[(g,[ ((1,es)|(2,ds))* ]) =>
                     \{ dept: g, 
                       deptName: (ds;Head).name, 
                       numEmps: count(es) \}]
\end{alltt}
\vspace{-1mm}
For what concerns XML, the types used here were originally defined
for XML, so it comes as a no surprise that they can seamlessly express
XML types and values. For example \cduce{} uses the very same types used here
to encode both XML types and elements as triples, the first element being the
tag, the second a record representing attributes, and the third a
heterogeneous sequence for the content of the element. Furthermore, we
can adapt the results of~\cite{Castagna2008} to encode forward
XPath queries in filters. Therefore, it requires little effort to use the
filters presented here to encode languages such as JSONiq~\cite{JSONiq}
designed to integrate JSON and XML, or to precisely type regular
expressions, the import/export of XML data, or XPath queries embedded in Jaql programs. 
\ifLONGVERSION
This is shown in the section that follows.
\else%%%%%%%%%%%
\rechanges{The description of these encodings can be found in the long version of
this paper, where we also argue that it is better to extend NoSQL
languages with XML primitives directly derived from our system rather
than to use our system to encode languages such as JSONiq. As a matter
of fact, existing approaches tend to juxtapose XML and JSON operators
thus yielding to stratified (\emph{ie}, not tightly integrated)
systems which have several drawbacks (\emph{eg}, JSONiq does not allow
XML nodes to contain JSON objects and arrays). Such restrictions are absent
from our approach since both XML and JSON operators are encoded in the same
basic building blocks and, as such, can be freely nested and combined.}
\fi%%%%%%%%%%%%%%%%%%%%%%%%
  
\ifLONGVERSION%%%%%%%%%%%%%%%%
\section{JSON, XML, Regex}
\label{label:xml}
There exist various attempts to integrate JSON and XML. For instance
JSONiq~\cite{JSONiq} is a query language designed to allow XML and JSON to be used in the same query. The motivation is that JSON and XML are both
widely used for data interchange on the Internet. In many
applications, JSON is replacing XML in Web Service APIs and data
feeds, while more and more applications support both formats. More
precisely, JSONiq embeds JSON into an XML query language (XQuery), but
it does it in a stratified way: JSONiq does not allow XML nodes to
contain JSON objects and arrays. The result is thus similar to
OCamlDuce, the embedding of  \cduce{}'s XML types and expressions into
OCaml, with the same drawbacks.

Our type system is derived from the type system of \cduce{}, whereas the
theory of filters was originally designed to use \cduce{} as an host
language. As a consequence XML types and expressions can be
seamlessly integrated in the work presented here, without any
particular restriction. To that end it suffices to use for XML elements
and types the same encoding used in the implementation of \cduce,
where an XML element is just a triple formed by a tag (here, an
expression), a record (whose labels are the attributes of the element),
and a sequence (of characters and or other XML elements) denoting its
content. So for instance the following element
\begin{alltt}
<product system="US-size">
  <number>557</number>
  <name>Blouse</name>
</product>
\end{alltt}
is encoded by the following triple:
\begin{alltt}
("product" , \{ system : "US-size" \} ,
   [ 
     ("number" , \{\} , [ 557 ])
     ("name", \{\}, "Blouse") 
   ]
)
\end{alltt}
and this latter, with the syntactic sugar defined for \cduce, can be written as:
\begin{alltt}
<product system="US-size">[
  <number>[ 557 ]
  <name>[ Blouse ]
]
\end{alltt}
Clearly in our system there are no restrictions in merging and nesting JSON and XML and no
further extension is required to our system to define XML query and processing
expressions. Just, the introduction of syntactic sugar to make the
expressions readable, seems helpful:
$$
\begin{array}{rcl}
e &::= &\texttt{<$e$ $e$>$e$}\\
f &::= &\texttt{<$f$\;$f$>$f$}
\end{array}$$
The system we introduced here is already able to reproduce (and type) the same
transformations as in JSONiq, but without the restrictions and drawbacks of the latter (this is why we argue that it is better to extend NoSQL languages with XML primitives directly derived from our system rather than to use our system to encode JSONiq).
For instance, the example given in the JSONiq draft to
show how to render in Xhtml the following JSON data:
\begin{alltt}
\{
  "col labels" : ["singular", "plural"],
  "row labels" : ["1p", "2p", "3p"],
  "data" :
     [
        ["spinne", "spinnen"],
        ["spinnst", "spinnt"],
        ["spinnt", "spinnen"]
     ]
\}
\end{alltt}
can be encoded in the filters presented in this work (with the new syntactic sugar) as:
\begin{alltt}
\{ "col labels" : cl ,
  "row labels" : rl ,
  "data" : dl 
\} =>
 <table border="1" cellpadding="1" cellspacing="2">[
   <tr>[ <th>[ ] !(cl; Transform[ x -> <th>x ]) ]
   !(rl; Transform[ h ->
          <tr>[ <th>h !(dl; Transform[ x -> <td>x ]) ] 
         ]
   )
 ] 
\end{alltt}
(where \texttt{!} expands a subsequence in the containing sequence). The resulting Xhtml document is rendered in a web browser as:
\begin{center}
\includegraphics[width=7cm]{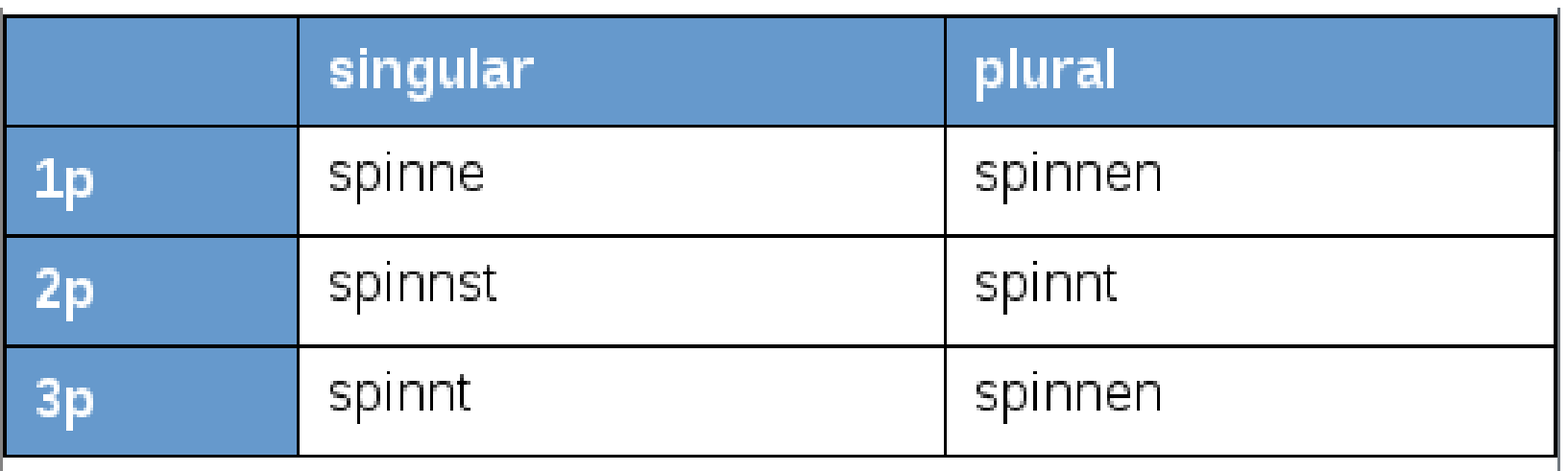} 
\end{center}
%% \begin{center}
%% \begin{tabular}{||l||l||l||}
%% \hline\hline
%%       & \bf singular & \bf plural \\\hline\hline
%% \bf 1p& spinne       & spinnen    \\\hline\hline
%% \bf 2p& spinnst      & spinnt     \\\hline\hline
%% \bf 3p& spinnt       & spinnen    \\\hline\hline
%% \end{tabular}
%% \end{center} 

Similarly, Jaql built-in libraries include functions to convert and
manipulate XML data. So for example as it is possible in Jaql to embed
SQL queries, so it is possible to evaluate XPath expressions, by the
function \texttt{\ibm xpath()} which takes two arguments, an XML document
and a string containing an xpath expression---\emph{eg},
\texttt{\ibm xpath(read(seq(("conf/addrs.xml"))\,,\,"content/city")} ---.  Filters can
encode forward XPath expressions (see~\cite{kimthesis}) and precisely type them. So while
in the current implementation there is no check of the type of the
result of an external query (nor for XPath or for SQL) and the XML document is produced independently from Jaql, by
encoding (forward) XPath into filters we can not only precisely type calls
to Jaql's \texttt{\ibm xpath()} function but also feed them with documents produced
by Jaql expressions.

Finally, the very same regular expressions types that are used to
describe heterogeneous sequences and, in particular, the content of XML
elements, can be used to type regular expressions. Functions working on
regular expressions (regex for short) form, in practice, yet another
domain specific language that is embedded in general purpose languages
in an untyped or weakly typed (typically, every result is of
type \String) way. Recursive patterns can straightforwardly encode
regexp matching. Therefore, by combining the pattern filter with other filters it is
possible to encode any regexp library, with the important advantage
that, as stated by Theorem~\ref{th:at}, the set of values
(respectively, strings) accepted by a pattern (respectively, by a
regular expression), can be precisely computed and can be expressed by a type. So a
function such as the built-in Jaql's function \texttt{\ibm
regex\_extract()}, which extracts the substrings that match a given
regex, can be easily implemented by filters and precisely typed by our
typing rules. Typing will then amount to intersect the type of the
string (which can be more precise than just {\ibm\String}) with the type
accepted by the pattern that encodes the regex at issue.

\section{Programming with filters}
\label{label:programming}

Up to now we have used the filters to encode operators hard-coded in
some languages in order to type them. Of course, it is possible to
embed the typing technology we introduced directly into the compilers
of these language so as to obtain the flexible typing that characterizes our
system. However, an important aspect of filters we have ignored so far
is that they can be used directly by the programmer to define user-defined operators that are typed as precisely as the hard-coded
ones. Therefore a possibility is extend existing NoSQL
languages by adding to their expressions the filter application $fe$ expression.

The next problem is to decide how far to go in the definition of filters. A
complete integration, that is taking for $f$ all the definitions given
so far, is conceivable but might disrupt the execution model of the
host language, since the user could then define complex iterators that
do not fit map-reduce or the chosen distributed compilation policy.
A good compromise could be to add to the host language
only filters which have ``local'' effects, thus avoiding to affect
the map-reduce or distributed compilation execution model. The minimal
solution consists in choosing just the filters for patterns, unions, and expressions:
$$f ~~:: =~~ e~~\mid~~ p\texttt{ => }f~~\mid~~ f\texttt{|}f$$
Adding such filters to Jaql (we use the ``\texttt{=>}'' arrow for patterns in order to avoid confusion with Jaql's ``\texttt{->}'' pipe operator) would not allow the user to define powerful operators, but their use would already dramatically improve type precision. For instance we could define the following Jaql expression
\else%%%%%%%%%%%%%%%%%%%%%%%%%%%%%%%%%%%%%%%%%%%%%%%%%%%%%%%%%%
\subsection{Extensions}\label{label:programming}
Hitherto we used filters only to encode primitive operators of some
NoSQL languages, in particular Jaql. However, it is possible
to \emph{add} filters to other languages, so as to have user-defined
operators typed as precisely as primitive ones. From a linguistic
point of view this is a no-brainer: it suffices to add filter
application to the expressions of the host language. However, such an
extension can be problematic from a computational viewpoint, since
it may disrupt the execution model, especially for what concerns
aspects of parallelism and distribution. A good compromise is to
add only filters that have ``local'' effects, which
can already bring dramatic increases in expressiveness and type precision
without disrupting the distributed compilation model. For instance, one
can add just pattern and union filters as in the following (extended)
Jaql program:
\fi%%%%%%%%%%%%%%%%%%%%%%%%%%%%%%%%%%%%%%%%%%%%%%%%%%%%%%%%%%%%
\\[1mm]
\texttt{\ibm transform\;(\;\patrec{a\col x,..}\;as\;y => \patrec{y.*,\;sum\col
x+x} | y\;=>\;y\;)}
\\[1mm] 
(with the convention that a filter occurring as an expression
denotes its application to the current argument
\texttt{\$}). With this syntax, our inference system
is able to deduce that feeding this expression with an argument of type
\LIST{\texttt{\patrec{a?\col int, c\col bool}*}} returns a result of
type \LIST{\texttt{(\pator{\patrec{a\col int, c\col bool, sum\col
int}}{\patrec{c\col bool}})*}}. This precision comes from the capacity of our
inference system to discriminate between the two branches of the filter and
deduce that a \texttt{\ibm sum} field will be added only if the \texttt{\ibm a}
field is present. Similarly by using pattern matching in a Jaql
``\texttt{\ibm filter}'' expression, we can deduce that \texttt{\ibm filter\,(\;int\,=>\,true | \_\;=>\,false\;)} fed with any sequence of elements always
returns a (possibly empty) list of integers. An even greater precision
can be obtained for grouping expressions when the generation of the
key is performed by a filter that discriminates on types: the result
type can keep a precise correspondence between keys and the corresponding groups.   
\ifLONGVERSION%%%%%%%%%%%%%%%%%%%%%%%%%%%%%%%%%%%%
As an example consider the following (extended) Jaql grouping
expression:{\ibm
\begin{verbatim}
group e by({town: ("Roma"|"Pisa"), ..} => "Italia" 
          |{town: "Paris", ..} => "France" 
          | _ => "?")
\end{verbatim}
}
if $e$  has type \texttt{[\{town:string, addr:string\}*]}, then the type inferred by our system for this groupby expression is 
\begin{verbatim}
[( ("Italia", [{town:"Roma"|"Pisa", addr:string}+])
  |("France", [{town:"Paris", addr:string}+])
  |("?", [{town:string/("Roma"|"Pisa"|"Paris"), 
           addr:string}+])
 )*]
\end{verbatim}
which precisely associates every key to the type of the elements it groups.

Finally, in order to allow a modular usage of filters, adding just filter application to the expression of the foreign language does not suffice: parametric filter definitions are also needed.
$$e::= fe ~~\mid~~ \LET{\FILTER F\LIST{F_1,\ldots,F_n}}{f}{e}$$ 
However, on the line of what we already said about the disruption of
the execution model, recursive parametric filter definitions should be
probably disallowed, since a compilation according to a map-reduce model would require to
disentangle recursive calls.

\fi%%%%%%%%%%%%%%%%%%%%%%%%%%%%%%%%%%%%%%%%%%%%%%%%%%%

% % We show how to encode JaQL Core operators with filters and
% % expressions.

% % \paragraph{Jaql records}

% % \begin{verbatim}
% %   r = { a: 1, b: 2, c: 3 };
% %   // CDuce let r = { a = 1; b = 2; c = 3 }

% %   ar = [ { a: 1, b: 2 }, { a: 10, b: 20 } ];

% %   // CDuce: [ { a = 1  b = 2 } { a = 10 b = 20 } ]
% %   // record projection
% %   //
% %   r.a; // returns 1

% %   r{.a,.b} // returns {a:1, b:2}

% %   // CDuce: {a=r.a; b=r.b}
% %   // array element access
% %   //

% %   ar[*].a; // [ 1, 10 ]

% %   // CDuce: transform (x -> x.a) ar

% %   ar[*]{.a}; // [ { a: 1 }, { a: 10 } ]

% %   // CDuce: transform (x -> {a=x.a}) ar

% %   // manipulating fields in a record
% %   // 
  
% %   // create a new record that has
% %   // r's fields and a 'd' field
% %   { r.*, d:4 }; // { a: 1, b: 2, c: 3, d: 4 }

% %   // CDuce: r + {d = 4}

% %   // create a new record that has
% %   // r's field, except 'b', and a 'd' field
% %   { r{* - .b}, d: 4 }; // { a: 1, c: 3, d: 4 }

% %   // CDuce: (r \ b) + {d = 4}
% % \end{verbatim}

\ifLONGVERSION\else
\section{JSON, XML, Regex}
\label{label:xml}

\fi

\section{Commentaries}
\label{label:comment}
%In this section we explain some of the subtler design choices we did when defining our system.

Finally, let us explain some subtler design choices for our system.
\paragraph{Filter design:}
The reader may wonder whether products and record \emph{filters} are
really necessary since, at first sight, the filter $\fprod{f_1}{f_2}$
could be encoded as $\fpat{(x,y)}{(f_1x,f_2y)}$ and similarly for
records. The point is that $f_1x$ and $f_2y$ are expressions ---and
thus their pair is a filter--- only if the $f_i$'s are closed \changes{(\emph{ie}, wihtout free term recursion variables). Without an explicit
product filter it would not be possible to program a filter as simple
as the identity map, $\frec{X}\falt{\fpat{\NIL}{\NIL}}{\fpat{(h,t)}{\fprod{h}{Xt}}}$ since
$Xt$ is not an expression ($X$ is a free term recursion variable)}%
\ifLONGVERSION%%%%%%%%%%%%%%%%%%%%%%%%%%%%%%%%%
.\footnote{Syntactically we could write
$\frec{X}{\falt{\fpat{\NIL}{\NIL}}{\fpat{(h,t)}{\fprod{h}{(Xt)v}}}}$
where $v$ is any value, but then this would not pass
type-checking since the expression $(Xt)v$ must be typeable
without knowing the $\Delta$ environment (cf.\ rule {\sc[Filter-App]} at the beginning of Section~\ref{sec:tfp}). We purposedly
stratified the system in order to avoid mutual recursion between filters and expressions.}
\else%%%%%%%%%%%%%%%%%%%%%%%%%%%%%%%%%%%%%%%%%%
. 
\fi%%%%%%%%%%%%%%%%%%%%%%%%%%%%%%%%%%%%%%%%%%%%
Similarly, we need an
explicit record filter to process recursively defined record types
such as
$\frec{X}{(\falt{\patrec{\texttt{head}\col\Int,\texttt{tail}\col
X}}{\NIL})}$.

Likewise, one can wonder why we put in filters only the ``open'' record variant that copy extra fields and not the closed one. The reason is that if we want a filter to be applied only to records with exactly the fields specified in the filter, then this can be simply obtained by a pattern matching. So the filter $\patrec{\ell_1\col f_1,\ldots,\ell_n\col f_n}$ (\emph{ie}, without the trailing ``\textbf{.\,.}'') can be simply introduced as syntactic sugar for $\fpat{\patrec{\ell_1\col \tany,\ldots,\ell_n\col \tany}}{\forec{\ell_1\col f_1,\ldots,\ell_n\col f_n}}$

\paragraph{Constructors:} 
The syntax for constructing records and pairs is exactly the same in
patterns, types, expressions, and filters. The reader may wonder why
we did not distinguish them by using, say, $\times$ for product types
or $=$ instead of $\col$ in record values. This, combined with the
fact that values and singletons have the same syntax, is a critical
design choice that greatly reduces the confusion in these languages,
since it makes it possible to have a unique representation for
constructions that are semantically equivalent. Consider for instance
the pattern $\patpair{x}{\patpair{3}{\NIL}}$. With our syntax
$\patpair{3}{\NIL}$ denotes both the product type of two singletons
$3$ and $\NIL$, or the value $\patpair{3}{\NIL}$, or the singleton
that contains this value. According to the interpretation we choose,
the pattern can then be interpreted as a pattern that matches a
product or a pattern that matches a value. If we had differentiated
the syntax of singletons from that of values (\emph{eg},
\texttt{\{$v$\}}) and that of pairs from products, then the pattern
above could have been written in five different ways. The point is
that they all would match exactly the same sets of values, which is why we
chose to have the same syntax for all of them.

\paragraph{Record types: }
\ifLONGVERSION%%%%%%%%%%%%%%
 The definition of records is redundant (both for types and patterns). Instead of the current definition we could have used  just $\patrec{\ell\col t}$, $\patrec{}$,
  and $\patrec{\textbf{..}}$, since the rest can be encoded by intersections. For instance,
  $\patorec{\ell_1\col t_1, \ldots,
    \ell_n\col t_n}= \patand{\patrec{\ell_1\col t}}{\patand{...}{\patand{\patrec{\ell_n\col t_n}}{\patrec{..}}}}$. We opted to use the redundant definition for the sake of clarity.
\fi%%%%%%%%%%%%%%%%%%%%%%%%%
\rechanges{
In order to type records with computed labels we distinguished two
cases according to whether the type of a record label is
finite or not. Although such a distinction is simple, it is not
unrealistic. Labels with singleton types cover the (most common) case of
records with statically fixed labels. The dynamic choice of a label from a
statically known list of labels is a usage pattern seen in JavaScript when
building an object which must conform to some interface based on a
run-time value. Labels with infinite types cover the fairly common
usage scenario in which records are used as dictionaries: we deduce
for the expression computing the label the type \texttt{string}, thus forcing
the programmer to insert some code that checks that the label is
present before accessing it.

The rationale behind the typing of records was twofold. First and foremost, in this
work we wanted to avoid type annotations at all costs (since there is
not even a notion of schema for JSON records and collections ---only
the notion of basic type is defined--- we cannot expect the Jaql
programmer to put any kind of type information in the code).  
More sophisticated type
systems, such as dependent types, would probably preclude type
reconstruction: dependent types need a lot of annotations and this does not fit
our requirements. Second, we wanted the type-system to be simple yet
precise.  Making the finite/infinite distinction increases typing
precision at no cost (we do not need any extra machinery since
we already have singleton types). Adding heuristics or complex analysis just
to gain some precision on records would have blurred the main focus of
our paper, which is not on typing records but on typing \emph{transformations} on records. We leave such additions for future work.
}

\paragraph{Record polymorphism:}
The type-oriented reader will have noticed that we do not use row
variables to type records, and nevertheless we have a high degree of
polymorphism. Row variables are useful to type functions or
transformations since they can keep track of record fields that are
not modified by the transformation. In this setting we do not need
them since we do not type transformations (\emph{ie}, filters) but
just the application of transformations (filters are not first-class terms). We have polymorphic typing
via filters 
(see how the first example given in Section~\ref{label:programming} keeps track of the {\ibm\tt c} field) and therefore open records suffice.
\ifLONGVERSION %%%%%%%%%%%%%%%%%%%%%

\paragraph{Record selection:}
Some languages ---typically, the dynamic ones such as Javascript, Ruby, Python---
allow the label of a field selection to be computed by an expression. We considered the definition of a fine-grained rule to type expressions of the form $e_1\pmb.e_2$: whenever $e_2$ is typed by a finite unions of strings, the rule would give a finite approximation of the type of the selection. However, such an extension would complex the definition of the type system, just to handle few interesting cases in which a finite union type can be deduced. Therefore, we preferred to omit its study and leave it for future work.
\fi%%%%%%%%%%%%%%%%%%%%%%%%%

% \paragraph{Subtyping:} We never spoke about subtyping. Although it
% is well known how to efficiently decide the subtyping relation for
% the types of Definition~\ref{def:types} (see~\cite{jacm08}) we do
% not need subtyping insofar as we do not have first class functions
% or filters. Since we directly type filter applications we do not
% need to explicitly mention subtyping, even though we need it for the
% proof of subject reduction.

\ifLONGVERSION
\section{Related Work}
\label{label:related}
\else\paragraph{Related work:}
\fi
%\beppe{From the introduction:
% What to encode? nested relational algebra (SPJ), OQL (djoin in
%    Guido) [typing issues?], UQL (Bunemann Davidson Suciu), UnQL,
%    Jaql, AQL, Pig. Compare the typing technique of
%    Ohori. Squeryl. See type system of Meijer.}

%\vero{Verifier le papier Ohori ICFP11 remain to be done: Stream calculus~\cite{Streams}, Filters~\cite{Castagna2008}.
%Should I develop more, give mora details ? Space is getting rare...}
In the (nested) relational (and SQL) context, many works have studied the  integration of (nested)-relational algebra or SQL into general purpose programming languages. Among the first attempts was the integration of the relational model in Pascal~\cite{pascalr} or in Smalltalk~\cite{opal}. Also, monads or comprehensions~\cite{WadTri89,DBLP:conf/icdt/TannenBW92,DBLP:journals/sigmod/BunemanLSTW94} have been successfully used to design and implement query languages including a way to embed queries within host languages. 
Significant efforts have been done to equip those languages with type systems and type checking disciplines~\cite{buneman1981, buneneman1988,machiavelli1989} and more recently~\cite{DBLP:conf/icfp/OhoriU11} for integration and typing aspects. 
However, these approaches only support homogeneous sequences of records in the context of specific classes of queries (practically equivalent to a nested relational algebra or calculus), they do not account for records with computable labels, and therefore they are not easily transposable to a setting where sequences are heterogeneous, data are semi-structured, and queries are much more expressive.

%\vero{ICI: filters [CN08]}
While the present work is inspired and stems from previous works on
the XML iterators, targeting
NoSQL languages made the filter calculus presented here substantially
different from the one of \cite{Castagna2008,kimthesis}
(dubbed XML filters in what follows), as well in  syntax as
in dynamic and static semantics.
In~\cite{Castagna2008} XML filters behave as some kind of
top-down tree transducers, termination is enforced by
heavy syntactic restrictions, and a
\emph{less} constrained use of the composition makes 
type inference challenging and requires sometimes cumbersome type
annotations. While XML filters are allowed to
operate by composition on the \emph{result} of a recursive call
(and, thus, simulate bottom-up tree transformations), the absence of explicit arguments in recursive calls makes programs understandable only to well-trained programmers. In contrast,
the main focus of the current work
was to make programs immediately intelligible to any functional programmer and make filters effective for the typing of sequence
transformations: sequence iteration, element filtering, one-level
flattening. The last two are especially difficult to write with XML
filters (and require type annotations). Also,
the integration of filters with record types (absent in
\cite{Castagna2008} and just sketched in \cite{kimthesis}) is novel
and much needed to encode JSON transformations.\looseness -1

%\kim{ Je n'aime pas cette phrase: \emph{``Indeed our purpose was (i) to
%  handle type inference and type checking in presence of heterogenous
%  sequences (of records) and regular type expressions (ii) to provide
%
%  an expressive calculus able to take into account the expressiveness
%  of existing languages while guaranteeing a very precise type
%  inference. Indeed our filter calculus is Turing complete and allows
%  for a high degree of precision in typing.''}
%C'est trop ambigu car ce qu'on peut typer n'est pas turing complet. 
%}

%\vero{Should we ive an example of what we obtain with our system and what is obtained by others ?}

\vspace{-1mm}

\section{Conclusion}
\label{label:conclusion}
\rechanges{Our work addresses two very practical problems, namely the typing of NoSQL languages and a comprehensive definition of their semantics. These languages add to list
comprehension and SQL operators the ability to work on
heterogeneous data sets and are based on JSON (instead of
tuples). Typing precisely  each of these features
using the best techniques of the literature would probably yield quite a
complex type-system (mixing row polymorphism for records,
   parametric polymorphism, some form of dependent typing,...) and we
   are skeptical that this could be achieved without using any
   explicit type annotation.
Therefore we explored the formalization of these languages from scratch, by
defining a calculus and a type system.} The thesis we
defended is that all operations typical of current NoSQL
languages, as long as they operate structurally (\emph{ie},
without resorting on term equality or relations), amount to a combination of more basic bricks: our filters. On the structural side, the claim is that
combining recursive records and pairs by unions, intersections, and
negations suffices to capture all possible structuring of data,
covering a palette ranging from comprehensions, to heterogeneous lists
mixing typed and untyped data, through regular expressions types and XML
schemas. Therefore, our calculus not only provides a simple way to give a
formal semantics to, reciprocally compare, and combine operators of different NoSQL
languages, but also offers a means to
equip  these languages, in they current definition (\emph{ie}, without any type definition or annotation), with precise type inference.
\ifLONGVERSION%%%%%%%%%%%%%%%%%

As such we accounted
for both components that, according to Landin, constitute the design
of a language: operators and data structures. But while
Landin considers the design of terms and types as independent activities we,
on the contrary, advocate an approach in which the design of former
is \emph{driven} by the form of latter. Although other approaches are
possible, we tried to convey the idea that this approach is
nevertheless the only one that yields a type system whose precision,
that we demonstrated all the work long, is comparable only to the precision
obtained with hard-coded (as opposed to user-defined)
operators. As such, our 
\else%%%%%%%%%%%%%%%%%%%
This
\fi%%%%%%%%%%%%%%%%%%%%%%
type inference yields and surpasses in precision systems
using parametric polymorphism and row variables. The price to pay is
that \emph{transformations} are not first class: we do not type 
filters but just their applications. However, this seems an advantageous
deal in the world of NoSQL languages where ``selects'' are never
passed around (at least, not explicitly), but early error detection is critical, 
especially in the view of
the cost of code deployment.\footnote{%
Only filter-encoded operators are not first class: if the host language provides, say, higher-order functions, then they stay higher-order and are typed by embedding the host type system, if any, via ``foreign type calls''.}

The result are filters, a set of untyped terms that can be easily included in a host language to complement in a typeful framework existing operators with user-defined ones. 
The requirements to include filters into a host
language are so minimal that every modern typed programming language
satisfies them. The interest resides not in the fact that we
can add filter applications to any language, rather that filters can be
used to define a smooth integration of calls to domain specific
languages (\emph{eg}, SQL, XPath, Pig, Regex) into general purpose ones
(\emph{eg}, Java, C\#, Python, OCaml) so as both can share the
same set of values and the same typing discipline. 
\rechanges{Likewise, even though filters provide an early prototyping platform 
for queries, they cannot currently be used as a final compilation
stage for NoSQL languages: their operations rely on a Lisp-like
encoding of sequences and this makes the correspondence with optimized bulk operations on lists awkward. Whether we
can derive an efficient compilation from filters to map-reduce
(recovering the bulk semantics of the high-level language) is a challenging question.

Future plans include practical experimentation of our technique: we
intend to benchmark our type analysis against existing collections of
Jaql programs, gauge the amount of code that is ill typed and verify on this
how frequently the programmer adopted defensive programming to cope
with the potential type errors.}

\vspace{-2mm}

\bibliographystyle{abbrv}
%\bibliography{main-short}
\bibliography{main}

\begin{thebibliography}{10}

\bibitem{Alon01}
N.~Alon, T.~Milo, F.~Neven, D.~Suciu, and V.~Vianu.
\newblock {XML} with data values: typechecking revisited.
\newblock In {\em Proceedings of the twentieth ACM SIGMOD-SIGACT-SIGART
  symposium on Principles of database systems}, PODS '01, pages 138--149, New
  York, NY, USA, 2001. ACM.

\bibitem{Alon03}
N.~Alon, T.~Milo, F.~Neven, D.~Suciu, and V.~Vianu.
\newblock Typechecking {XML} views of relational databases.
\newblock {\em ACM Trans. Comput. Logic}, 4:315--354, July 2003.

\bibitem{Asterix}
A.~Behm, V.~R. Borkar, M.~J. Carey, R.~Grover, C.~Li, N.~Onose, R.~Vernica,
  A.~Deutsch, Y.~Papakonstantinou, and V.~J. Tsotras.
\newblock Asterix: towards a scalable, semistructured data platform for
  evolving-world models.
\newblock {\em Distributed and Parallel Databases}, 29(3):185--216, 2011.

\bibitem{BCF03}
V.~Benzaken, G.~Castagna, and A.~Frisch.
\newblock {CD}uce: an {XML}-friendly general purpose language.
\newblock In {\em ICFP '03, 8th ACM International Conference on Functional
  Programming}, pages 51--63, Uppsala, Sweden, 2003. ACM Press.

\bibitem{BeyerEGBEKOS11}
K.~S. Beyer, V.~Ercegovac, R.~Gemulla, A.~Balmin, M.~Y. Eltabakh, C.-C. Kanne,
  F.~{\"O}zcan, and E.~J. Shekita.
\newblock Jaql: A scripting language for large scale semistructured data
  analysis.
\newblock {\em PVLDB}, 4(12):1272--1283, 2011.

\bibitem{xquery}
S.~Boag, D.~Chamberlain, M.~F. Fern{\'a}ndez, D.~Florescu, J.~Robie, and
  J.~Sim{\'e}on.
\newblock {XQuery} 1.0: An {XML} query language, {W3C} recommendation, 2007.

\bibitem{DBLP:journals/sigmod/BunemanLSTW94}
P.~Buneman, L.~Libkin, D.~Suciu, V.~Tannen, and L.~Wong.
\newblock Comprehension syntax.
\newblock {\em SIGMOD Record}, 23(1):87--96, 1994.

\bibitem{buneman1981}
P.~Buneman, R.~Nikhil, and R.~Frankel.
\newblock {A Practical Functional Programming System for Databases}.
\newblock In {\em Proc. Conference on Functional Programming and Architecture}.
  ACM, 1981.
\newblock 1981.

\bibitem{Castagna2008}
G.~Castagna and K.~{Nguy$\tilde{\text{\^e}}$n}.
\newblock Typed iterators for {XML}.
\newblock In {\em Proceeding of the 13th ACM SIGPLAN international conference
  on Functional programming}, ICFP '08, pages 15--26, New York, NY, USA, 2008.
  ACM.

\bibitem{tata07}
H.~Comon, M.~Dauchet, R.~Gilleron, F.~Jacquemard, C.~L{\"o}ding, D.~Lugiez,
  S.~Tison, and M.~Tommasi.
\newblock Tree automata techniques and applications.
\newblock http://www.grappa.univ-lille3.fr/tata, 2007.

\bibitem{opal}
G.~Copeland and D.~Maier.
\newblock Making smalltalk a database system.
\newblock In {\em In Proceedings of the ACM SIGMOD International Conference on
  Management of Data}, pages 316--325, 1984.

\bibitem{JSONiq}
J.~R. (editor).
\newblock Jsoniq.
\newblock \url{http://jsoniq.org}.

\bibitem{Engelfriet75}
J.~Engelfriet.
\newblock {Bottom-up and Top-down Tree Transformations -- A comparison}.
\newblock {\em Theory of Computing Systems}, 9(2):198--231, 1975.

\bibitem{Engelfriet76}
J.~Engelfriet.
\newblock {Top-down tree transducers with regular look-ahead}.
\newblock {\em Mathematical Systems Theory}, 10(1):289--303, Dec. 1976.

\bibitem{alainthesis}
A.~Frisch.
\newblock {\em Th\'eorie, conception et r\'ealisation d'un langage de
  programmation fonctionnel adapt\'e \`a {XML}.}
\newblock PhD thesis, Universit\'e Paris 7 Denis Diderot, 2004.

\bibitem{jacm08}
A.~Frisch, G.~Castagna, and V.~Benzaken.
\newblock Semantic subtyping: Dealing set-theoretically with function, union,
  intersection, and negation types.
\newblock {\em Journal of the ACM}, 55(4):1--64, 2008.

\bibitem{haruobook}
H.~Hosoya.
\newblock {\em Foundations of {XML} Processing: The Tree Automata Approach}.
\newblock Cambridge University Press, 2010.

\bibitem{jaql}
Jaql.
\newblock \texttt{http://code.google.com/p/jaql}.

\bibitem{json}
Javascript object notation (json).
\newblock \texttt{http://json.org/}.

\bibitem{Martens03}
W.~Martens and F.~Neven.
\newblock Typechecking top-down uniform unranked tree transducers.
\newblock In {\em Proceedings of the 9th International Conference on Database
  Theory}, ICDT '03, pages 64--78, London, UK, 2002. Springer-Verlag.

\bibitem{Meijer11}
E.~Meijer.
\newblock The world according to {LINQ}.
\newblock {\em ACM Queue}, 9(8):60, 2011.

\bibitem{MeijerB11}
E.~Meijer and G.~Bierman.
\newblock A co-relational model of data for large shared data banks.
\newblock {\em Communications of the ACM}, 54(4):49--58, 2011.

\bibitem{kimthesis}
K.~Nguy$\tilde{\text{\^e}}$n.
\newblock {\em Language of {C}ombinators for {XML}: {C}onception, {T}yping,
  {I}mplementation}.
\newblock PhD thesis, Universit\'e Paris-Sud 11,
  \url{http://www.lri.fr/~kn/files/thesis.pdf}, 2008.

\bibitem{odata}
Odata.
\newblock \texttt{http://www.odata.org/}.

\bibitem{buneneman1988}
A.~Ohori and P.~Buneman.
\newblock {Type Inference in a Database Programming Language}.
\newblock In {\em ACM Conference on LISP and Functional Programming}, pages
  174--183, 1988.

\bibitem{machiavelli1989}
A.~Ohori, P.~Buneman, and V.~Tannen.
\newblock {Database Programming in Machiavelli -a Polymorphic Language with
  Static Type Inference}.
\newblock {\em Proc. ACM SIGMOD Conference}, pages 46--57, 1989.

\bibitem{DBLP:conf/icfp/OhoriU11}
A.~Ohori and K.~Ueno.
\newblock Making standard {ML} a practical database programming language.
\newblock In M.~M.~T. Chakravarty, Z.~Hu, and O.~Danvy, editors, {\em ICFP},
  pages 307--319. ACM, 2011.

\bibitem{OlstonRSKT08}
C.~Olston, B.~Reed, U.~Srivastava, R.~Kumar, and A.~Tomkins.
\newblock Pig latin: a not-so-foreign language for data processing.
\newblock In {\em SIGMOD Conference}, pages 1099--1110, 2008.

\bibitem{biginsights}
F.~\"{O}zcan, D.~Hoa, K.~S. Beyer, A.~Balmin, C.~J. Liu, and Y.~Li.
\newblock Emerging trends in the enterprise data analytics: connecting {Hadoop}
  and {DB2} warehouse.
\newblock In {\em Proceedings of the 2011 international conference on
  Management of data}, SIGMOD '11, pages 1161--1164, 2011.

\bibitem{SW96}
A.~Sabry and P.~Wadler.
\newblock A reflection on call-by-value.
\newblock In {\em Proceedings of the first ACM SIGPLAN international conference
  on Functional programming}, ICFP '96, pages 13--24. ACM, 1996.

\bibitem{pascalr}
J.~Schmidt and M.~Mall.
\newblock {Pascal/R Report}.
\newblock Technical Report~66, Fachbereich Informatik, universit\'e de Hamburg,
  1980.

\bibitem{Squeryl}
Squeryl: A {Scala} {ORM} and {DSL} for talking with {Databases} with minimum
  verbosity and maximum type safety\hfill.
\newblock \texttt{http://squeryl.org/}.

\bibitem{DBLP:conf/icdt/TannenBW92}
V.~Tannen, P.~Buneman, and L.~Wong.
\newblock Naturally embedded query languages.
\newblock In {\em ICDT}, pages 140--154, 1992.

\bibitem{WadTri89}
P.~Trinder and P.~Wadler.
\newblock Improving list comprehension database queries.
\newblock In {\em Fourth IEEE Region 10 Conference (TENCON)}, pages 186--192,
  Nov. 1989.

\bibitem{UnQL}
Unql.
\newblock \texttt{http://www.unqlspec.org/}.

\end{thebibliography}

% \clearpage
% \pagestyle{plain}
% \begin{multicols}{2}
% \vfill

% \begin{center}
% \Huge \bf
% END OF SUBMISSION\\
% EXTRA MATERIAL FOLLOWS
% \end{center}

% \vfill
% \end{multicols}
\clearpage

\begin{figure*}[h]
\begin{center}
\Huge Appendix\\

\smallskip

%\LARGE (not included in the submission)
\end{center}

\end{figure*}

\clearpage

\appendix
\section{Termination analysis algorithm}
\label{label:termination}
In order to deduce the result type of the application of a filter to
an expression, the type inference algorithm abstractly executes the
filter on the type of the expression.   
As explained in Section~\ref{label:typing}, the algorithm essentially
analyzes what may happen to the original input type (and to its subtrees) after
a recursive call and checks that, in all possible cases, every subsequent
recursive call will only be applied to subtrees of  the original
input type. In order to track the subtrees of the original input type, we use an
infinite set $\Id = \{ \idf_1, \idf_2, \ldots \}$ of \emph{variable
  identifiers} (ranged over by, possible indexed, \idf). These
identifiers are used to identify the subtrees of the original input type that are
bound to some variable after a recursive call. 
Consider for instance
the recursive filter:
$$\frec{X}{\falt{\fpat{\patpair{x_1}{\patpair{x_2}{x_3}}}{X\fprod{x_1}{x_3}}~~}}{~~\fpat{\_}{\NIL}}.$$ 
The
algorithm records that at the first call of this filter each variable
$x_i$ is bound to a subtree $\idf_i$ of the input type. The
recursive call $X\fprod{x_1}{x_3}$ is thus applied to the ``abstract argument''
$\fprod{\idf_1}{\idf_3}$. If we perform the substitutions for this
call, then we see that $x_1$ will be bound to $\idf_1$ and that $x_2$ and
$x_3$ will be bound to two subtrees of the $\idf_3$ subtree. We do not \emph{care} about $x_2$
since it is not used in any recursive call. What is important is that
both $x_1$ and $x_3$ are bound to subtrees of the original input type
(respectively, to $\idf_1$ and to a subtree of $\idf_3$) and therefore,
for this filter, the type inference algorithm will terminate for every
possible input type.

The previous example introduces the two important concepts that are
still missing for the definition of our analysis:
\begin{enumerate}
\item Recursive calls are applied to \emph{symbolic arguments}, such as
  $\fprod{\idf_1}{\idf_3}$, that are obtained from arguments by replacing variables
  by variable identifiers. Symbolic arguments are ranged
  over by $A$ and are formally defined as follows:
  \begin{displaymath}
   \begin{array}{lclr}
\textbf{Symb Args}
\quad  A & ::= & \idf & \textbf{(variable identifiers)}\\
       & |   &  c   & \textbf{(constants)}\\
       & |   &  \fprod{A}{A}   & \textbf{(pairs)}\\
       & |   &  \patrec{\ell\col A,\ldots,\ell\col A}\hspace*{-8mm}   & \textbf{(records)}\\
       & |   &  \bottom   & \textbf{(indeterminate)}\\
  \end{array}
  \end{displaymath}
 \item We said that we ``care'' for some variables and disregard
   others. When we analyze a recursive filter $\frec{X}f$ the
   variables we care about are those that occur in arguments of
   recursive calls of $X$. Given a filter $f$ and a filter recursion
   variable $X$ the set of these variables is formally defined as
   follows
   $$\Mark{X}{f} = \{ x \mid X a \sqsubseteq f\mbox{ and }x \in
   \vars{a}\}$$
   where $\sqsubseteq$ denotes the subtree containment relation and
   $\vars{e}$ is the set of expression variables that occur free in
   $e$. With an abuse of notation we will use $\vars{p}$ to denote
   the capture variables occurring in $p$ (thus,
   $\vars{fe}=\vars{f}\cup\vars{e}$ and
   $\vars{\fpat{p}{f}}=\vars{f}\setminus\vars{p}$, the rest of the
   definition being standard).
\end{enumerate}
As an aside notice that the fact that for a given filter $f$ the type inference
algorithm terminates on all possible input types \emph{does not}
imply that the execution of $f$ terminates on all possible input
values. For instance, our analysis correctly detects that for the filter
$\frec{X}{\falt{\fpat{\patpair{x_1}{x_2}}{X\fprod{x_1}{x_2}}}}{\fpat{\_}{\_}}$
type inference terminates on all possible input types (by returning
the very same input
type) although the application of this same filter never terminates on arguments of the
form $\fprod{v_1}{v_2}$.

We can now formally define the two phases of our analysis algorithm. 

\paragraph{First phase.} The first
phase is implemented by the function \Trees{\_}{\_}{\_}. For every
filter recursion variable $X$, this function explores a filter and does
the following two things:
\begin{enumerate}
\item It builds a substitution $\EE:\Vars\to\Id$ from expression
  variables to variables identifiers, thus associating each capture
  variable occurring in a pattern of the filter to a fresh identifier
  for the abstract subtree of the input type it will be bound to.
\item It uses the substitution $\EE$ to compute the set of symbolic
  arguments of recursive calls of $X$.
\end{enumerate}
In other words, if $n$ recursive calls of $X$ occur in the filter
$f$, then $\Trees{X}{\EE}{f}$ returns the set $\{A_1,...,A_n\}$ of
the $n$ symbolic arguments of these calls, obtained under the
hypothesis that the free variables of $f$ are associated to subtrees
as specified by $\EE$. The formal definition is as follows:
$$
\hspace*{-1mm}\begin{array}{l@{\;=~~}l}
\Trees{X}{\EE }{e} & \emptyset\\
\Trees{X}{\EE }{\fpat{p}{f}} &\Trees{X}{\EE \cup\hspace{-3mm} \bigcup \limits_{x_i \in \vars{p}}\hspace{-4mm}\{x_i \mapsto \idf_i\}}{f}\quad \mbox{($\idf_i$ fresh)}\\
\Trees{X}{\EE }{\fprod{f_1}{f_2}} & \Trees{X}{\EE}{f_1} \cup\Trees{X}{\EE }{f_2}\\
\Trees{X}{\EE }{\falt{f_1}{f_2}} &\Trees{X}{\EE }{f_1} \cup\Trees{X}{\EE }{f_2}\\
\Trees{X}{\EE }{\frec{X}{f}} & \emptyset\\
\Trees{X}{\EE }{\frec{Y}{f}} &\Trees{X}{\EE }{f}\hfill (X\not=Y)\\
\Trees{X}{\EE }{X a} & \{a \EE\} \\
\Trees{X}{\EE }{\fseq{f_1}{f_2}} &\Trees{X}{\EE }{f_2}\\
\Trees{X}{\EE }{of} &\Trees{X}{\EE }{f}\hfill (o=\GROUPBY{},\ORDERBY{})\\
\multicolumn{2}{l}{\Trees{X}{\EE }{\forec{\ell_i{\col}f_i}_{i\in
I}}= \bigcup \limits_{i \in I}\Trees{X}{\EE}{f_i}}\\
\end{array}
$$
where $a \EE$ denotes the application of the substitution $\EE$ to $a$.

The definition above is mostly straightforward. The two important
cases are those for the pattern filter where the substitution $\EE$ is
updated by associating every capture variable of the pattern with a
fresh variable identifier, and the one for the recursive call where
the symbolic argument is computed by applying $\EE$ to the actual
argument.

\paragraph{Second phase.}
The second phase is implemented by the function $\chk{}$. The
intuition is that $\chk{f}$ must ``compute'' the application of $f$ to
all the symbolic arguments collected in the first phase and then
check whether the variables occurring in the arguments of recursive
calls (\emph{ie}, the ``marked'' variables) are actually bound to subtrees of
the original type (\emph{ie}, they are bound either to variable identifiers
or to subparts of variable identifiers). If we did not have filter composition, then this
would be a relatively easy task: it would amount to compute
substitutions (by matching symbolic arguments against patterns), apply
them, and finally verify that the marked variables satisfy the sought
property. Unfortunately, in the case of a composition filter
$\fseq{f_1}{f_2}$ the analysis is more complicated than that. Imagine that we want
to check the property for the application of  $\fseq{f_1}{f_2}$ to a
symbolic argument $A$. Then to check the property for the recursive
calls occurring in $f_2$ we must compute (or at least, approximate)
the set of the symbolic arguments that
will be produced by the application of $f_1$ to $A$ and that will be
thus fed into $f_2$ to compute the composition.  Therefore $\chk{f}$
will be a function that  rather than returning just true or false, it will
return a set of symbolic arguments that are the result of the
execution of $f$, or it will fail if any recursive call does not
satisfy the property for marked variables. 

\begin{figure*}[t!]
%\frame{
\begin{itemize}
\item $\chk[{\cal V}]{\EE ,f,\{A_1,...,A_n\}} = \bigcup \limits_{i=1}^n \chk[{\cal V}]{\EE ,f,A_i}$

\item \textbf{If } $(\patand{\atype{A}}{\accept{f}} = \emptyset)$ \textbf{then} $\chk[{\cal V}]{\EE ,f,A}= \emptyset$ \textbf{otherwise:}\\

%\begin{displaymath}
$
\begin{array}{rclr}
  \chk[{\cal V}]{\EE ,e,A} &=& 
      \begin{cases}
         \{a \EE \} & \text{if } e\equiv a\\
         \bottom  &\text{otherwise}\\
       \end{cases}\\
\\
  \chk[{\cal V}]{\EE ,\fseq{f_1}{f_2},A} &=& \chk[{\cal V}]{\EE ,f_2, \chk[\cal V]{\EE ,f_1,A}}\\
\\
  \chk[{\cal V}]{\EE ,X a,A} &=& \bottom\\
\\
  \chk[{\cal V}]{\EE ,\frec{X}{f},A} &=&
  \begin{cases}
    \textit{fail} &\text{if } $\chk[\Mark{X}{f}]{\EE ,f,A} = $\textit{fail}\\
    \bottom       &\text{otherwise}
  \end{cases}\\
\\
  \chk[{\cal V}]{\EE ,\falt{f_1}{f_2},A} &=&
  \begin{cases}
    \chk[{\cal V}]{\EE,f_2,A} &\text{if } \patand{\accept{f_1}}{\atype{A}}=\emptyset\\
    \chk[{\cal V}]{\EE,f_1,A} &\text{if } \atype{A} \setminus \accept{f_1}=\emptyset\\
    \chk[{\cal V}]{\EE,f_1,A} \cup \chk[{\cal V}]{\EE,f_2,A}& \text{otherwise }\\
  \end{cases}\\
\\
  \chk[{\cal V}]{\EE ,\fprod{f_1}{f_2},A} &=&
  \begin{cases}
    \chk[{\cal V}]{\EE ,f_1,A_1}\,{\times}\, \chk[{\cal V}]{\EE ,f_2,A_2} &\mbox{if } A \equiv (A_1,A_2)\\
    \chk[{\cal V}]{\EE ,f_1,\idf_1\,\,} \times \chk[{\cal V}]{\EE ,f_2,\idf_2} &\mbox{if } A \equiv \idf \qquad\text{\small(where $\idf_1$ and $\idf_2$ are fresh)}\\
    \textit{fail} & \text{if }A\equiv\bottom\\
  \end{cases}\\
\\
  \chk[{\cal V}]{\EE ,\forec{\ell_i{\col}f_i}_{1\leq i \leq n},A} &=&
  \begin{cases}
    \bigcup\limits_{B_i\in\chk[{\cal V}]{\EE ,f_i,A_i}} \hspace{-5mm}\forec{\ell_1{\col}B_1,...,\ell_n{\col}B_n,...,\ell_{n+k}{\col}A_{n+k}} &\mbox{if } A \equiv \forec{\ell_i{\col}f_i}_{1\leq i \leq n+k}\\
    \bigcup\limits_{B_i\in\chk[{\cal V}]{\EE ,f_i,\idf_i}} \hspace{-5mm}\forec{\ell_1{\col}B_1,...,\ell_n{\col}B_n,...,\ell_{n+k}{\col}A_{n+k}}&\mbox{if } A \equiv \idf \qquad\text{\small(where all $\idf_i$ are fresh)}\\
    \textit{fail} & \text{if }A\equiv\bottom\\
  \end{cases}\\
\\  \chk[{\cal V}]{\EE ,\fpat{p}{f},A} &=& 
       \begin{cases}
         \textit{fail} & \text{if }\matchv{A}{p}=\textit{fail}\\
         \chk[{\cal V}]{\EE  \cup (\matchv{A}{p}), f, A} &\text{otherwise}
       \end{cases}
     \end{array}
$
%\end{displaymath}
\end{itemize}
\caption{Definition of the second phase of the analysis}
\label{fig:chk}
\end{figure*}

More precisely, the function $\chk{}$ will have the
form $$\chk[\cal V]{\EE,f,\{A_1,...,A_n\}}$$ where $\cal
V\subseteq\Vars$ stores the set of marked variables, $f$ is a filter,
$\EE$ is a substitution for (at least) the free variables of $f$ into $\Id$, and $A_i$ are symbolic
arguments. It will either fail if the marked variables do not satisfy
the property of being bound to (subcomponents of) variable identifiers
or return an over-approximation of  the
result of applying $f$ to all the $A_i$ under the hypothesis
$\EE$. This approximation is either a set of new symbolic arguments, or
\bottom. The latter simply indicates that \chk{} is not able to
compute the result of the application, typically because it is the
result of some expression belonging to the host language
(\emph{eg}, the application of a function) or because it is the result
of a recursive filter or of a recursive call of some filter.  

The full definition of $\chk[\_]{\_,\_,\_}$ is given in
Figure~\ref{fig:chk}. 
Let us comment the different cases in detail.
$$\chk[{\cal V}]{\EE ,f,\{A_1,...,A_n\}} = \bigcup \limits_{i=1}^n \chk[{\cal V}]{\EE ,f,A_i}$$
simply states that to compute the filter $f$ on a set of symbolic
arguments we have to compute $f$ on each argument and return the union
of the results. Of course, if any $\chk[{\cal V}]{\EE ,f,A_i}$ fails,
so does $\chk[{\cal V}]{\EE ,f,\{A_1,...,A_n\}}$. The next case is
straightforward: if we know that an argument of a given form makes the
filter $f$ fail, then we do not perform any check (the
recursive call of $f$ will never be called) and no result will be
returned. For instance, if we apply the filter $\forec{\ell\col f}$ to
the symbolic argument $\fprod{\idf_1}{\idf_2}$, then this will always
fail so it is useless to continue checking this case. Formally, what we do is to
consider $\atype{A}$, the set of all values that have the form of
$A$ (this is quite simply obtained from $A$ by replacing  $\tany$ for every
occurrence of a variable identifier or of \bottom) and check whether in it
there is any value accepted by $f$, namely, whether the intersection
$\patand{\atype{A}}{\accept{f}}$ is empty. If it is so, we directly
return the empty set of symbolic arguments, otherwise we have to
perform a fine grained analysis according to the form of the filter.

If the filter is an expression, then there are two possible cases:
either the expression has the form of an argument ($\equiv$ denotes
syntactic equivalence), in which case we return it after having applied 
the substitution $\EE$ to it; or it is some other expression, in which case
the function says that it is not able to compute a result and returns
\bottom.

\begin{figure*}[t!]
\begin{displaymath}
\begin{array}{rcl}
\matchvv{A}{x} &=&
\begin{cases}
  \{ x \mapsto A \} &
     \hspace*{-2mm}%
     \begin{array}{l}
           \text{if } A\equiv v\text{ or }A\equiv\idf\text{ or }x \notin {\cal V}
     \end{array}
  \\[1mm]
  \textit{fail} &\text{otherwise}
\end{cases}\\
\\
\matchvv{A}{\patpair{p_1}{p_2}} &=&
  \begin{cases}
    (\matchvv{A_1}{p_1}) \cup (\matchvv{A_2}{p_2})&\text{if } A\equiv\patpair{A_1}{A_2}\\
    (\matchvv{\idf_1}{ p_1}) \cup (\matchvv{\idf_2 }{ p_2})& 
                 \text{if } A\equiv\idf\qquad \text{\small(where $\idf_1$, $\idf_2$ are fresh)}\\
    \textit{fail} & \text{if }A\equiv\bottom\\
    \Omega  &\text{otherwise}
  \end{cases}\\
  \\
\matchvv{A}{\patand{p_1}{p_2}} &=& \matchvv{A}{p_1} \cup \matchvv{A}{p_2}\\
\\
\matchvv{A}{\pator{p_1}{p_2}} &=&
    \begin{cases} 
       \matchvv{A}{p_1} &\text{if }\matchvv{A}{p_1}\not=\Omega\\
       \matchvv{A}{p_2} &\text{otherwise}
    \end{cases}\\
\\
\matchvv{A}{\patrec{\ell_1\col p_1,...,\ell_n\col p_n}} &=&
  \begin{cases}
      \bigcup \limits_{i\leq i\leq n} \matchvv{A_i}{p_i}&
                  \text{if } A\equiv\patrec{\ell_1\col A_1,...,\ell_n\col A_n} \\
      \bigcup \limits_{i\leq i\leq n} \matchvv{\idf_i}{p_i}&
                  \text{if } A\equiv\idf\qquad \text{\small(where $\idf_i$ are fresh)}\\
    \textit{fail} & \text{if }A\equiv\bottom\\
    \Omega  &\text{otherwise}
  \end{cases}\\
\\
\matchvv{A}{\patorec{\ell_1\col p_1,...,\ell_n\col p_n}} &=&
  \begin{cases}
      \bigcup \limits_{i\leq i\leq n} \matchvv{A_i}{p_i}&
                  \text{if } A\equiv\patrec{\ell_1\col A_1,...,\ell_n\col A_n,...,\ell_{n+k}\col A_{n+k}} \\
      \bigcup \limits_{i\leq i\leq n} \matchvv{\idf_i}{p_i}&
                  \text{if } A\equiv\idf\qquad \text{\small(where $\idf_i$ are fresh)}\\
    \textit{fail} & \text{if }A\equiv\bottom\\
    \Omega  &\text{otherwise}
  \end{cases}\\
\\
\matchvv{A}{t} &=& \emptyset\\
%\\
%\matchvv{A}{\patcon{x}{c}} &=& \emptyset
\end{array}
\end{displaymath}
\caption{Pattern matching with respect to a set ${\cal V}$ of marked variables}\label{fig:match}
\end{figure*}

When the filter is a composition of two filters, $\fseq{f_1}{f_2}$, we first call
$\chk[\cal V]{\EE ,f_1,A}$ to analyze $f_1$ and compute the result of
applying $f_1$ to $A$, and then we feed this result to the analysis of $f_2$:
 $$\chk[{\cal V}]{\EE ,\fseq{f_1}{f_2},A} = \chk[{\cal V}]{\EE ,f_2,
   \chk[\cal V]{\EE ,f_1,A}}$$
When the filter is recursive or a recursive call, then we are not able
to compute the symbolic result. However, in the case of a recursive filter
we must check whether the type inference  terminates on it. So
we mark the variables of its recursive calls, check whether its
definition passes our static analysis, and return \bottom\ only if this
analysis ---\emph{ie}, $\chk[\Mark{X}{f}]{\EE ,f,A}$--- did not fail.
 Notice that \bottom\ is quite different from failure since it
allows us to compute the approximation as long as the filter after the
composition does not bind (parts of) the result that are \bottom\ to
variables that occur in arguments of recursive calls. A key example is
the case of the filter \texttt{Filter[$F$]} defined in
Section~\ref{builtin}. If the parameter filter $F$ is recursive, then
without \bottom\ \texttt{Filter[$F$]} would be rejected by the static
analysis. Precisely, this is due to the composition
$\fseq{Fx}{(\falt{\fpat{\false}{Xl}}{\cdots})}$. When $F$ is recursive
the analysis supposes that $Fx$ produces \bottom, which is then passed
over to the pattern. The recursive call $Xl$ is executed in the
environment $\bottom/\true$, which is empty. Therefore the result of
the $Fx$ call, whatever it is, cannot affect the termination of the subsequent recursive
calls. Even though the composition uses the result of $Fx$, the
form of this result cannot be such as to make typechecking of \texttt{Filter[$F$]}
diverge.

For the union filter $\falt{f_1}{f_2}$ there are three possible cases:
$f_1$ will always fail on $A$, so we return the result of $f_2$; or
$f_2$ will never be executed on $A$ since $f_1$ will never fail on
$A$, so we return just the result of $f_1$; or we cannot tell
which one of $f_1$ or $f_2$ will be executed, and so we return the union
of the two results.

If the filter is a product filter $\fprod{f_1}{f_2}$, then its accepted
type $\accept{\fprod{f_1}{f_2}}$ is a subtype of $\fprod{\tany}{\tany}$. Since we
are in the case where  $\patand{\atype{A}}{\accept{\fprod{f_1}{f_2}}}$
is not empty, then $A$ can have just two forms: either it is a pair or
it is a variable identifier $\idf$. In the former case  we check the
filters component-wise and return as result the set-theoretic product
of the results. In the latter case, recall that $\idf$ represents a
subtree of the original input type. So if the application does not
fail it means that each subfilter will be applied to a subtree of
$\idf$. We introduce two fresh variable identifiers to denote these
subtrees of $\idf$ (which, of course are subtrees of the original
input type, too) and, as in the previous case, apply the check
component-wise and return the set-theoretic product of the results. The
case for the record filter is similar to the case of the product filter.

Finally, the case of pattern filter is where the algorithm checks that the marked
variables are indeed bound to subtrees of the initial input type. What
$\chk{}$ does is to match the symbolic argument against the pattern and
update the substitution $\EE$ so as to check the subfilter $f$ in an
environment with the corresponding assignments for the capture
variables in $p$. Notice however that the pattern matching
\matchv{A}{p} receives as extra argument the set $\cal V$ of marked
variables, so that while computing the substitutions for the capture variables
in $p$ it also checks that all capture variables that are also marked are
actually bound to a subtree of the initial input type. \matchv{A}{p}
is defined in Figure~\ref{fig:match} (to enhance readability we omitted the $\cal V$
index since it matters only in the first case and it does not change
along the definition). 
The check for marked variables is performed in the first case of the definition: when a symbolic
argument $A$ is matched against a variable $x$, if the variable is not
marked ($x\not\in\cal V$), then the substitution $\{ x \mapsto A \}$ is
returned; if it is marked, then the matching (and thus the analysis)
does not fail only if the symbolic argument is either a value (so it
does not depend on the input type) or it is a variable identifier (so
it will be exactly a subtree of the original input type).
The other cases of the definition of \matchv{A}{p} are standard. Just
notice that in the product and record patterns when $A$ is a variable
identifier $\idf$, the algorithm creates fresh
new variable identifiers to denote the new subtrees in which $\idf$
will be decomposed by the pattern. Finally, the reader should not
confound $\Omega$ and \textit{fail}. The former indicates that the argument
does not match the value, while the latter indicates that a marked variable
is not bound to a value or to exactly a subtree of the initial input
type (notice that the case $\Omega$ cannot happen in the definition of
$\chk{}$ since \matchv{A}{p} is called only when $\atype{A}$ is
contained in $\accept p$).

\paragraph{The two phases together.} Finally we put the two phases
together. Given a recursive filter $\frec{X}{f}$ we run the analysis
by marking the variables in the recursive calls of $X$ in $f$, running
the first phase and then and feeding the result to the second phase:
$\chk[\Mark{X}{f}]{\emptyset,f,\Trees{X}{\emptyset}{f}}$. If this
function does not fail then we say that
$\frec{X}{f}$ passed the check:
$$\chk{\frec{X}{f}}\iff \chk[\Mark{X}{f}]{\emptyset,f,\Trees{X}{\emptyset}{f}}\not=\textit{fail}$$
For filters that are not recursive at toplevel it suffices to add a
dummy recursion:  $\chk{f}\eqdef\chk{\frec{X}{f}}$ with $X$ fresh.

\begin{theorem}[Termination]\label{th:halt}
If  $\chk{f}$ then the type inference algorithm terminates for $f$ on every input type $t$.
\end{theorem}
In order to prove the theorem above we need some auxiliary definitions:

\begin{definition}[Plinth~\cite{alainthesis}]
\label{def:plinth}
  A plinth $\beth\subset\Types$ is a set of types with the
  following properties:
\begin{itemize}
\item $\beth$ is finite
\item $\beth$ contains \tany, \tempty{} and is closed under Boolean
  connectives ($\patand{}{}$,$\pator{}{}$,$\synneg$)
\item for all types $t =\patpair{t_1}{t_2}$ in $\beth$,
  $t_1\in\beth$ and $t_2\in\beth$
\item for all types $t = \patrec{\ell_1\col t_1,\ldots,\ell_n\col t_n,(\pmb{..})}$ in $\beth$, $t_i\in\beth$ for all $i\in[1..n]$.
\end{itemize}
We define the plinth of $t$, noted $\beth(t)$, as the smallest plinth containing $t$.
\end{definition}
\noindent
Intuitively, the plinth of $t$ is the set of types that can be obtained by all possible boolean combination of the subtrees of $t$. Notice that  $\beth(t)$ is always defined since our types are regular: they have finitely many distinct subtrees, which (modulo type equivalence) can thus be combined in finitely many different ways. 
\begin{definition}[Extended plinth and support]
\label{def:explinth}
Let $t$ be a type and $f$ a filter. The \emph{extended plinth} of $t$ and $f$, noted $\hat\beth(f,t)$ is defined as $\beth(t\vee\bigvee_{v \sqsubseteq f}v)$ (where $v$ ranges over values).

The \emph{support} of $t$ and $f$, noted as \supp{f}{t}, is defined as   
$$\supp{f}{t}\eqdef\;\hat\beth(f,t)\;\;\cup\!\!\!\!\!\bigcup_{A\in\Trees{\emptyset}{\emptyset}{f}}\!\!\!\!\!\{A\sigma\mid \sigma:\Id(A)\to\hat\beth(f,t)\} $$
\end{definition}
The extended plinth of $t$ and $f$ is the  set of types that can be obtained by all possible boolean combination of the subtrees of $t$ and of values that occur in the filter $f$. The intuition underlying the definition of support of $t$ and $f$ is that it includes all the possible types of arguments of recursive calls occurring in $f$, when $f$ is applied to an argument of type $t$.

Lastly, let us prove the following technical lemma:
\begin{lemma}
\label{lem:chk}
Let $f$ be a filter such that \chk{f} holds. Let $t$ be a type. For
every derivation $D$ (finite or infinite) of\\
\centerline{$\ftypej{\Gamma\sep \Delta\sep M}{f}{t}{s}$},
for every occurrence of the rule {\sc [Fil-Call-New]}
\begin{center}
  \AxiomC{$\ftypej{\Gamma'\sep\Delta'\sep M',((X,t'')\mapsto T)}{\Delta'(X)}{t''}{t'}$}
  \RightLabel{\scriptsize$%\left(
    \begin{array}{c}
      t''=\tyof{\Gamma'}{a}\\
      (X,t'')\not\in\dom(M')\\
      T \text{~fresh}
    \end{array}
    % \right)
    $}
  \UnaryInfC{$\ftypej{\Gamma'\sep\Delta'\sep M'}{(X~a)}{s}{\frec{T}{t'}}$}
  \DisplayProof
\end{center}
in $D$, for every $x\in\vars{a}$, $\Gamma'(x)\in\hat\beth(f,t)$ (or
equivalently, $\tyof{\Gamma'}{a}\in\supp{f}{t}$)
\end{lemma}
\begin{proof}
By contradiction. Suppose that there exists an instance of the rule
for which $x\in\vars{a}\land x\notin \hat\beth(f,t)$. This means that
$\Gamma'(x)$ is neither a singleton type occurring in $f$ nor a type
obtained from $t$ by applying left projection, right projection or
label selection ($\star$).
 But since $\chk{f}$ holds, $x$ must be
bound to either an identifier or a value $v$ (see
Figure~\ref{fig:match}) during the computation of
$\chk[\Mark{X}{f}]{\emptyset,f,\Trees{X}{\emptyset}{f}}$. Since
identifiers in $\chk{f}$ are only
introduced for the input parameter or when performing a left projection,
right projection or label selection of another identifier, this ensure
that $x$ is never bound to the result of an expression whose type is
not in $\hat\beth(f,t)$, which contradicts $(\star)$.
\end{proof}

We are now able to prove Theorem~\ref{th:halt}. Let $t$ be a type and $f$ a filter define
$$\DD{f}{t} = \{(X,t')\mid X\sqsubseteq f, t'\in\supp{f}{t}\}$$
Notice that since $\supp{f}{t}$ is finite and there are finitely many different recursion variables occuring in a filter, then $\DD{f}{t}$ is finite, too.  
Now let $f$ and $t$ be the filter and type mentioned in the statement of
Theorem~\ref{th:halt} and consider the (possibly infinite) derivation
of $\ftypej{\Gamma\sep\Delta\sep M}{f}{t}{s}$. Assign to every
judgment $\ftypej{\Gamma'\sep\Delta'\sep M'}{f'}{t'}{s'}$ the
following weight
$$\textit{Wgt}(\ftypej{\Gamma'\sep\Delta'\sep M'}{f'}{t'}{s'})=\fprod{|\DD{f}{t}\setminus\dom(M')|~}{~\textit{size}(f')}$$
where $|S|$ denotes the cardinality of the set $S$ (notice that in the definition $S$ is finite), $\dom(M')$ is the
domain of $M'$, that is, the set of pairs $(X,t)$ for which $M'$ is
defined, and $\textit{size}(f')$ is the depth of the syntax tree of $f'$. 

Notice that the set of all weights lexicographically ordered form a
well-founded order. It is then not difficult to prove that every
application of a rule in the derivation of
$\ftypej{\Gamma\sep\Delta\sep M}{f}{t}{s}$ strictly decreases
\textit{Wgt}, and therefore that this derivation must be finite.
This can be proved by case analysis on the applied rule, and we must distinguish only three cases:
\begin{description}
\item[{\sc[Fil-Fix]}] Notice that in this case the first component of \textit{Wgt} either decreases or remains constant, and the second component strictly decreases.
\item[{\sc[Fil-Call-New]}] In this case Lemma~\ref{lem:chk} ensures that the first component of the \textit{Wgt} of the premise strictly decreases.  Since this is the core of the proof let us expand the rule. Somewhere in the derivation of $\ftypej{\Gamma\sep\Delta\sep M}{f}{t}{s}$ we have the following rule: \\[2mm]
\hspace*{-8mm}\begin{tabular}{l}
{\small\sc [Fil-Call-New]}\\
  \AxiomC{$\ftypej{\Gamma'\sep\Delta'\sep M',((X,t'')\mapsto T)}{\Delta'(X)}{t''}{t'}$}
  \RightLabel{\scriptsize$%\left(
                 \begin{array}{c}
                  t''=\tyof{\Gamma'}{a}\\
                  (X,t'')\not\in\dom(M')\\
                  T \text{~fresh}
                \end{array}
                % \right)
                $}
  \UnaryInfC{$\ftypej{\Gamma'\sep\Delta'\sep M'}{(Xa)}{s}{\frec{T}{t'}}$}
  \DisplayProof
\end{tabular}\\[2mm]
and we want to prove that $(\DD{f}{t}\setminus\dom(M'))\supsetneq(\DD{f}{t}\setminus(\dom(M')\cup(X,t''))$ and the containment must be strict to ensure that the measure decreases. First of all notice that $(X,t'')\not\in\dom(M')$, since it is a side condition for the application of the rule. So in order to prove that containment is strict it suffices to prove that $(X,t'')\in\DD{t}{f}$. But this is a consequence of Lemma~\ref{lem:chk} which ensures that $t''\in\supp{f}{t}$, whence the result.

\item[{\sc[Fil-*]}] With all other rules the first component of \textit{Wgt} remains constant, and the second component strictly decreases.
\end{description}

\begin{figure*}[t!]
\begin{alltt}
  type M = (K,V) | (V,V,K)
  type V = \{ var : string \} | \{ lambda2 : (string, string, M) \}
  type K = \{ var : string \} | \{ lambda1 : (string, M) \}

  let filter Eval = 
    |  ( \{ lambda2 : (x, k, m) \}, v , h ) -> m ; Subst[x,v] ; Subst[k,h]; Eval
    |  ( \{ lambda1 : (x, m) \}, v ) -> m ; Subst[x,v] ; Eval
    |  x -> x

  let filter Subst[\(c\),F] = 
    |  ( Subst[\(c\),F] ,  Subst[\(c\),F] ,  Subst[\(c\),F] ) 
    |  ( Subst[\(c\),F] ,  Subst[\(c\),F] ) 
    |  \{ var : \(c\) \}  ->  F
    |  \{ lambda2 : (x&\({\neg}c\), k&\({\neg}c\), m) \}  ->  \{ lambda2 : (x, k, m;Subst[\(c\),F]) \}
    |  \{ lambda1 : (x&\({\neg}c\), m) \}  ->  \{ lambda1 : (x, m;Subst[\(c\),F]) \}
    |  x  ->  x
\end{alltt}
\caption{Filter encoding of $\lambda_\text{cps}$}\label{fig:tc}
\end{figure*}

\subsection{Improvements}

Although the analysis performed by our algorithm is already fine grained,
it can be further refined in two simple ways. As explained in
Section~\ref{label:typing}, the algorithm checks that after \emph{one
step} of reduction the capture variables occurring in recursive calls
are bound to subtrees of the initial input type. So a possible
improvement consists to try to check this property for a higher number
of steps, too. For instance consider the filter:
\begin{tabbing}
\qquad$\frec{X}{}$\=$\fpat{\patpair{\NIL}{x}}{X(\patrec{\ell\col x})}$\\
  \qquad\quad\pmb{|}\>$\fpat{\patpair{y}{\_}}{X(\patpair{\NIL}{\patpair{y}{y}})}$\\
  \qquad\quad\pmb{|}\>$\fpat{\_}{\NIL}$
\end{tabbing}
This filter does not pass our analysis since if $\idf_y$ is the
identifier bound to the capture variable $y$, then when unfolding
the recursive call in the second branch the $x$ in the
first recursive call will be bound to $\patpair{\idf_y}{\idf_y}$. But
if we had pursued our abstract execution one step further  we would have
seen that $\patpair{\idf_y}{\idf_y}$ is not used in a recursive call
and, thus, that type inference terminates. Therefore, a first
improvements is
to modify $\chk{}$ so that it does not stop just after one step of
reduction but tries
to go down $n$ steps, with $n$ determined by some heuristics
based on sizes of the filter and of the input type.

A second and, in our opinion, more promising improvement is to enhance
the precision of the test  $\patand{\atype{A}}{\accept{f}}
= \emptyset$ in the definition of \textit{Check}, that verifies whether
the filter $f$ fails on the given symbolic argument $A$. In the current
definition the only information we collect about the type of symbolic
arguments is their structure. But further type
information, currently unexploited, is provided by patterns. For instance, in the following
(admittedly stupid) filter
\begin{tabbing}
\qquad$\frec{X}{}$\=$\fpat{\patpair{\Int}{x}}{x}$\\
  \qquad\quad\pmb{|}\>$\fpat{\patand{y}{\Int}}{X({\patpair{y}{y}})}$\\
  \qquad\quad\pmb{|}\>$\fpat{z}{z}$
\end{tabbing}
if in the first pass we associate $y$ with $\idf_y$, then we know that
$\patpair{\idf_y}{\idf_y}$ has type $\patpair{\Int}{\Int}$. If we
record this information, then we know that in the second pass
$\patpair{\idf_y}{\idf_y}$ will always match the first pattern, and so it will
never be the argument of a further recursive call. In other words, there is at
most one nested recursive call. The solution is conceptually simple (but
yields a cumbersome formalization, which is why we chose the current
formulation) and amounts to modify $\Trees{}{}$ so that when it
introduces fresh variables it records their type information with
them. It then just suffices to modify the definition of $\atype{A}$ so
as it is obtained from $A$ by
replacing every 
occurrence of a variable identifier by its type information (rather
than by $\tany$) and the current definition of $\chk$ will do the
rest.

\section{Proof of Turing completeness (Theorem~\ref{th:tc})}
\label{label:tc}
In order to prove Turing completeness we show how to define by filters
an evaluator for untyped (call-by-value) $\lambda$-calculus. If we
allowed recursive calls to occur on the left-hand side of composition, then the encoding
would be straightforward: just implement $\beta$ and context
reductions. The goal however is to show that the restriction on
compositions does not affect expressiveness. To that end we have to
avoid context reductions, since they require recursive calls before
composition. To do so, we first translate $\lambda$-terms via
Plotkin's call-by-value CPS translation and apply Steele and Rabbit's 
administrative
reductions to them obtaining terms in $\lambda_\text{cps}$. The latter
is isomorphic to cbv $\lambda$-calculus (see for instance~\cite{SW96})
and defined as follows.
$$\begin{array}{rcl}
M & ::= & KV \mid V V K\\
V & ::= & x \mid \lambda x.\lambda k. M\\
K & ::= & k \mid \lambda x . M
\end{array}$$
with the following reduction rules (performed at top-level, without any reduction under context).
$$\begin{array}{rcl}
 (\lambda x.\lambda k. M)VK & \longrightarrow & M[x:=V][k:=K]\\
 (\lambda x.M)V& \longrightarrow & M[x:=V]
\end{array}
$$
In order to prove Turing completeness it suffices to show how to
encode $\lambda_\text{cps}$ terms and reductions in our calculus. For
the sake of readability, we use mutually recursive types (rather
than their encoding in $\mu$-types),  we use records (though pairs would
have sufficed), we write just the recursion variable $X$ for the filter {\tt
  x -> $X$x}, and use $\neg t$ to denote the type $\ANY\syndiff
t$. Term productions are encoded by the recursive types given at the beginning of Figure~\ref{fig:tc}.

Next we define the evaluation filter \texttt{Eval}. In its body it calls the filter \texttt{Subst[$c$,F]} which implements the capture free substitution and, when $c$ denotes a constant, is defined as right below.\footnote{We were a little bit sloppy in the notation and used a filter parameter as a pattern. Strictly speaking this is not allowed in filters and, for instance, the branch \texttt{ \{ var : \(c\) \} -> F} in \texttt{Subst} should rather be written as \texttt{ \{ var : y \} -> ( (y=$c$); (true -> F | false ->  \{ var : y \}) )}. Similarly, for the other two cases. }

It is straightforward to see that the definitions in Figure~\ref{fig:tc} implement the
reduction semantics of CPS terms. Of course the definition above would
not pass our termination condition. Indeed while \texttt{Subst} would
be accepted by our algorithm, \texttt{Eval} would fail since it is not
possible to ensure that the recursive calls of \texttt{Eval} will
receive from \texttt{Subst} subtrees of the original input type. This
is expected: while substitution always terminate they may return trees that are not subtrees of the original term.

\section{Proof of subject reduction (Theorem \ref{thm:subred})}
\label{prf:subred}
We first give the proof of Lemma~\ref{lem:omega}, which we restate:

Let $f$ be a filter and $v$ be a value such that $v\notin\ftype{f}$. 
Then for every $\gamma$, $\delta$,
if $\japp{\delta\textbf{;}\gamma}{f}{v}{r}$, $r\equiv\Omega$.

\begin{proof}
  By induction on the derivation of
  $\japp{\delta\textbf{;}\gamma}{f}{v}{r}$, and case analysis on the
  last rule of the derivation:
  \begin{description}
    \item[(\textbf{expr}):] here, $\ftype{e} = \ANY$ for any
      expression $e$, therefore this rule cannot be applied (since
      $\forall v, v\in\ANY$).
    \item[\textbf{(prod)}:] assume $v\notin
      \ftype{\fprod{f_1}{f_2}}\equiv\patpair{\ftype{f_1}}{\ftype{f_2}}$.
      then either $v\notin \patpair{\ANY}{\ANY}$, in which case only
      rule \textbf{(error)} applies and therefore $r\equiv\Omega$.
      Or $v\equiv(v_1,v_2)$ and $v_1\notin\ftype{f_1}$. Then by
      induction hypothesis on the first premise, 
      $\japp{\delta\textbf{;}\gamma}{f_1}{v_1}{r_1}$ and
      $r_1\equiv\Omega$, which contradicts the side condition
      $r_1\neq\Omega$ (and similarly for the second
      premise). Therefore this rule cannot be applied to
      evaluate $\fprod{f_1}{f_2}$.
    \item[(\textbf{patt}):] Similarly to the previous case. If $v\notin
      \ftype{\fpat{p}{f}}=\patand{\ftype p}{\ftype f}$ then either $v\notin\ftype{p}$ (which
      contradicts $v/p\neq\Omega$) or  $v\notin\ftype{f}$ and by induction hypotheis,
      $\japp{\delta\textbf{;}\gamma, v/p}{f}{v}{\Omega}$.
    \item[(\textbf{comp}):] If $v\notin
      \ftype{\fseq{f_1}{f_2}}\equiv\ftype{f_1}$, then by induction
      hypothesis on the first premise, $r_1\equiv\Omega$, which
      contradicts the side condition $r_1\neq\Omega$. Therefore only
      rule \textbf{(error)} can be applied here.
    \item[(recd):] Similar to product type: either $v$ is not a record
      and therefore only the rule \textbf{(errro)} can be applied, or
      one of the $r_i=\Omega$ (by induction hypothesis) which
      contradicts the side condition.
    \item[\textbf{(union1)} and \textbf{(union2)}:] since
      $v\notin\ftype{\falt{f_1}{f_2}}$, this means that
      $v\notin\ftype{f_1}$ \emph{and}
      $v\notin\ftype{f_2}$. Therefore if rule \textbf{(union1)} is
      chosen, by induction hypothesis, $r_1\equiv\Omega$ which
      contradicts the side condition. If rule \textbf{(union2)} is
      chosen, then by induction hypothesis $r_2\equiv\Omega$ which
      gives $r\equiv\Omega$.
    \item[\textbf{(rec-call)}] trivially true, since it cannot be that
    $v\notin\tany$
    \item[\textbf{(rec)}] we can apply straightfowardly the induction
      hypothesis on the premise and have that $r\equiv\Omega$.
    \item[\textbf{(groupby)} and \textbf{(orderby)}:] If $v\notin{f}$
      then the only rule that applies is \textbf{(error)} which gives
      $r\equiv\Omega$.
  \end{description}
\end{proof}
We are now equipped to prove the subject reduction theorem which we
restate in a more general manner:

For every $\Gamma$, $\Delta$,
$M$, $\gamma$, $\delta$ such that $\forall x\in \dom(\gamma),
x\in\dom(\Gamma)
\land \gamma(x):\Gamma(x)$, if
$\ftypej{\Gamma\sep\Delta\sep M}{f}{t}{s}$, then for all $v:t$,
$\japp{\delta\textbf{;}\gamma}{f}{v}{r}$ implies $r:s$.

The proof is by induction on the derivation of
$\japp{\delta\textbf{;}\gamma}{f}{v}{r}$, and by case analysis on the
rule. Beside the basic case, the only rules which are not a
straightforward application of induction are the rules \textbf{union1} and
\textbf{union2} that must be proved simulatneously. Other cases being
proved by a direct application of the induction hypothesis, we only detail the
case of the product constructor.
\begin{description}
  \item[\textbf{(expr)}:] we suppose that the host languages enjoys
    subject reduction, hence $e(v) = r : s$.
  \item[\textbf{(prod)}:] Here, we know that
    $t\equiv\bigvee_{i\leq\texttt{rank}(t)}(t_1^i,t_2^i)$. Since $v$ is a
    value, and $v:t$, we obtain that $v\equiv (v_1, v_2)$. Since
    $(v_1, v_2)$ is a value, $\exists i\leq \texttt{rank}(t)$ such
    that $(v_1, v_2):(t_1^i,t_2^i)$. The observation that $v$ is a
    value is crucial here, since in general given a type $t'\leq t$
    with $t'\equiv\bigvee_{j\leq\texttt{rank}(t')}({t'_1}^j,{t'_2}^j)$
    it is not true that $\exists i,
    j~\textrm{s.t.~}({t'_1}^j,{t'_2}^j)\leq({t_1}^i,{t_2}^i)$. However
    this property holds for singleton types and therefore for values.
    We have therefore that $v_1:t_1^i$ and $v_2:t_1^i$. Furthermore,
    since we suppose that a typing derivation exists and the typing
    rules are syntax directed, then 
    $\ftypej{\Gamma\sep\Delta\sep M}{f}{t_1^i}{s_1^i}$ must be a
    used to prove our typing judgement. We can apply the induction
    hypothesis and deduce that
    $\japp{\delta\textbf{;}\gamma}{f}{v_1}{r_1}$ and similarly for
    $v_2$. We have therefore that the filter evaluates to $(r_1, r_2)$
    which has type $(s_1^i, s_2^i)\leq
    \bigvee_{j\leq\texttt{rank}(s)}({s_1}^i,{s_2}^i)$ which proves
    this case.
\item[\textbf{(union1)} and \textbf{(union2)}:] both rules must be proved
  together. Indeed, given a filter $\falt{f_1}{f_2}$ for which we have
  a typing derivation for the judgement
  $\ftypej{\Gamma\sep\Delta\sep M}{\falt{f_1}{f_2}}{t}{s}$, either
  $v:\patand{t}{\ftype{f_1}}$ and we can apply the induction
  hypothesis and therefore, $\ftypej{\Gamma\sep\Delta\sep
    M}{\falt{f_1}}{\patand{t}{\ftype{f_1}}}{s_1}$ and if
  $\japp{\delta\textbf{;}\gamma}{f}{v_1}{r_1}$ (case \textbf{union1})
  $r_1:s_1$. However it might be that $v\notin
  \patand{t}{\ftype{f_1}}$. Then by Lemma~\ref{lem:omega} we have that
  $\japp{\delta\textbf{;}\gamma}{f}{v_1}{\Omega}$ (case
  \textbf{union2}) and we can apply the induction hypothesis on the
  second premise, which gives us
  $\japp{\delta\textbf{;}\gamma}{f}{v_1}{r_2}$ which allows us to
  conclude.
\item[\textbf(error):] this rule can never occur. Indeed, if $v:t$ and
  $\ftypej{\Gamma\sep\Delta\sep M}{f}{t}{s}$, that means in particular
  that $t\leq \ftype{f}$ and therefore that all of the side conditions
  for the rules \textbf{prod}, \textbf{pat} and \textbf{recd} hold.
  
\end{description}
\section{Complexity of the typing algorithm}
\label{sec:complex}
\begin{proof}  For clarity we restrict
  ourselves to types without record constructors but the proof can
  straihgtforwardly be extended to them. First remark that our
  types are isomorphic to alternating tree automata (ATA) with
  intersection, union and complement (such as those defined
  in~\cite{tata07, haruobook}). From such an ATA $t$ it is possible to
  compute a non-deterministic tree automaton $t'$ of size
  $O(2^{|t|})$. When seeing $t'$ in our type formalism, it is a type
  generated by the following grammar:
\begin{displaymath}
\begin{array}{lcl}
  \tau & ::= & \mu X. \tau~\mid~\pator{\tau}{\textit{const}} ~\mid~\textit{const}\\
  \textit{const} & ::= & \patpair{\tau}{\tau} ~\mid~\textit{atom}\\
  \textit{atom} & ::= & X ~\mid~ b
\end{array}
\end{displaymath}
where $b$ ranges over negation and intersections of basic
types. Intuitively each recursion variable $X$ is a state in the NTA,
a finite union of products is a set of non-deterministict transitions whose
right-hand-side are each of the products and \textit{atom} productions
are leaf-states. Then it is clear than if \chk{f} holds, the algorithm
considers at most $|f|\times|t'|$ distinct cases (thanks to the
memoization set $M$). Furthermore for each rule, we may test a
subtyping problem (\emph{eg}, $t\leq\ftype{f}$), which itself is EXPTIME
thus giving the bound
\end{proof}

\section{Precise typing of \GROUPBY{}}
\label{groupby}

The process of inferring a precise type for $\GROUPBY{f}(t)$ is decomposed in several steps.
First we have to compute the set $\DIS{f}$ of discriminant domains for
the filter $f$. The idea is that these domains are the types on which
we know that $f$ may give results of different types. Typically this corresponds to all possible branchings that the filter $f$ can do. For instance, if  $\{t_1,...,t_n\}$ are pairwise disjoint types and $f$ is the filter  $\fpat{t_1}{e_1}\pmb{|}\cdots\pmb{|}\fpat{t_n}{e_n}$, then the set of discriminant domains
for the filter $f$ is $\{t_1,...,t_n\}$, since the various $s_i$ obtained from $\ftypej{\emptyset\sep \emptyset\sep
          \emptyset}{f}{t_i}{s_i}$ may all be different, and we want to keep track of the relation between a result type $s_i$ and the  input type $t_i$ that produced it. Formally $\DIS{f}$ is defined as follows:
$$
\hspace*{-1mm}\begin{array}{l@{\;=~~}l}
\DIS{e} & \{\tany\}\\
\DIS{\fprod{f_1}{f_2}} & \DIS{f_1} \times\DIS{f_2}\\
\DIS{\fpat{p}{f}} &\{\patand{\accept{p}}{t}\mid t\in\DIS{f}\}\\
\DIS{\falt{f_1}{f_2}} &\DIS{f_1} \cup\{t\syndiff{\accept{f_1}}\mid t\in\DIS{f_2}\}\\
\DIS{\frec{X}{f}} & \DIS{f}\\
\DIS{X a} & \{\tany\} \\
\DIS{\fseq{f_1}{f_2}} &\DIS{f_1}\\
\DIS{of} &\{\LIST{\tany\ks}\}\qquad (o=\GROUPBY{},\ORDERBY{})\\
%\multicolumn{2}{l}{
\DIS{\forec{\ell_i{\col}f_i}_{i\in
I}}&{\bigcup \limits_{t_i\in\DIS{f_i}}\forec{\ell_i{\col}t_i}_{i \in I}}\\
\end{array}
$$

Now in order to have a high precision in typing we want to compute the
type of the intermediate list of the groupby on a set of disjoint
types. For that we define a normal form of a set of types that given a
set of types computes a new set formed of pairwise disjoint types whose
union covers the union of all the original types.

$$\NF{\{t_i\mid i\in S\}} = \bigcup_{\emptyset\subset I \subseteq S}(\bigcap_{i\in I}t_i\setminus\bigcup_{j\in S\setminus I}t_j)$$

Recall that we are trying to deduce a type for  $\GROUPBY{f}(t)$, that is for the expression $\GROUPBY{f}$ when applied to an argument of type $t$. Which types shall we use as input for our filter $f$ to compute the
type of the intermediate result? Surely we want to have all the types
that are in the $\ITEM{t}$. This however is not enough since we
would not use the information of the descriminant domains for the
filter. For instance if the filter gives two different result types
for positive and negative numbers and the input is a list of integers,
we want to compute two different result types for positive and
negatives and not just to compute the result of the filter application
on generic integers (indeed  $\ITEM{\LIST{\Int\ks}} = \{\Int\}$). So the idea is to add to the set that must be
normalized all the types of the set of discriminant types. This
however is not precise enough, since these domains may be much larger
than the item types of the input list. For this reason we just take
the part of the domain types that intersects at least some possible
value of the list. In other terms we consider the normal form of the
following set
$$\ITEM{t}\cup\{\sigma_i\synwedge\!\!\bigvee_{t_i\in\DIS{f}}\!\!\!t_i\mid \sigma_i \in \ITEM{t}\}$$

The idea is then that if $\{t_i\}_{i\in I}$ is the normalized set
of the set above, then the type of grouby is the type
$\LIST{\bigvee_{i\in I}\patpair{f(t_i)}{\LIST{t_i\kss}}\ks}$ with the
optimization that we can replace the $\ks$ by a $\kss$ if we know that the
input list is not empty.

\section{Comparison with top-down tree transducers}
\label{sec:comptttrl}

We first show that every Top-down tree transducers with regular
look-ahead can be encoded by a filter. We use the definition of
\cite{Engelfriet76} for top-down tree transducers:
\begin{definition}[Top-down Tree transducer with regular look-ahead]
A \emph{top-down tree transducer with regular look-ahead} (TDTTR) is a
5-tuple $<\Sigma, \Delta, Q, Q_d, R>$ where $\Sigma$ is the input
alphabet, $\Delta$ is the output alphabet, $Q$ is a set of states,
$Q_d$ the set of initial states and $R$ a set of rules of the form:
\begin{displaymath}
q(a(x_1,\ldots, x_n)), D \rightarrow b(q_1(x_{i_1}), \ldots, q_m(x_{i_m}))
\end{displaymath}
where $a\in \Sigma_n$, $b\in \Delta_m$, $q$, $q_1$,\ldots, $q_m\in Q$,
$\forall j \in 1..m, i_j \in 1..n$ and $D$ is a mapping from
$\{x_1,\ldots, x_n\}$ to $2^{T_\Sigma}$ ($T_\Sigma$ being the set of
regular tree languages over alphabet $\Sigma$). A TDTTR is
\emph{deterministic} if $Q_d$ is a singleton and the left-hand sides
of the rules in $R$ are pairwise disjoint.
\end{definition}

Since our filters encode programs, it only make sense to compare
filters and \emph{deterministic} TDTTR.
We first show that given such a \emph{deterministic} TDTTR we can
write an equivalent filter $f$ and, furthermore, that $\chk{f}$ does
not fail. First we recall the encoding of ranked labeled trees into
the filter data-model:
\begin{definition}[Tree encoding]
Let $t\in T_\Sigma$. We define the tree-encoding of $t$ and we write
$\inter{t}$ the value defined inductively by:
\begin{displaymath}
\begin{array}{lclr}
\inter{ a } & = & \texttt{(`a,[])} & \forall a\in \Sigma_0\\
\inter{ a(t_1, \ldots, t_n) } & = & \texttt{(`a,[ $\inter{t_1}$ \ldots{} $\inter{t_n}$ ])} & \forall a \in \Sigma_n\\
\end{array}
\end{displaymath}
where the list notation \texttt{[ $v_1$ \ldots{} $v_n$ ]} is a
short-hand for \texttt{($v_1$, \ldots, ($v_n$, `nil))}. We generalize
this encoding to tree languages and types. Let $S\subseteq T_\sigma$,
we write $\inter{S}$ the set of values such that $\forall t\in S,
\inter{t}\in \inter{S}$.
\end{definition}
\noindent In particular it is clear than when $S$ is regular, $\inter{S}$ is a
type.
\begin{lemma}[TDTTR $\rightarrow$ filters]
Let $T = <\Sigma, \Delta, Q, \{ q_0 \} , R>$ be a deterministic
TDTTR. There exists a filter $f_T$ such that:
\begin{displaymath}
  \forall t\in \dom(T), \inter{T(t)} \equiv f_T \inter{t}
\end{displaymath}
\end{lemma}
\begin{proof}
The encoding is as follows. For every state $q_i\in Q$, we will introduce a
recursion variable $X_i$. Formally, the translation is performed by
the function $\textit{TR}: 2^{Q} \times Q \rightarrow \Filters$ defined as:
\begin{displaymath}
\begin{array}{lclr}
\textit{TR}(S, q_i) & = & \fpat{x}{X_i ~x} & \textrm{if } q_i \in S\\
\textit{TR}(S, q_i) & = & \mu X_i. (\falt{f_1}{\falt{\ldots}{f_n}}) & \textrm{if } q_i \notin S\\
\end{array}
\end{displaymath}
where every rule $r_j \in R$
\begin{displaymath}
r_j\equiv q_i(a_j(x_1, \ldots, x_n)), D_j
\rightarrow b_j(q_{j_1}(x_{j_{k_1}}), \ldots, q_{j_m}(x_{j_{k_m}}))
\end{displaymath}
is translated into:
\begin{displaymath}
\begin{array}{l}
\fpat{\texttt{(`a$_j$, [ $\patand{x_1}{\inter{D_j(x_1)}}$ \ldots{}
    $\patand{x_n}{\inter{D_j(x_n)}}$ ])}}{\texttt{`b$_j$,[]}} \\[1mm]
\hfill \textrm{if } b_j\in\Delta_0\\[2mm]
\fpat{\texttt{(`a$_j$, [ $\patand{x_1}{\inter{D_j(x_1)}}$ \ldots{}
    $\patand{x_n}{\inter{D_j(x_n)}}$ ])}}{}\\
\hfill {\fprod{\texttt{`b}_j}{\fprod{\fseq{x_{j_{k_1}}}{\textit{TR}(S',q_{j_1})}}{
      \fprod{\ldots}{\fprod{\fseq{x_{j_{k_m}}}{\textit{TR}(S',q_{j_m})}}{\NIL}}
}}}\\[1mm]
\textrm{where} S' = S\cup\{q_i\} \hfill  \textrm{otherwise}\\
\end{array}
\end{displaymath}
The fact that $  \forall t\in \dom(T), \inter{T(t)} \equiv
\textit{TR}(\varnothing,q_0) \inter{t}$ is proved by a straightforward
induction on $t$. The only important point to remark is that since $T$
is deterministic, there is exactly one branch of the alternate filter
$\falt{f_1}{\falt{\ldots}{f_n}}$ that can be selected by pattern
matching for a given input $v$.
For as to why $\chk{f_T}$ holds, it is sufficient to remark that each
recursive call is made on a strict subtree of the input which
guarantees that $\chk{f_T}$ returns true.
\end{proof}

\begin{lemma}[TDTTR $\not\leftarrow$ filters]
Filters are strictly more expressive than TDTTRs.
\end{lemma}
\begin{proof}
Even if we restrict filters to have the same domain as TDTTRs
(meaning, we used a fixed input alphabet $\Sigma$ for atomic types) we
can define a filter that cannot be expressed by a TDTTR. For instance
consider the filter:
\begin{displaymath}
 \begin{array}{r@{}r@{}l}
    \fpat{y}{} &(\mu X. &  \fpat{\patpair{\texttt{`a}}{\NIL}}{y}\\
    & \falt{}{} & \fpat{\patpair{\texttt{`b}}{\NIL}}{y}\\
    & \falt{}{} & \fpat{\patpair{\texttt{`a}}{\texttt{[$x$]}}}{(\texttt{`a},\texttt{[ $X\,x$ ]})}) \\
    & \falt{}{}     &\fpat{\patpair{\texttt{`b}}{\texttt{[$x$]}}}{(\texttt{`b},\texttt{[ $X\,x$ ]})}) \\
  \end{array}
\end{displaymath}
This filter works on monadic trees of the form $u_1(\ldots u_{n-1}(u_n)\ldots)$
where $u_i\in\{a, b\}$, and essentially replaces the leaf $u_n$ of a
tree $t$ by a copy of $t$ itself. This cannot be done by a
TDTTR. Indeed, TDTTR have only two ways to ``remember'' a subtree
and copy it. One is of course by using variables; but their scope
is restricted to a rule and therefore an arbitrary subtree can only be
copied at a \emph{fixed distance} of its original position.
For instance in a rule of the form
$q(a(x)), D \rightarrow a(q_1(x), b(b(q_1(x))) )$, assuming than $q_1$
copies its input, the copy of the original subtree is two-levels down
from its next sibling but it cannot be arbitrary far.
A second way to
copy a subtree is to remember it using the states. Indeed, states can
encode the knowledge that the TDTTR has accumulated along a
path. However, since the number of states is finite, the only thing
that a TDTTR can do is copy a fixed path. For instance for any given
$n$, there exists a TDTTR that performs the transformation above, for
trees of height  $n$ (it has essentially $2^n-1$ states which
remember every possible path taken). For instance for $n=2$, the TDTTR
is:
\begin{displaymath}
\begin{array}{lcl}
q_0(a(x)), D & \rightarrow & a(q_1(x))\\
q_1(a()), D & \rightarrow & a(a)\\
q_1(b()), D & \rightarrow & a(b)\\
q_0(b(x)), D & \rightarrow & b(q_2(x))\\
q_2(a()), D & \rightarrow & b(a)\\
q_2(b()), D & \rightarrow & b(b)\\
\end{array}
\end{displaymath}
where $\Sigma = \Delta = \{a,b\}$, $Q=\{q0, q1, q2\}$, $Q_d=\{q_0\}$
and $D = \{ x \mapsto T_\Sigma\}$. It is however impossible to write a
TDTTR that replaces a leaf with a copy of the whole input, for inputs of
arbitrary size. A similar example is used in \cite{Engelfriet75} to
show that TDTTR and bottom-up tree transducers are not comparable.
\end{proof}

\ifLONGVERSION\else
\section{JSON, XML, Regex}
\label{label:xml}

\fi
\section{Operators on record types.}
\label{tsum}

We use the theory of records defined for \cduce. We summarize here the
main definitions. These are adapted from those given in Chapter 9 of
Alain Frisch's PhD thesis~\cite{alainthesis} where the interested reader
can also find detailed definitions of the semantic interpretation of record types and of the subtyping relation it induces, the modifications that must be done to the algorithm to decide it, finer details on pattern matching definition and
compilation techniques for record types and expressions.

Let $Z$ denote some set, a function $r:\mathcal{L}\to Z$ is
\emph{quasi-constant} if there exists $z\in Z$ such that the set
$\{\ell{\in}\mathcal{L}\mid r(\ell)\not=z\}$ is finite; in this case
we denote this set by $\textsf{dom}(r)$ and the element $z$ by
$\textsf{def}(r)$. We use $\mathcal{L}\rightarrowtriangle Z$ to denote
the set of quasi-constant functions from $\mathcal{L}$ to $Z$ and the
notation $\{\ell_1=z_1, \ldots,\ell_n=z_n,\_=z\}$ to denote the
quasi-constant function $r:\mathcal{L}\rightarrowtriangle Z$ defined
by $r(\ell_i)=z_i$ for $i=1..n$ and $r(\ell)=z$ for
$\ell\in\mathcal{L}\setminus\{\ell_1,\ldots,\ell_n\}$. Although this
notation is not univocal (unless we require $z_i\not=z$), this is largely sufficient for the
purposes of this section.

Let $\bot$ be a
distinguished constant, then the sets
$\texttt{string}\rightarrowtriangle\textbf{Types}\cup\{\bot\}$ and
$\texttt{string}\rightarrowtriangle\textbf{Values}\cup\{\bot\}$
denote the set of all record types expressions and of all
record values, respectively. The constant $\bot$ represents the value of the fields
of a record that are ``undefined''. To ease the presentation we use
the same notation both for a constant and the singleton type that
contains it: so when $\bot$ occurs in
$\mathcal{L}\rightarrowtriangle\textbf{Values}\cup\{\bot\}$ it denotes a value,
while in $\texttt{string}\rightarrowtriangle\textbf{Types}\cup\{\bot\}$ it denotes
the singleton type that contains only the value $\bot$.

Given the definitions above, it is clear that the record types in Definition~\ref{def:types} are nothing but specific notations for some quasi-constant functions in  $\texttt{string}\rightarrowtriangle\textbf{Types}\cup\{\bot\}$. More precisely, the open record type expression $\patorec{\ell_1\col t_1,\ldots,\ell_n\col t_n}$ denotes the quasi-constant function $\{\ell_1= t_1,\ldots,\ell_n=t_n,\_=\tany\}$ while the closed record type  expression $\patrec{\ell_1\col t_1,\ldots,\ell_n\col t_n}$ denotes the quasi-constant function $\{\ell_1= t_1,\ldots,\ell_n=t_n,\_=\bot\}$. Similarly, the optional field notation
$\patrec{..., \ell\texttt{?:}t,...}$ denotes the record type expressions in which $\ell$ is mapped either to $\bot$ or to the type $t$, that is, $\{..., \ell = \pator{t\,}{\bot},...\}$. 

Let $t$ be a type and $r_1, r_2$ two record type expressions, that is $r_1, r_2:\texttt{string}\rightarrowtriangle\textbf{Types}\cup\{\bot\}$. The \emph{merge} of  $r_1,$ and $r_2$ with respect to $t$, noted $\oplus_t$ and used infix, is the record type expression defined as follows:
\begin{equation*}
(r_1\oplus_t r_2)(\ell) \eqdef \left\{
    \begin{array}{ll}
       r_1(\ell) &\text{if } \patand{r_1(\ell)}{t}\leq\tempty\\
       \pator{(r_1(\ell)\setminus t)}{r_2(\ell)} & \text{otherwise}  
     \end{array}\right.
\end{equation*}
Recall that by Lemma~\ref{le:rcddec} a \emph{record type} (\emph{ie}, a subtype of $\{..\}$) is equivalent to a finite union of \emph{record type expressions} (\emph{ie}, quasi-constant functions in  $\texttt{string}\rightarrowtriangle\textbf{Types}\cup\{\bot\}$). So the definition of \emph{merge} can be easily extended to all record types as follows
$$
(\bigvee_{i\in I}r_i)\oplus_t(\bigvee_{j\in J}r'_j)\eqdef\bigvee_{i\in I,j\in J}(r_i\oplus_tr'_j)
$$
Finally, all the operators we used for the typing of records in the rules of Section~\ref{se:tyrcd} are defined in terms of the merge operator:
\begin{eqnarray}
t_1 + t_2 & \eqdef & t_2 \oplus_\bot t_1\\
t\setminus \ell & \eqdef& \{\ell =\bot , \_ = c_0\} \oplus_{c_0} t
\end{eqnarray}
where $c_0$ is any constant different from $\bot$ (the semantics of the operator does not depend on the choice of $c_0$ as long as it is different from $\bot$).

Notice in particular that the result of the concatenation of two record type expressions $r_1+r_2$ may result for each field $\ell$ in three different outcomes: 
\begin{enumerate}
\item if $r_2(\ell)$ does not contain $\bot$ (\emph{ie}, the field $\ell$ is surely defined), then we take the corresponding field of $r_2$: $(r_1+r_2)(\ell) = r_2(\ell)$
\item if $r_2(\ell)$ is undefined (\emph{ie}, $r_2(\ell)=\bot$), then we take the corresponding field of $r_1$: $(r_1+r_2)(\ell) = r_1(\ell)$
\item  if $r_2(\ell)$ \emph{may} be undefined (\emph{ie},  $r_2(\ell)=\pator{t}{\bot}$ for some type $t$), then we take the union of the two corresponding fields since it can results either in $r_1(\ell)$ or  $r_2(\ell)$ according to whether the record typed by $r_2$ is undefined in $\ell$ or not:  $(r_1+r_2)(\ell) = \pator{r_1(\ell)}{(r_2(\ell)\setminus\bot)}$.
\end{enumerate}
This explains all the examples we gave in the main text. In particular, $\patrec{a\col\Int,b\col\Int}+\patrec{a\texttt{?}\col\Bool}= \patrec{a\col\pator{\Int}{\Bool},b\col\Int}$ since ``$a$'' \emph{ma}y be undefined in the right hand-side record while ``$b$'' \emph{is} undefined in it, and $\patrec{a\col\Int}+\patrec{\textbf{..}}= \patrec{\textbf{..}}$, since ``$a$'' in the right hand-side record  is defined (with $a\mapsto\tany$) and therefore has priority over the corresponding definition in the left hand-side record.

\section{Encoding of Co-grouping}
\label{sec:cogroupingencoding}

As shown in Section~\ref{sec:jaqltr}, our $\GROUPBY{}$ operator can
encode JaQL's \texttt{group each} $x=e$ \texttt{as} $y$ \texttt{by}
$e'$, where $e$ computes the grouping key, and for each distinct key,
$e'$ is evaluated in the environment where $x$ is bound to the key and
$y$ is bound to the sequence of elements in the input having that key.
Co-grouping is expressed by:
\begin{alltt}\ibm
group 
   \(l\sb{1}\) by \(x\) = \(e\sb{1}\) as \(y\sb{1}\)
         \(\vdots\)
   \(l\sb{n}\) by \(x\) = \(e\sb{n}\) as \(y\sb{n}\)
into \(e\)
\end{alltt}
Co-grouping is encoded by the following composition of
filters:
\begin{alltt}
[ \(\SemParen{l\sb{1}}\) \ldots  \(\SemParen{l\sb{n}}\) ];
[ Transform[\(\fpat{x}{(1,x)}\)] \ldots Transform[\(\fpat{x}{(n,x)}\)] ];
Expand;
\(\GROUPBY{(\falt{\fpat{\patpair{1}{\$}}{\SemParen{e\sb{1}}}}{\falt{\ldots}{\fpat{\patpair{n}{\$}}{\SemParen{e\sb{n}}}}})}\);
Transform[\(\;\fpat{\patpair{x}{l}}{}\)([\((\fseq{l}{\texttt{Rgrp1}})\)\ldots\((\fseq{l}{\texttt{ Rgrp\(n\)}})\)];
                  \(\fpat{\LIST{y\sb{1},\ldots,y\sb{n}}}{\SemParen{e}}\) )]
\end{alltt}
where
\begin{alltt}
let filter Rgrp\(i\) = \NIL => \NIL
                 | ((\(i\),x),tail) => (x , Rgrp\(i\) tail) 
                 | _ => Rgrp\(i\) tail
\end{alltt}
Essentially, the co-grouping encoding takes as argument a sequence of
sequences of values (the $n$ sequences to co-group). Each
 of these $n$ sequence, is tagged with an integer $i$. Then we flatten
 this sequences of tagged values. We can on this single sequence apply
 our $\GROUPBY{}$ operator and modify each key selector $e_i$ so that
 it is only applied to a value tagged with integer $i$. Once this is
 done, we obtain a sequence of pairs $(k,l)$ where $k$ is the commen
 grouping key and $l$ a sequence of tagged values. We only have to
 apply the auxiliary \texttt{Regrp} filter which extracts the $n$
 subsequences from $l$ (tagged with 1..$n$) and removes the integer
 tag for each value. Lastly we can call the recombining expression $e$
 which has in scope $x$ bound to the current grouping key and
 $y_1$,\ldots,$y_n$ bound to each of the sequences of values from
 input $l_1$, \ldots,$l_n$ whose grouping key is $x$.

\end{document}